\documentclass[a4paper,11pt,UKenglish]{article}

\usepackage{bm,mathptmx,booktabs}
\usepackage{tcolorbox}
\usepackage[usestackEOL]{stackengine}

\usepackage{tikz}
\usetikzlibrary{decorations.pathreplacing}

\usepackage{ifthen}
\usepackage{comment}
\makeatletter
\usepackage[sort,comma,square,numbers]{natbib}
\renewcommand*{\NAT@spacechar}{~} %
\renewcommand\bibsection %
{
  \section*{\refname
    \@mkboth{\MakeUppercase{\refname}}{\MakeUppercase{\refname}}}
}

\usepackage{color,tikz,environ}
\usetikzlibrary{decorations.markings}

\usepackage{etoolbox}

\usepackage{titlesec}

\titleformat{\paragraph}
{\normalfont\normalsize\bfseries}{\theparagraph}{1em}{}
\titlespacing*{\paragraph}
{0pt}{3.25ex plus 1ex minus .2ex}{1.5ex plus .2ex}


\newcommand{\boundellipse}[3]%
{(#1) ellipse (#2 and #3)
}

\usepackage[margin=1in]{geometry}

\ifthenelse{\isundefined{\lipics}}
{
\usepackage{graphicx}
\usepackage[colorlinks=true,bookmarks=false]{hyperref}
\usepackage[small,bf]{caption}
\usepackage{subfigure}
\usepackage[british]{babel}
}
{}
\usepackage{xcolor}
\definecolor{darkblue}{rgb}{0,0,0.45}
\definecolor{darkred}{rgb}{0.6,0,0}
\definecolor{darkgreen}{rgb}{0.13,0.5,0}
\hypersetup{colorlinks, linkcolor=darkblue, citecolor=darkgreen,
urlcolor=darkblue}

\ifthenelse{\isundefined{\llncs}}{
  \usepackage{amsthm}
  \usepackage{authblk}
  
}

\usepackage{amsmath,amsfonts,amssymb}
\usepackage{mathtools}
\usepackage{mathrsfs}
\usepackage{wrapfig}
\usepackage{xspace}
\usepackage[shortlabels]{enumitem}
\setlist[enumerate]{nosep} %
\setlist[itemize]{nosep} %

\usepackage{multirow}

\usepackage[protrusion=true]{microtype}

\usepackage{aliascnt}
\ifthenelse{\isundefined{\llncs}}{
  \theoremstyle{plain}
}{}
\newtheorem{theorem}{Theorem}[section]
\newaliascnt{lemma}{theorem}
\newaliascnt{corollary}{theorem}
\newaliascnt{definition}{theorem}
\newaliascnt{claim}{theorem}
\newaliascnt{proposition}{theorem}
\newaliascnt{remark}{theorem}
\newaliascnt{hypothesis}{theorem}
\newaliascnt{observation}{theorem}
\newtheorem{lemma}[lemma]{Lemma}
\newtheorem{claim}[claim]{Claim}

\newtheorem{corollary}[corollary]{Corollary}

\newtheorem{hypothesis}[hypothesis]{Hypothesis}
\newtheorem{observation}[observation]{Observation}
\ifthenelse{\isundefined{\llncs}}{
  \theoremstyle{definition}
}{}
\newtheorem{definition}[definition]{Definition}
\aliascntresetthe{lemma}
\aliascntresetthe{remark}
\aliascntresetthe{corollary}
\aliascntresetthe{proposition}
\aliascntresetthe{definition}
\aliascntresetthe{claim}
\aliascntresetthe{hypothesis}
\aliascntresetthe{observation}

\newcommand{\mc}[1]{{\mathcal{#1}}}

\newcommand{\eps}{{\varepsilon}}

\DeclareMathOperator{\cost}{cost}
\DeclareMathOperator*{\val}{val}

\newcommand{\hy}{\hbox{-}\nobreak\hskip0pt}

\newcommand{\ignore}[1]{}
\usepackage{environ,xstring}
\newif\iflabel
\newif\ifdbs
\newif\ifamp
\NewEnviron{doitall}{%
  \noexpandarg
  \expandafter\IfSubStr\expandafter{\BODY}{\label}{\labeltrue}{\labelfalse}%
  \expandafter\IfSubStr\expandafter{\BODY}{\\}{\dbstrue}{\dbsfalse}%
  \expandafter\IfSubStr\expandafter{\BODY}{&}{\amptrue}{\ampfalse}%
  \iflabel\def\doitallstar{}\else\def\doitallstar{*}\fi
  \ifdbs
    \ifamp
      \def\doitallname{align}%
    \else
      \def\doitallname{multline}%
    \fi
  \else
    \def\doitallname{equation}
  \fi
  \begingroup\edef\x{\endgroup
    \noexpand\begin{\doitallname\doitallstar}%
    \noexpand\BODY
    \noexpand\end{\doitallname\doitallstar}%
  }\x
}
\def\[#1\]{\begin{doitall}#1\end{doitall}}

\newcommand{\pname}[1]{\textsc{#1}}

\newcommand{\newreptheorem}[2]{\newtheorem*{rep@#1}{\rep@title}\newenvironment{rep#1}[1]{\def\rep@title{\bf #2 \ref*{##1}}\begin{rep@#1}}{\end{rep@#1}}}

\makeatother

\newreptheorem{theorem}{Theorem}
\newreptheorem{lemma}{Lemma}
\newreptheorem{corollary}{Corollary}

\newtheorem*{rep@thm}{\rep@title} \newcommand{\newrepthm}[2]{%
\newenvironment{rep#1}[1]{%
\def\rep@title{\autoref{##1}}%
\begin{rep@thm} }%
{\end{rep@thm} } }
\makeatother
\newrepthm{thm}{}
\newrepthm{lem}{}
\newrepthm{crl}{}

\usepackage{framed}

\usepackage{color,tikz,environ}
\usetikzlibrary{decorations.markings,arrows,plotmarks}
\tikzset{nomorepostaction/.code={\let\tikz@postactions\pgfutil@empty}}

\tikzset{middlearrow/.style={
        decoration={markings,
            mark= at position 0.5 with {\arrow{#1}} ,
        },
        postaction={decorate}
    }
}

\tikzset{onethirdarrow/.style={
        decoration={markings,
            mark= at position 0.33 with {\arrow{#1}} ,
        },
        postaction={decorate}
    }
}

\tikzset{twothirdarrow/.style={
        decoration={markings,
            mark= at position 0.67 with {\arrow{#1}} ,
        },
        postaction={decorate}
    }
}

\tikzset{endarrow/.style={
        decoration={markings,
            mark= at position 0.9 with {\arrow{#1}} ,
        },
        postaction={decorate}
    }
}

\tikzset{startarrow/.style={
        decoration={markings,
            mark= at position 0.1 with {\arrow{#1}} ,
        },
        postaction={decorate}
    }
}

\newcommand{\polyn}{\cdot n^{O(1)}}
\newcommand{\dsn}{\textsc{DSN}\xspace}

\newcommand{\dsnP}{\textsc{DSN$_\textsc{Planar}$}\xspace}
\newcommand{\bidsn}{\textsc{bi\hy{}DSN}\xspace}

\newcommand{\bidsnP}{\textsc{bi\hy{}DSN$_\textsc{Planar}$}\xspace}

\newcommand{\scss}{\textsc{SCSS}\xspace}

\newcommand{\scssP}{\textsc{SCSS$_\textsc{Planar}$}\xspace}

\newcommand{\csi}{\textsc{Colored Subgraph Isomorphism}\xspace}
\newcommand{\mcsi}{\textsc{Maximum Colored Subgraph Isomorphism}\xspace}

\newcommand{\gt}{\textsc{Grid Tiling}\xspace}
\newcommand{\altgt}{\mdseries{\gt}\xspace}

\newcommand{\uni}{\text{unique}}

\newcommand{\altscssP}{\mdseries{\scssP}\xspace}

\newcommand{\dirmcfour}{\textsc{Directed Multicut With $4$ Pairs}\xspace}
\newcommand{\dirmctwo}{\textsc{Directed Multicut With $2$ Pairs}\xspace}
\newcommand{\dirmc}{\textsc{Directed Multicut}\xspace}
\newcommand{\dirmwc}{\textsc{Directed Multiway Cut}\xspace}
\newcommand{\SCSS}{\textsc{Strongly Connected Steiner Subgraph}\xspace}
\newcommand{\DSN}{\textsc{Directed Steiner Network}\xspace}

\newcommand{\mcsik}{MCSI$(K_{\ell,\ell})$\xspace}

\newcommand{\horizontal}{\textsc{Horizontal}\xspace}
\newcommand{\verticall}{\textsc{Vertical}\xspace}

\newcommand{\good}{\textsc{Good}\xspace}
\newcommand{\goodpairs}{\textsc{Good-Pairs}\xspace}

\newcommand{\opt}{\text{OPT}\xspace}

\newcommand{\cH}{\mathcal{H}}

\newcommand{\cP}{\mathcal{P}}
\newcommand{\bA}{\mathbb{A}}
\newcommand{\bB}{\mathbb{B}}

\date{}
\title{FPT Inapproximability of Directed Cut and Connectivity Problems\thanks{A preliminary version of this paper appeared in IPEC 2019}}

\author[1]{Rajesh~Chitnis\thanks{Work done while at the University of Warwick, UK and supported by ERC grant 2014-CoG 647557.}}
\author[2]{Andreas~Emil~Feldmann\thanks{Supported by the Czech Science Foundation GA{\v C}R
(grant \#19-27871X), and by the Center for Foundations of Modern Computer
Science (Charles Univ.\ project UNCE/SCI/004)}}
\affil[1]{School of Computer Science, University of Birmingham, UK.
\texttt{rajeshchitnis@gmail.com}}
\affil[2]{%
Charles University in Prague, Czechia. \texttt{feldmann.a.e@gmail.com}}

\begin{document}
\renewcommand*{\sectionautorefname}{Section}
\renewcommand*{\subsectionautorefname}{Section}
\renewcommand*{\subsubsectionautorefname}{Section}

\maketitle

\begin{abstract}
Cut problems and connectivity problems on digraphs are two 
well-studied classes of problems from the viewpoint of parameterized complexity.
After a series of papers over the last decade, we now have (almost) tight bounds
for the running time of several standard variants of these problems
parameterized by two parameters: the number $k$ of terminals 
and the size $p$ of the solution. 
When  there is evidence of FPT intractability, then the next natural alternative
is to consider FPT approximations.
In this paper, we show two types of results for directed cut and connectivity
problems, building on existing 
results 
from the literature: first is to circumvent the
hardness results for these problems by designing FPT approximation algorithms,
or alternatively strengthen the existing hardness results by creating
``gap-instances" under stronger hypotheses such as the (Gap-)Exponential Time
Hypothesis (ETH). Formally, we show the following results:\\
%
\textbf{Cutting paths between a set of terminal pairs, i.e., \dirmc}:
      Pilipczuk and Wahlstrom [TOCT~'18] showed that \dirmc is W[1]-hard when
parameterized by $p$ if $k=4$. We complement this by showing the following two
results:
        \begin{itemize}
                \item \dirmc has a $k/2$-approximation in $2^{O(p^2)}\cdot
n^{O(1)}$ time (i.e., a $2$-approximation if $k=4$),
                \item Under Gap-ETH, \dirmc does not admit an
$(\frac{59}{58}-\epsilon)$-approximation
               in $f(p)\cdot n^{O(1)}$ time, for any computable function $f$,
even if $k=4$.
        \end{itemize}
\noindent \textbf{Connecting a set of terminal pairs, i.e., \DSN (DSN)}:
The DSN problem on general graphs is known to be W[1]-hard parameterized by $p+k$ due to Guo et al. [SIDMA '11]. Dinur and Manurangsi [ITCS '18] further showed that there is no FPT $k^{1/4-o(1)}$-approximation algorithm parameterized by $k$, under Gap-ETH. Chitnis et al. [SODA '14] considered the restriction to special graph classes, but unfortunately this does not lead to FPT algorithms either: DSN on planar graphs is W[1]-hard parameterized by $k$. In this paper we consider the \dsnP
problem which is an intermediate version: the
graph is general, but we want to find a solution whose cost is at most that of
an optimal planar solution (if one exists). We show the following lower bounds
for \dsnP:
\begin{itemize}
 \item \dsnP has no $(2-\epsilon)$-approximation in FPT time parameterized by $k$, under Gap-ETH. This answers in the negative a question of Chitnis et al. [ESA '18].
  \item \dsnP is W[1]-hard parameterized by $k+p$. Moreover, under ETH, there is
no $(1+\epsilon)$-approximation for \dsnP in $f(k,p,\epsilon)\cdot
n^{o(k+\sqrt{p+1/\epsilon})}$ time for any computable function~$f$.
\end{itemize}
\textbf{Pairwise connecting a set of terminals, i.e., \SCSS (SCSS)}:
Guo et al. [SIDMA '11] showed that SCSS is W[1]-hard parameterized by $p+k$,
while Chitnis et al. [SODA '14] showed that SCSS remains W[1]-hard parameterized
by $p$, even if the input graph is planar. In this paper we consider the \scssP
problem which is an intermediate version: the
graph is general, but we want to find a solution whose cost is at most that of
an optimal planar solution (if one exists). We show the following lower bounds
for \scssP:
        \begin{itemize}
          \item \scssP is W[1]-hard parameterized by $k+p$. Moreover, under ETH,
there is no $(1+\epsilon)$-approximation for \scssP in $f(k,p,\epsilon)\cdot
n^{o(\sqrt{k+p+\frac{1}{\epsilon}})}$ time for any computable function~$f$.
        \end{itemize}
           Previously, the only known FPT approximation results for SCSS
applied to general graphs parameterized by~$k$: a $2$-approximation
by Chitnis et al. [IPEC '13], and a matching $(2-\epsilon)$-hardness
under Gap-ETH by Chitnis et al. [ESA '18].




\end{abstract}


\section{Introduction}
\label{sec:intro}

Given a weighted directed graph $G=(V,E)$ with two terminal vertices $s,t$ the
problems of finding a minimum weight $s\leadsto t$ cut and a minimum weight
$s\leadsto t$ path can both be famously solved in polynomial time. There are two
natural generalizations when we consider more than two terminals: either we look
for connectivity/cuts between all terminals of a given set, or we look for
connectivity/cuts between a given set of terminal pairs. This leads to the four
problems of \dirmwc, \dirmc, \SCSS and \DSN:
\begin{itemize}
  \item \textbf{Cutting all paths between a set of terminals}: In the \dirmwc
problem, we are given a set of terminals $T=\{t_1, t_2, \ldots, t_k\}$ and the
goal is to find a minimum weight subset $X\subseteq V$ such that $G\setminus X$
has no $t_i \leadsto t_j$ path for any $1\leq i\neq j\leq k$.
  \item \textbf{Cutting paths between a set of terminal pairs}: In the \dirmc
problem, we are given a set of terminal pairs $T=\{(s_i, t_i)\}_{i=1}^{k}$ and
the goal is to find a minimum weight subset $X\subseteq V$ such that
$G\setminus X$ has no $s_i \leadsto t_i$ path for any $1\leq i\leq k$.
  \item \textbf{Connecting all terminals of a given set}: In the \SCSS (SCSS)
problem, we are given a set of terminals $T=\{t_1, t_2, \ldots, t_k\}$ and the
goal is to find a minimum weight subset $X\subseteq V$ such that $G[X]$ has a
$t_i \leadsto t_j$ path for every $1\leq i\neq j\leq k$.
  \item \textbf{Connecting a set of terminal pairs}: In the \DSN (DSN) problem,
we are given a set of terminal pairs $T=\{(s_i, t_i)\}_{i=1}^{k}$ and the goal
is to find a minimum weight subset $X\subseteq V$ such that $G[X]$ has an $s_i
\leadsto t_i$ path for every $1\leq i\leq k$.
\end{itemize}

All four of the aforementioned problems are known to be NP-hard, even for small
values of $k$. One way to cope with NP-hardness is to try to design polynomial
time approximation algorithms with small approximation ratio. However, apart
from \dirmwc, which admits a 2-approximation in polynomial
time~\cite{seffi-multiway}, all the other three problems are known to have
strong lower bounds (functions of $n$) on the approximation ratio of polynomial
time
algorithms~\cite{julia-multicut,dodis1999design,DBLP:conf/stoc/HalperinK03}.
Another way to cope with NP-hardness is to try to design FPT algorithms.
However, apart from \dirmwc which has an FPT algorithm parameterized by the
size $p$ of the cutset, all the other three problems are known to be W[1]-hard
(and hence fixed-parameter intractable) parameterized by size $p$ of the
solution $X$  plus the number $k$ of terminals/terminal pairs. When neither of
the paradigms of polynomial time approximation algorithms nor (exact) FPT
algorithm seem to be successful, the next natural alternative is to try to
design FPT approximation algorithms or show hardness of FPT approximation
results.

In this paper, we consider the remaining three problems of \dirmc, \SCSS
and \DSN, for which strong approximation and parameterized lower bounds exist,
from the viewpoint of FPT approximation algorithms. We obtain two types of
results for these three problems: the first is to circumvent the W[1]-hardness
and polynomial-time inapproximability results for these problems by designing
FPT approximation algorithms, and the second is to strengthen the existing
W[1]-hardness by creating ``gap-instances" under stronger hypotheses than
FPT$\neq$ W[1] such as (Gap-) Exponential Time Hypothesis (ETH).
Throughout, we use $k$ to denote number of terminals or terminal pairs and $p$
to denote size of the solution. First, in Section~\ref{subsec:previous-ours}, we
give a brief overview of the current state-of-the-art results for each the three
problems 
from the lens of polynomial time approximation algorithms, FPT algorithms, and
FPT approximation algorithms followed by the formal statements of our results.
Then, in Section~\ref{subsec:approx-hardness-framework} we describe the recent
flux of results which have set up the framework of FPT hardness of approximation
under (Gap-)ETH, and how we use it obtain our hardness results in this paper.


\subsection{Previous work and our results}
\label{subsec:previous-ours}



\paragraph*{The \textsc{Directed Multicut} problem}

Garg et al.~\cite{DBLP:conf/icalp/GargVY94} showed that \dirmc is NP-hard even
for $k=2$. The current best approximation ratio in terms of $n$ is
$O(n^{11/23}\cdot \log^{O(1)} n)$ due to Agarwal et al.~\cite{noga-mulitcut},
and it is known that \dirmc is hard to approximate in polynomial time to within
a factor of $2^{\Omega(\log^{1-\epsilon} n)}$ for any constant $\epsilon>0$,
unless NP $\subseteq$ ZPP~\cite{julia-multicut}. There is a simple
$k$-approximation in polynomial time obtained by solving each terminal pair as a
separate instance of min $s\leadsto t$ cut and then taking the union of all the
$k$ cuts. Chekuri and Madan~\cite{madan} and later Lee~\cite{lee2017improved}
showed that this is tight: assuming the Unique Games Conjecture of
Khot~\cite{khot-ugc}, it is not possible to approximate \dirmc better than
factor $k$ in polynomial time, for any fixed $k$. On the FPT side, Marx and
Razgon~\cite{marx-multicut} showed that \dirmc is W[1]-hard paramterized by $p$.
For the case of bounded $k$, Chitnis et al.~\cite{rajesh-sicomp} showed that
\dirmc is FPT parameterized by $p$ when $k=2$, but Pilipczuk and
Wahlstrom~\cite{marcin-magnus-4-mulitcut} showed that the problem remains
W[1]-hard parameterized by $p$ when $k=4$. The status of \dirmc parameterized by
$p$ when $k=3$ is an outstanding open question. We first obtain the following
FPT approximation for \dirmc parameterized by $p$, which beats any approximation
obtainable when parameterizing by $k$ (even in XP time) according
to~\cite{madan,lee2017improved}:

\begin{theorem}\label{thm:dirmc-upper}
The \dirmc problem admits an $\lceil k/2\rceil$-approximation in $2^{O(p^2)}\cdot n^{O(1)}$ time.
\end{theorem}

The proof of the above theorem uses the FPT algorithm of Chitnis et
al.~\cite{rajesh-sicomp,rajesh-talg} for \dirmwc parameterized by $p$ as a
subroutine. Note that Theorem~\ref{thm:dirmc-upper} gives an FPT
$2$-approximation for \dirmcfour. We complement this upper bound with a constant
factor lower bound for approximation ratio of any FPT algorithm for \dirmcfour.

\begin{theorem}\label{thm:lb-approx-dirmc-4}
Under Gap-ETH, for any $\eps > 0$ and any computable function $f$, there is
no $f(p)\cdot n^{O(1)}$ time algorithm that computes an $(\frac{59}{58}-\eps)$-approximation
for \dirmcfour.
\end{theorem}

We did not optimize the constant $59/58$ in order to keep the analysis simple:
we believe it can be easily improved, but our techniques would not take it close
to the upper bound of $2$.



\paragraph*{The \textsc{Directed Steiner Network} (DSN) problem}

The DSN problem is known to be NP-hard, and furthermore even computing  an $O(2^{\log^{1-\eps}n})$\hy{}approximation is not possible~\cite{dodis1999design}
in polynomial time, unless NP~$\subseteq$~DTIME$(n^{\text{polylog}(n)})$. The
best known approximation factors for polynomial time algorithms are
$O(n^{2/3+\eps})$ and
$O(k^{1/2+\eps})$~\cite{berman2013approximation,DBLP:journals/talg/ChekuriEGS11,
FKN12}. On the FPT side, Feldman and Ruhl~\cite{feldman-ruhl} designed an
$n^{O(k)}$ algorithm for DSN (cf.~\cite{DBLP:conf/icalp/FeldmannM16}). Chitnis
et al.~\cite{rajesh-soda-14} showed that the Feldman-Ruhl algorithm is tight:
under ETH, there is no $f(k)\cdot n^{o(k)}$ algorithm (for any computable
function $f$) for DSN even if the input graph is a planar directed acyclic
graph. Guo et al.~\cite{guo-suchy} showed that DSN remains W[1]-hard even when
parameterized by the larger parameter~$k+p$.  Dinur and Manurangsi~\cite{pasin-irit} further showed that DSN on
general graphs has no FPT approximation algorithm with ratio $k^{1/4-o(1)}$ when parameterized by $k$, under Gap-ETH.

Chitnis et al.~\cite{rajesh-esa-18}
considered two relaxations of the \DSN problem: the \bidsn problem where the
input graph is bidirected\footnote{Bidirected graphs are directed graphs which
have the property that for every
edge $u\rightarrow v$ in $G$ the reverse edge $v\rightarrow u$ exists in $G$ as
well and moreover has the same weight as $u\rightarrow v$.}, and the \dsnP
problem \textcolor[rgb]{0.00,0.00,0.00}{where the input graph is general but the goal is to find a solution whose cost is at most that of an optimal
planar solution (if one exists)}. The main result of Chitnis et al.~\cite{rajesh-esa-18} is that
although \bidsnP (i.e., the intersection of \bidsn and \dsnP) is W[1]-hard
parameterized by $k+p$, it admits a \emph{parameterized approximation scheme}:
for any $\eps>0$, there is a \smash{$\max\{2^{k^{2^{O(1/\varepsilon)}}},
n^{2^{O(1/\varepsilon)}}\}$} time algorithm for \bidsnP which computes a
$(1+\eps)$-approximation. Such a parameterized approximation is not possible for
\bidsn as Chitnis et al.~\cite{rajesh-esa-18} showed that under Gap-ETH there is
a constant $\alpha>0$ such that there is no FPT $\alpha$-approximation. They
asked whether a parameterized approximation scheme for the remaining variant of
DSN, i.e., the \dsnP problem, exists. We answer this question in the negative
with the following lower bound


\begin{theorem}\label{thm:lb-approx-DSNP}
Under Gap-ETH, for any $\eps > 0$ and any computable function $f$, there is
no $f(k)\cdot n^{O(1)}$ time algorithm that computes a $(2-\eps)$-approximation
for \dsnP, even if the input graph is a directed acyclic graph (DAG).
\end{theorem}

The W[1]-hardness proof of~\cite{rajesh-soda-14} for DSN on planar graphs
parameterized by $k$ does not give hardness parameterized by $p$ since in that
reduction the value of $p$ grows with $n$. Our next result shows
that the slightly more general problem of \dsnP (\textcolor[rgb]{0.00,0.00,0.00}{here the input graph is
general, but we want to find a solution of cost $\leq p$ if there is a planar solution of size $\leq p$}) is
indeed W[1]-hard parameterized by $k+p$. Also we obtain a lower bound for
approximation schemes for this problem under ETH, i.e., under a weaker
assumption than the one used for \autoref{thm:lb-approx-DSNP}.\footnote{In the
following, $o(f(k,p,\eps))$ means any function $g(f(k,p,\eps))$ such that
$g(x)\in o(x)$.}

\begin{theorem}
\label{thm:dsnp-no-epas}
The \dsnP problem is \textup{W[1]}-hard parameterized by $p+k$, even if the input graph is a directed acyclic graph (DAG). Moreover, under
ETH, for any computable function $f$
\begin{itemize}
 \item there is no $f(k,p)\cdot n^{o(k+\sqrt {p})}$ time algorithm for \dsnP,
and
 \item there is no $f(k,\eps,p)\cdot n^{o(k+\sqrt{p+1/\epsilon})}$ time
algorithm which computes a $(1+\eps)$-approximation for \dsnP for every
$\eps>0$.
\end{itemize}
\end{theorem}

\textcolor[rgb]{0.00,0.00,0.00}{Note that just the W[1]-hardness of \dsnP
parameterized by $k+p$ already follows from~\cite{rajesh-esa-18} who showed that
even the special case of \bidsnP is W[1]-hard parameterized by $k+p$. However,
this reduction from~\cite{rajesh-esa-18} was from $\ell$-Clique to an instance
of \bidsnP with $k=O(\ell^2)$ and $p=O(\ell^5)$,
whereas~\autoref{thm:dsnp-no-epas} gives a reduction from $\ell$-Clique to \dsnP
with $k=O(\ell)$ and $p=O(\ell^2)$. This gives much improved lower bounds on
the running times.}




\paragraph*{The \textsc{Strongly Connected Steiner Subgraph} (SCSS) problem}

The SCSS problem is NP-hard, and the best known approximation ratio in polynomial time for SCSS is $k^{\epsilon}$ for any $\epsilon>0$~\cite{DBLP:journals/jal/CharikarCCDGGL99}. A result
of Halperin and Krauthgamer~\cite{DBLP:conf/stoc/HalperinK03} implies SCSS has no $\Omega(\log^{2-\epsilon} n)$-approximation
for any $\epsilon>0$, unless NP has quasi-polynomial Las Vegas algorithms. On
the FPT side, Feldman and Ruhl~\cite{feldman-ruhl} designed an $n^{O(k)}$
algorithm for SCSS (cf.~\cite{DBLP:conf/icalp/FeldmannM16}). Chitnis et
al.~\cite{rajesh-soda-14} showed that the Feldman-Ruhl algorithm is almost
optimal: under ETH, there is no $f(k)\cdot n^{o(k/\log k)}$ algorithm (for any
computable function $f$) for SCSS. Guo et al.~\cite{guo-suchy} showed that SCSS
remains W[1]-hard even when parameterized by the larger parameter $k+p$.
Chitnis et al.~\cite{rajesh-esa-18} showed that the SCSS problem restricted to
bidirected graphs remains NP-hard, but is FPT parameterized by $k$. The SCSS
problem admits a \emph{square-root phenomenon} on planar graphs: Chitnis et
al.~\cite{rajesh-soda-14} showed that SCSS on planar graphs has an $2^{O(k\log
k)}\cdot n^{O(\sqrt{k})}$ algorithm, and under ETH there is a tight lower bound
of $f(k)\cdot n^{o(\sqrt{k})}$ for any computable function $f$. The
W[1]-hardness proof of~\cite{rajesh-soda-14} for SCSS on planar graphs
parameterized by $k$ does not give hardness parameterized by $p$, since in
that reduction the value of $p$ grows with~$n$. Our next result shows that the
slightly more general problem of \scssP (\textcolor[rgb]{0.00,0.00,0.00}{here the input graph is
general, but we want to find a solution of cost $\leq p$ if there is a planar solution of size $\leq p$}) is indeed W[1]-hard
parameterized by $k+p$. We also obtain a lower bound for approximation
schemes for this problem under ETH:

\begin{theorem}
\label{thm:scssp-no-epas}
The \scssP problem is \textup{W[1]}-hard parameterized by $p+k$. Moreover, under
ETH, for any computable function $f$
\begin{itemize}
 \item there is no $f(k,p)\cdot n^{o(\sqrt {k+p})}$ time algorithm for \scssP,
and
 \item there is no $f(k,\eps,p)\cdot n^{o(\sqrt{k+p+1/\epsilon})}$ time
algorithm which computes an $(1+\eps)$-approximation for \mbox{\scssP} for
every $\eps>0$.
\end{itemize}
\end{theorem}

To the best of our knowledge, the only known FPT approximation results for SCSS applied to general graphs parameterized by $k$: a simple FPT $2$-approximation due to Chitnis et al.~\cite{rajesh-ipec-13}, and a matching $(2-\epsilon)$-hardness (for any constant $\epsilon>0$) under Gap-ETH due to Chitnis et al.~\cite{rajesh-esa-18}.


%


\subsection{FPT inapproximability results under (Gap-)ETH}
\label{subsec:approx-hardness-framework}

A standard hypothesis for showing lower bounds for running times of FPT and
exact exponential time algorithms is the Exponential Time Hypothesis (ETH) of
Impagliazzo and Paturi~\cite{ImpagliazzoP01}.
\begin{hypothesis}
\textbf{\emph{Exponential Time Hypothesis (ETH)}}: There exists a constant $\delta>0$ such that no algorithm can decide whether any given $3$-CNF formula is satisfiable in time $O(2^{\delta m})$ where $m$ denotes the number of clauses.
\end{hypothesis}

The original conjecture stated the lower bound as exponential in terms of the
number of variables not clauses, but the above statement follows from the
Sparsification Lemma of~\cite{ImpagliazzoPZ01}. The Exponential Time Hypothesis
has been used extensively to show a variety of lower bounds including those for
FPT algorithms, exact exponential time algorithms, hardness of polynomial time
approximation, and hardness of FPT approximation. We refer the interested
reader to~\cite{eth-survey} for a survey on lower bounds based on ETH.

To show the W[1]-hardness of \dsnP (Theorem~\ref{thm:dsnp-no-epas}) and \scssP
(Theorem~\ref{thm:scssp-no-epas}) parameterized by $k+p$ we design parameterized
reductions from $\ell$-Clique to these problems such that $\max\{k,p\}$ is upper
bounded by a function of $\ell$. Furthermore, by choosing $\epsilon$ to be small
enough such that computing an $(1+\epsilon)$-approximation is the same as
computing the optimal solution, we also obtain runtime lower bounds for
$(1+\epsilon)$-approximations for these two problems by translating the
$f(\ell)\cdot n^{o(\ell)}$ lower bound for $\ell$-Clique~\cite{chen-hardness}
under ETH (for any computable function $f$).

Recently, a \emph{gap version} of the ETH was proposed:
\begin{hypothesis}
\textbf{\emph{Gap-ETH}}~\cite{Dinur16,MR17}:
There exists a constant $\delta > 0$ such that, given a \pname{3CNF}
formula $\Phi$ on $n$ variables, no $2^{o(n)}$-time algorithm can distinguish between the following two
cases correctly with probability at least 2/3:
\begin{itemize}
\item $\Phi$ is satisfiable.
\item Every assignment to the variables violates at least a $\delta$-fraction
of the clauses of $\Phi$.
\end{itemize}
\label{hyp:gap-eth}
\end{hypothesis}

It is known~\cite{param-inapprox,Applebaum17} that Gap-ETH follows from ETH given other standard conjectures,
such as the existence of linear sized PCPs or exponentially-hard locally-computable one-way functions. We refer the interested reader to ~\cite{Dinur16,param-inapprox} for a discussion on why Gap-ETH is a plausible assumption. In a breakthrough result, Chalermsook et al.~\cite{param-inapprox} used Gap-ETH to show that the two famous parameterized intractable problems of Clique and Set Cover are completely inapproximable in FPT time parameterized by the size of the solution. In this paper, we obtain two hardness of approximation results (Theorem~\ref{thm:lb-approx-dirmc-4} and Theorem~\ref{thm:lb-approx-DSNP}) based on Gap-ETH. The starting point of our hardness of approximation results are based on the recent results on parameterized inapproximability of the \pname{Densest $k$-Subgraph} problem. Recall that, in the \pname{Densest $k$-Subgraph (D$k$S)}  problem~\cite{KortsarzP93}, we are given an undirected graph $G = (V, E)$ and an integer $k$ and the goal is to find a subset $S \subseteq V$ of size $\ell$ that induces as many edges in $G$ as
possible. Chalermsook et al.~\cite{param-inapprox} showed that, under randomized Gap-ETH, there is no FPT approximation (parameterized by $k$) with ratio $k^{o(1)}$. This was improved recently by Dinur and Manurangsi~\cite{pasin-irit} who showed better hardness and under deterministic Gap-ETH. We state their result formally\footnote{Dinur and Manurangsi~\cite{pasin-irit} actually state their result for 2-CSPs}:

\begin{theorem}[{\cite[Theorem~2]{pasin-irit}}] \label{thm:inapprox-dks}
Under Gap-ETH, for any function $h(\ell) = o(1)$, there is no $f(\ell)
\polyn$-time algorithm that, given a graph $G$ on $n$ vertices and an
integer $k$, can distinguish between the following two cases:
\begin{itemize}
\item (YES) $G$ contains at least one $\ell$-clique as a subgraph.
\item (NO) Every $\ell$-subgraph of $G$ contains less than $\ell^{h(\ell)-1} \cdot
\binom{\ell}{2}$ edges.
\end{itemize}
\end{theorem}

Note that this result is essentially tight: there is a simple $O(\ell)$ approximation since the number of edges induced by a $\ell$-vertex subgraph is at most $\binom{\ell}{2}$ and at least $\lfloor \ell/2 \rfloor$ (without loss of generality, we can assume there are no isolated vertices). Instead of working with \pname{D$k$S}, we will reduce from a ``colored" version of the problem called \mcsi, which can be defined as follows.

\begin{center}
\noindent\framebox{\begin{minipage}{0.9\textwidth}
\textbf{\mcsi} (\pname{MCSI})\\
\emph{Input }: An instance $\Gamma$ of \pname{MCSI} consists of three
components:
\begin{itemize}
\item An undirected graph $G = (V_G, E_G)$,
\item A partition of vertex set $V_G$ into disjoint subsets $V_1, \dots, V_\ell$,
\item An undirected graph $H = (V_H = \{1, \dots, \ell\}, E_H)$.
\end{itemize}
\emph{Goal}: Find an assignment $\phi: V_H \to V_G$ where $\phi(i) \in V_i$ for every $i \in [\ell]$ that maximizes the number of edges $i-j \in E_H$ such that $\phi(i)-\phi(j) \in E_G$.
\end{minipage}}
\end{center}

This problem is referred to as \pname{Label Cover} in the hardness of approximation literature~\cite{AroraBSS97}. However, Chitnis et al.~\cite{rajesh-esa-18} used the name \mcsi to be consistent with the naming conventions in the FPT community: this problem is an optimization version of \csi~\cite{marx-beat-treewidth}. The graph $H$ is sometimes referred to as the \emph{supergraph} of $\Gamma$.
Similarly, the vertices and edges of $H$ are called \emph{supernodes} and
\emph{superedges} of $\Gamma$. Moreover, the size of $\Gamma$ is defined as $n =
|V_G|$, the number of vertices of $G$. Additionally, for each assignment $\phi$,
we define its value $\val(\phi)$ to be the fraction of superedges $i-j \in E_H$
such that $\phi(i)-\phi(j) \in E_G$; such superedges are said to be
\emph{covered} by $\phi$. The objective of \pname{MCSI} is now to find an
assignment $\phi$ with maximum value. We denote the value of the optimal
assignment by $\val(\Gamma)$, i.e., $\val(\Gamma) = \max_\phi \val(\phi)$.

Using Theorem~\ref{thm:inapprox-dks} we derive the following two corollaries regarding hardness of approximation for \mcsi when the supergraph $H$ has special structure. These corollaries follow quite straightforwardly from Theorem~\ref{thm:inapprox-dks} using the idea of splitters, but we provide proofs here for completeness.

\begin{definition}\label{defn:splitters}
(\textbf{splitters}) Let $n \geq r \geq s$. An $(n,s,r)$-splitter is a family $\Lambda$ of functions $[n]\rightarrow [r]$ such that for every subset $S\subseteq [n]$ of size $s$ there is a function $\lambda \in \Lambda$ such that $\lambda$ is injective on $S$.
\end{definition}

The following constructions of special families of splitters are due to
\cite{color-coding} and \cite{DBLP:conf/focs/NaorSS95}.

\begin{theorem}
\label{thm:color-coding}
There exists a $2^{O(q)} \cdot n^{O(1)}$-time algorithm that takes in $n, q \in \mathbb{N}$ such that $n \geq q$ and outputs an $(n,q,q)$-splitter family of functions $\Lambda_{n, q}$ such that $|\Lambda_{n, q}| = 2^{O(q)}\cdot \log n$.
\end{theorem}



\begin{corollary} 
Assuming Gap-ETH, for any function $h(\ell) = o(1)$, there is no
$f(\ell)\polyn$-time algorithm that, given a \pname{MCSI} instance $\Gamma$ of
size $n$ such that the supergraph $H=K_{\ell}$,
can distinguish between the following two cases:
\begin{itemize}
\item (YES) $\val(\Gamma) = 1$.
\item (NO) $\val(\Gamma) < \ell^{h(\ell)-1}$
\end{itemize}
 \label{crl:inapprox-colored-dks}
\end{corollary}
\begin{proof}
Suppose for the sake of contradiction that there exists an algorithm
$\mathbb{B}$ that can solve the distinguishing problem stated in
Corollary~\ref{crl:inapprox-colored-dks} in $f(\ell) \cdot n^{O(1)}$ time for some
computable function $f$. We will use this to construct another algorithm $\mathbb{B}'$ that
can solve the distinguishing problem stated in Theorem~\ref{thm:inapprox-dks} in
time $f'(\ell) \cdot n^{O(1)}$ for some computable function $f'$, which will thereby violate
Gap-ETH.

The algorithm $\mathbb{B}'$, on input $(G, \ell)$, proceeds as follows. We assume
w.l.o.g. that $V = [n]$. First, $\mathbb{B}'$ runs the algorithm from
Theorem~\ref{thm:color-coding} on $(n, \ell)$ to produce an $(n,\ell,\ell)$-splitter
family of functions $\Lambda_{n, \ell}$. For each $\lambda \in \Lambda_{n, \ell}$, it
creates a \pname{MCSI} instance $\Gamma^\lambda = (G^\lambda, H^\lambda,
V^\lambda_1 \cup \cdots \cup V^\lambda_\ell)$ where
\begin{itemize}
\item the graph $G^\lambda$ is simply the input graph $G$,
\item for each $i \in [\ell]$, we set $V^\lambda_i = \lambda^{-1}(\{i\})$, and,
\item the supergraph $H^\lambda$ is simply the complete graph on $[\ell]$, i.e.,
$H^\lambda = ([\ell], \binom{[\ell]}{2})$.
\end{itemize}
Then, it runs the given algorithm $\mathbb{B}$ on $\Gamma^\lambda$. If
$\mathbb{B}$ returns YES for some $\lambda \in \Lambda$, then $\mathbb{B}'$
returns YES. Otherwise, $\mathbb{B}'$ outputs NO.

It is obvious that the running time of $\mathbb{B}'$ is at most $O(2^{O(\ell)}
f(\ell) \cdot n^{O(1)})$. Moreover, if $G$ contains an $\ell$-clique, say $(v_1,
\dots, v_\ell)$, then by the properties of splitters we are guaranteed that there
exists $\lambda^* \in \Lambda_{n, \ell}$ such that $\lambda^{*}(\{v_1, \dots, v_\ell\}) =
[\ell]$.  Hence, the assignment $i \mapsto v_i$ covers all superedges in
$E_{H^{\lambda^{*}}}$, implying that $\mathbb{B}$ indeed outputs YES on such
$\Gamma^{\lambda^{*}}$. On the other hand, if every $\ell$-subgraph of $G$ contains less
than $\ell^{h(\ell)-1} \cdot \binom{\ell}{2}$ edges, then, for any $\lambda \in \Lambda_{n, \ell}$
and any assignment $\phi$ of $\Gamma^\lambda$, $(\phi(1), \dots, \phi(\ell))$
induces less than $\ell^{h(\ell)-1} \cdot \binom{\ell}{2}$ edges in $G$. This also upper
bounds the number of superedges covered by $\phi$, which implies that
$\Gamma^\lambda$ is a NO instance of Corollary~\ref{crl:inapprox-colored-dks}. Thus,
in this case, $\mathbb{B}$ outputs NO on all $\Gamma^\lambda$'s. In other words,
$\mathbb{B}'$ can correctly distinguish the two cases in
Theorem~\ref{thm:inapprox-dks} in $f'(\ell)\cdot n^{O(1)}$ time where $f'(\ell)=2^{O(\ell)}\cdot f(\ell)$. This
concludes our proof of Corollary~\ref{crl:inapprox-colored-dks}.
\end{proof}

\begin{corollary} 
\label{crl:inapprox-colored-biclique}
Assuming Gap-ETH, for any function $h(\ell) = o(1)$, there is no
$f(\ell)\polyn$-time algorithm that, given a \pname{MCSI} instance $\Gamma$ of
size $n$ such that the supergraph $H$ is the complete bipartite subgraph $K_{\frac{\ell}{2},\frac{\ell}{2}}$,
can distinguish between the following two cases:
\begin{itemize}
\item (YES) $\val(\Gamma) = 1$.
\item (NO) $\val(\Gamma) < \ell^{h(\ell)-1}$.
\end{itemize}
\end{corollary}
\begin{proof}
Suppose for the sake of contradiction that there exists an algorithm
$\mathbb{B}$ that can solve the distinguishing problem stated in
Corollary~\ref{crl:inapprox-colored-biclique} in $f(\ell) \cdot n^{O(1)}$ time for some
function $f$. We will use this to construct another algorithm $\mathbb{B}'$ that
can solve the distinguishing problem stated in Theorem~\ref{thm:inapprox-dks} in
time $f'(\ell) \cdot n^{O(1)}$ for some computable function $f$, which will thereby violate
Gap-ETH.

The algorithm $\mathbb{B}'$, on input $(G, \ell)$, proceeds as follows. We assume
w.l.o.g. that $V = [n]$. First, $\mathbb{B}'$ runs the algorithm from
Theorem~\ref{thm:color-coding} on $(n, \ell)$ to produce an $(n,\ell,\ell)$-splitter
family of functions $\Lambda_{n, \ell}$. For each $\lambda \in \Lambda_{n, \ell}$, it
creates a \pname{MCSI} instance $\Gamma^\lambda = (G^\lambda, H^\lambda,
V^\lambda_1 \cup \cdots \cup V^\lambda_{\ell})$ where
\begin{itemize}
\item the graph $G^\lambda$ is simply the input graph $G$,
\item for each $i \in [\ell]$, we set $V^\lambda_i = \lambda^{-1}(\{i\})$, and,
\item the supergraph $H^\lambda$ is simply $K_{\frac{\ell}{2}, \frac{\ell}{2}}$  where one side of the bipartition is $\{1,2,\ldots, \frac{\ell}{2}\}$ and the other side is $\{\frac{\ell}{2}+1, \frac{\ell}{2}+2, \ldots, \ell\}$.
\end{itemize}
Then, it runs the given algorithm $\mathbb{B}$ on $\Gamma^\lambda$. If
$\mathbb{B}$ returns YES for some $\lambda \in \Lambda$, then $\mathbb{B}'$
returns YES. Otherwise, $\mathbb{B}'$ outputs NO.

It is obvious that the running time of $\mathbb{B}'$ is at most $O(2^{O(\ell)}
f(\ell) \cdot n^{O(1)})$. Moreover, if $G$ contains a $\ell$-clique, say $(v_1,
\dots, v_\ell)$, then by the properties of splitters we are guaranteed that there
exists $\lambda^{*} \in \Lambda_{n, \ell}$ such that $\lambda(\{v_1, \dots, v_\ell\}) =
[\ell]$.  Hence, the assignment $i \mapsto v_i$ covers all superedges in
$E_{H^{\lambda^{*}}}$, implying that $\mathbb{B}$ indeed outputs YES on such
$\Gamma^{\lambda^{*}}$. On the other hand, if every $\ell$-subgraph of $G$ contains less
than $\ell^{h(\ell)-1} \cdot \binom{\ell}{2}$ edges, then, for any $\lambda \in \Lambda_{n, \ell}$
and any assignment $\phi$ of $\Gamma^\lambda$, the mapping $(\phi(1), \dots, \phi(\ell))$
induces less than $\ell^{h(\ell)-1} \cdot \binom{\ell}{2}$ edges in $G$. This also upper
bounds the number of superedges covered by $\phi$. Since $\ell^{h(\ell)-1} \cdot \binom{\ell}{2}\geq \ell^{h(\ell)-1}\cdot (\frac{\ell}{2})^{2}$, it follows implies that
$\Gamma^\lambda$ is a NO instance of Corollary~\ref{crl:inapprox-colored-biclique}. Thus,
in this case, $\mathbb{B}$ outputs NO on all $\Gamma^\lambda$'s. In other words,
$\mathbb{B}'$ can correctly distinguish the two cases in
Theorem~\ref{thm:inapprox-dks} in $f'(\ell) \cdot n^{O(1)})$ time where $f'(\ell)=2^{O(\ell)}\cdot f(\ell)$. This
concludes our proof of Corollary~\ref{crl:inapprox-colored-biclique}.
\end{proof}

\noindent We prove Theorem~\ref{thm:lb-approx-dirmc-4} and Theorem~\ref{thm:lb-approx-DSNP} via reductions from Corollary~\ref{crl:inapprox-colored-dks} and Corollary~\ref{crl:inapprox-colored-biclique} resepctively.

\section{FPT (In)Approximability of \mdseries{\dirmc}}
\label{sec:multicut}

In this section we design an FPT $2$-approximation for \dirmcfour parameterized by $p$ (Section~\ref{sec:multicut-approx}) and complement this with a lower bound (Section~\ref{sec:multicut-hardness}) showing that no FPT algorithm (parameterized by $p$) for \dirmcfour can achieve a ratio of $(\frac{59}{58}-\epsilon)$ under Gap-ETH.

\newcommand{\sincx}{s_{0 \to n}^{x}}
\newcommand{\tincx}{t_{0 \to n}^{x}}
\newcommand{\sincy}{s_{0 \to n}^{y}}
\newcommand{\tincy}{t_{0 \to n}^{y}}
\newcommand{\sdeclt}{s_{n \to 0}^{<}}
\newcommand{\tdeclt}{t_{n \to 0}^{<}}
\newcommand{\sdecgt}{s_{n \to 0}^{>}}
\newcommand{\tdecgt}{t_{n \to 0}^{>}}

\subsection{FPT approximation algorithm}
\label{sec:multicut-approx}


It is well-known that a $k$-approximation can be computed in polynomial time by
taking union of min cuts of each of the $k$ terminal pairs. Chekuri and
Madan~\cite{madan} and later Lee~\cite{lee2017improved} showed that this
approximation ratio is best-possible for polynomial time algorithms under the
Unique Games Conjecture of Khot~\cite{khot-ugc}. The same lower bound also
applies for any constant $k$, i.e., even an XP algorithm parameterized by $k$
cannot compute a better approximation than a polynomial time algorithm. We now
design an FPT $\lceil k/2\rceil$-approximation for \dirmc. The idea is borrowed
from the proof of Chitnis et al.~\cite{rajesh-sicomp} that \dirmctwo is FPT
parameterized by $p$.

\begin{reptheorem}{thm:dirmc-upper}
The \dirmc problem admits a $\lceil k/2\rceil$-approximation in
$2^{O(p^2)}\cdot n^{O(1)}$ time. Formally, the algorithm takes an instance
$(G,\mathcal{T})$ of \dirmc
and in $2^{O(p^2)}\cdot n^{O(1)}$ time either concludes that there is no
solution of cost at most $p$, or produces a solution of cost at most~$p\lceil
k/2\rceil$.
\end{reptheorem}
\begin{proof}
Let the pairs be $\mathcal{T}=\{(s_i, t_i)\ :\ 1\leq i\leq k\}$, and let $\opt$
be the optimum value for the instance $(G,\mathcal{T})$ of \dirmc. For
now, assume that $k$ is even. Introduce $k/2$ new vertices $r_j,q_j$, for $1\leq
j\leq k/2$,  of weight $p+1$ each, and add the following edges:
\begin{itemize}
  \item $r_j\rightarrow s_{2j-1}$ and $t_{2j-1} \rightarrow q_j$
  \item $q_j\rightarrow s_{2j}$ and $t_{2j} \rightarrow r_j$
\end{itemize}
Let the resulting graph be $G'$, and note that $G$ has an $s_i\rightarrow t_i$
path for some $1\leq i\leq k$ if and only if $G'$ has a $q_{i/2}\rightarrow
r_{i/2}$ or $r_{(i-1)/2}\rightarrow q_{(i-1)/2}$ path (depending on whether
$i$ is even or odd). Since the vertices $r_j, q_j$ have weight
$p+1$ each, it follows that $G$ has a solution of size at most $p$ for the
instance $(G, \{(s_{2j-1}, t_{2j-1}), (s_{2j}, t_{2j})\})$ of \dirmc if and
only if $G'$ has a solution of size at most $p$ for the \dirmwc instance with
input graph $G$ and terminals $r_j, q_j$. We use the algorithm
of Chitnis et al.~\cite{rajesh-sicomp,rajesh-talg} for \dirmwc which checks in
$2^{O(p^2)}\cdot n^{O(1)}$ time\footnote{This is independent of number of
the terminals} if there is a solution of cost at most $p$. If there is no
solution of cost at most $p$ between $r_j$ and $q_j$ in $G'$ then this implies
that $G$ has no cut of size at most $p$ separating $(s_{2j-1}, t_{2j-1})$ and
$(s_{2j}, t_{2j})$ and hence $\opt > p$. Otherwise, there is a cut $C_j$ in $G$
of cost at most $p$ which separates $(s_{2j-1}, t_{2j-1})$ and  $(s_{2j},
t_{2j})$.

The output of the algorithm is the cut $C=\bigcup_{j=1}^{k/2}
C_j$. Clearly, if $k$ is even then $C$ is a feasible solution for the instance
$(G, \mathcal{T})$ of \dirmc with cost at most $\sum_{j=1}^{k/2}\cost(C_j) \leq
pk/2$. In case $k$ is odd we use the above procedure for the terminal pairs
$\{(s_i,t_i): 1\leq i\leq k-1\}$, and finally add a min cut between the last
terminal pair $(s_k,t_k)$. This results in the desired $\lceil
k/2\rceil$-approximation.
\end{proof}

\subsection{No FPT $(\frac{59}{58}-\epsilon)$-approximation under Gap-ETH}
\label{sec:multicut-hardness}

With the parameterized hardness of approximating \pname{MCSI} ready, we can now
prove our hardness results for \dirmc with 4 terminal pairs.

\begin{reptheorem}{thm:lb-approx-dirmc-4}
Under Gap-ETH, for any $\eps > 0$ and any computable function $f$, there is
no $f(p)\cdot n^{O(1)}$ time algorithm that computes an $(\frac{59}{58}-\eps)$-approximation
for \dirmcfour.
\end{reptheorem}

Our proof of the parameterized inapproximability of \dirmcfour is based on a
reduction from \mcsi whose properties are described below.

\begin{lemma} \label{lem:dirmc-4}
There exists a polynomial time reduction that,
given an instance $\Gamma = (G, K_\ell, V_1 \cup \cdots \cup V_\ell)$ of \pname{MCSI}
, produces an instance $(G',\mathcal{T}')$ of \dirmcfour such
that
\begin{itemize}
\item (\emph{\textbf{Completeness}}): If $\val(\Gamma) = 1$, then there exists a
solution $N \subseteq V(G')$ of cost $29\ell^2$ for the instance
$(G',\mathcal{T}')$ of \dirmcfour
\item (\emph{\textbf{Soundness}}): If $\val(\Gamma) < \frac{1}{10}$, then every
solution $N \subseteq V(G')$ for the instance $(G',\mathcal{T}')$ of
\dirmcfour has cost more than $29.5\ell^2$.
\item (Parameter Dependency): The size of the solution is $p=O(\ell^2)$.
\end{itemize}
\end{lemma}

\noindent In the proof of Lemma~\ref{lem:dirmc-4}, we actually use the same reduction as from~\cite{marcin-magnus-4-mulitcut}, but with different weights. We reduce to the vertex-weighted variant of \dirmcfour where we have four different types of weights for the vertices:
\begin{itemize}
  \item \emph{light} vertices (shown using gray color) which have weight $B=\frac{\ell^2}{\binom{\ell}{2}}$
  \item \emph{medium} vertices (shown using green color) which have weight $2B$
  \item \emph{heavy} vertices (shown using orange color) which have weight $20\ell$
  \item \emph{super-heavy} vertices (shown using white color) which have weight $100\ell^2$
\end{itemize}

\subsubsection{Construction of the \mdseries{\dirmcfour} instance}
\label{subsubsec:construction-dirmc}

 Without loss of generality (by adding isolated vertices if necessary) we can assume that $|V_i|=n$ for each $i\in [\ell]$. For each $i\in [\ell]$ let $V_i = \{v^{i}_1, v^{i}_2, v^{i}_3, \ldots, v^{i}_{n} \}$. Then $|V(G)|=n\ell$. We now describe the construction of the (vertex-weighted) \dirmcfour instance $(G',\mathcal{T'})$.
\begin{itemize}
  \item Introduce eight terminals, arranged in four terminal pairs as follows: $$ \mathcal{T}' = \{(\sincx,\tincx),\quad (\sincy,\tincy),\quad (\sdeclt,\tdeclt), \quad (\sdecgt,\tdecgt)\}$$ Each of the 8 terminals is super-heavy.

  \item For every $1 \leq i \leq \ell$, we introduce a bidirected path on $2n+1$ vertices (see Figure~\ref{fig:multicut-zoomed-in}) $$Z_i := z^i_0 \leftrightarrow \hat{z}^i_1 \leftrightarrow z^i_1 \leftrightarrow \hat{z}^i_2 \leftrightarrow z^i_2 \leftrightarrow \ldots \leftrightarrow \hat{z}^i_n \leftrightarrow z^i_n,$$ called henceforth the \emph{$z$-path for color class $i$}. For each $0\leq a\leq n$ the vertex $z^{i}_{a}$ is super-heavy and for each $1\leq a\leq n$ the vertex $\hat{z}^{i}_{a}$ is heavy.

  \item For every pair $(i,j)$ where $1 \leq i, j \leq \ell$, $i \neq j$, we introduce two bidirected paths (see Figure~\ref{fig:multicut-zoomed-in} and Figure~\ref{fig:multicut-big-picture}) on $2n+1$ vertices
  $$X_{i,j}:= x^{i,j}_0\leftrightarrow \hat{x}^{i,j}_1\leftrightarrow x^{i,j}_1\leftrightarrow \hat{x}^{i,j}_2\leftrightarrow x^{i,j}_2\leftrightarrow \ldots\leftrightarrow \hat{x}^{i,j}_n\leftrightarrow x^{i,j}_n$$
        and
        $$Y_{i,j}:= y^{i,j}_0\leftrightarrow \hat{y}^{i,j}_1\leftrightarrow y^{i,j}_1 \leftrightarrow \hat{y}^{i,j}_2 \leftrightarrow y^{i,j}_2 \leftrightarrow \ldots \leftrightarrow \hat{y}^{i,j}_n \leftrightarrow y^{i,j}_n$$
We call these paths \emph{the $x$-path and the $y$-path for the pair $(i,j)$}. For each $0\leq a\leq n$ the vertices $x^{i,j}_{a}$ and $y^{i,j}_{a}$ are super-heavy. For each $1\leq a\leq n$ the vertices $\hat{x}^{i,j}_{a}$ and $\hat{y}^{i,j}_{a}$ are medium.

  \item For every pair $(i,j)$ with $1 \leq i,j \leq \ell$, $i\neq j$, and every $0 \leq a \leq n$, we add arcs $(x^{i,j}_a,z^i_a)$ and $(z^i_a,y^{i,j}_a)$. See Figure~\ref{fig:multicut-zoomed-in} for an illustration.

  \item Furthermore, we attach terminals to the paths as follows: (shown using magenta edges in Figure~\ref{fig:multicut-big-picture} and Figure~\ref{fig:multicut-zoomed-in})
        \begin{itemize}
            \item for every pair $(i,j)$ with $1 \leq i,j \leq \ell$, $i \neq j$, we add arcs $(\sincx,x^{i,j}_0)$ and $(y^{i,j}_n,\tincy)$;
            \item for every $1 \leq i \leq \ell$ we add arcs $(\sincy,z^i_0)$ and $(z^i_n,\tincx)$; 
            \item for every pair $(i,j)$ with $1 \leq i < j \leq \ell$ we add arcs $(\sdeclt, x^{i,j}_n)$ and $(y^{i,j}_0,\tdeclt)$;
            \item for every pair $(i,j)$ with $\ell \geq i > j \geq 1$ we add arcs $(\sdecgt, x^{i,j}_n)$ and $(y^{i,j}_0,\tdecgt)$.
        \end{itemize}

  \item For every pair $(i,j)$ with $1 \leq i < j \leq \ell$ we introduce an acyclic $n \times n$ grid $P_{i,j}$ with vertices $p^{i,j}_{a,b}$ for $1 \leq a,b \leq n$
and arcs $(p^{i,j}_{a,b}, p^{i,j}_{a+1,b})$ for every $1 \leq a < n$ and $1 \leq b \leq n$, as well as
$(p^{i,j}_{a,b}, p^{i,j}_{a,b+1})$ for every $1 \leq a \leq n$ and $1 \leq b < n$.
We call this grid $P_{i,j}$ as the \emph{$p$-grid for the pair $(i,j)$}.
We set the vertex $p^{i,j}_{a,b}$ to be a light vertex if $v^i_a v^j_b \in E(G)$, and super-heavy otherwise.
Finally, for every $1 \leq a \leq n$ we introduce the following arcs (shown as dotted in Figure~\ref{fig:multicut-big-picture}):
$$(x^{i,j}_a, p^{i,j}_{a,1}),\quad (p^{i,j}_{a,n}, y^{i,j}_{a-1}),\quad (x^{j,i}_a,p^{i,j}_{1,a}), \quad (p^{i,j}_{n,a}, y^{j,i}_{a-1}).$$
\end{itemize}

\begin{figure}

\centering

\begin{tikzpicture}[scale=1.15]


        \foreach \x in {0,1,2,3,4,5,6}
    \foreach \y in {0,1,2,3,4,5,6}
    {
        \draw [fill=white,thick] plot [only marks, mark size=2.5, mark=*] coordinates
{(\x,\y)};
    }

        \foreach \x in {0,1,2,3,4,5,6}
    \foreach \y in {1,2,3,4,5,6}
    {
    \path (\x,\y) node(a) {} (\x,\y-1) node(b) {};
        \draw[very thick,->] (a) -- (b);
    }

\foreach \y in {0,1,2,3,4,5,6}
    \foreach \x in {0,1,2,3,4,5}
    {
        \path (\x,\y) node(a) {} (\x+1,\y) node(b) {};
        \draw[very thick,->] (a) -- (b);
    }

\draw [gray,thick] plot [only marks, mark size=2.5, mark=*] coordinates {(0,0)};
\draw [gray,thick] plot [only marks, mark size=2.5, mark=*] coordinates {(1,0)};
\draw [gray,thick] plot [only marks, mark size=2.5, mark=*] coordinates {(2,0)};
\draw [gray,thick] plot [only marks, mark size=2.5, mark=*] coordinates {(3,0)};
\draw [gray,thick] plot [only marks, mark size=2.5, mark=*] coordinates {(5,0)};
\draw [gray,thick] plot [only marks, mark size=2.5, mark=*] coordinates {(6,0)};

\draw [gray,thick] plot [only marks, mark size=2.5, mark=*] coordinates {(0,1)};
\draw [gray,thick] plot [only marks, mark size=2.5, mark=*] coordinates {(1,1)};
\draw [gray,thick] plot [only marks, mark size=2.5, mark=*] coordinates {(2,1)};
\draw [gray,thick] plot [only marks, mark size=2.5, mark=*] coordinates {(5,1)};
\draw [gray,thick] plot [only marks, mark size=2.5, mark=*] coordinates {(6,1)};

\draw [gray,thick] plot [only marks, mark size=2.5, mark=*] coordinates {(1,2)};
\draw [gray,thick] plot [only marks, mark size=2.5, mark=*] coordinates {(3,2)};

\draw [gray,thick] plot [only marks, mark size=2.5, mark=*] coordinates {(2,3)};
\draw [gray,thick] plot [only marks, mark size=2.5, mark=*] coordinates {(3,3)};
\draw [gray,thick] plot [only marks, mark size=2.5, mark=*] coordinates {(6,3)};

\draw [gray,thick] plot [only marks, mark size=2.5, mark=*] coordinates {(0,4)};
\draw [gray,thick] plot [only marks, mark size=2.5, mark=*] coordinates {(1,4)};
\draw [gray,thick] plot [only marks, mark size=2.5, mark=*] coordinates {(4,4)};
\draw [gray,thick] plot [only marks, mark size=2.5, mark=*] coordinates {(5,4)};

\draw [gray,thick] plot [only marks, mark size=2.5, mark=*] coordinates {(0,5)};
\draw [gray,thick] plot [only marks, mark size=2.5, mark=*] coordinates {(1,5)};
\draw [gray,thick] plot [only marks, mark size=2.5, mark=*] coordinates {(5,5)};

\draw [gray,thick] plot [only marks, mark size=2.5, mark=*] coordinates {(2,6)};
\draw [gray,thick] plot [only marks, mark size=2.5, mark=*] coordinates {(4,6)};
\draw [gray,thick] plot [only marks, mark size=2.5, mark=*] coordinates {(5,6)};
\draw [gray,thick] plot [only marks, mark size=2.5, mark=*] coordinates {(6,6)};


\foreach \y in {0,1,2,3,4,5,6,7}
    {
        \draw [fill=white,thick] plot [only marks, mark size=2.5, mark=*] coordinates {(-2,\y)};
    }
\foreach \y in {0,1,2,3,4,5,6}
    {
    \draw [green] plot [only marks, mark size=2.5, mark=*] coordinates {(-2.5,\y+0.5)};
    }
\foreach \y in {0,1,...,6}
    {
        \path (-2,\y) node(a) {} (-2.5,\y+0.5) node(b) {};
        \draw[very thick,<->] (a) -- (b);

        \path (-2.5,\y+0.5) node(a) {} (-2,\y+1) node(b) {};
        \draw[very thick,<->] (a) -- (b);
    }

\begin{scope}[shift={(10,-1)}]
\foreach \y in {0,1,2,3,4,5,6,7}
    {
        \draw [fill=white,thick] plot [only marks, mark size=2.5, mark=*] coordinates {(-2,\y)};
    }
\foreach \y in {0,1,2,3,4,5,6}
    {
    \draw [green] plot [only marks, mark size=2.5, mark=*] coordinates {(-1.5,\y+0.5)};
    }
\foreach \y in {0,1,...,6}
    {
        \path (-2,\y) node(a) {} (-1.5,\y+0.5) node(b) {};
        \draw[very thick,<->] (a) -- (b);

        \path (-1.5,\y+0.5) node(a) {} (-2,\y+1) node(b) {};
        \draw[very thick,<->] (a) -- (b);
    }
\end{scope}

\foreach \y in {0,1,...,6}
    {
        \path (-2,\y) node(a) {} (0,\y) node(b) {};
        \draw[thick,middlearrow={>},dotted] (a) -- (b);

        \path (6,\y) node(a) {} (8,\y) node(b) {};
        \draw[thick,middlearrow={>},dotted] (a) -- (b);
    }


\foreach \x in {0,1,2,3,4,5,6,7}
    {
        \draw [fill=white,thick] plot [only marks, mark size=2.5, mark=*] coordinates {(\x,-2)};
    }

\foreach \x in {0,1,2,3,4,5,6}
    {
    \draw [green] plot [only marks, mark size=2.5, mark=*] coordinates {(\x+0.5,-2.5)};
    }
\foreach \x in {0,1,...,6}
    {
        \path (\x,-2) node(a) {} (\x+0.5,-2.5) node(b) {};
        \draw[very thick,<->] (a) -- (b);

        \path (\x+0.5,-2.5) node(a) {} (\x+1,-2) node(b) {};
        \draw[very thick,<->] (a) -- (b);
    }

\begin{scope}[shift={(-1,10)}]
 \foreach \x in {0,1,2,3,4,5,6,7}
    {
        \draw [fill=white,thick] plot [only marks, mark size=2.5, mark=*] coordinates {(\x,-2)};
    }

 \foreach \x in {0,1,2,3,4,5,6}
    {
    \draw [green] plot [only marks, mark size=2.5, mark=*] coordinates {(\x+0.5,-1.5)};
    }
 \foreach \x in {0,1,...,6}
    {
        \path (\x,-2) node(a) {} (\x+0.5,-1.5) node(b) {};
        \draw[very thick,<->] (a) -- (b);

        \path (\x+0.5,-1.5) node(a) {} (\x+1,-2) node(b) {};
        \draw[very thick,<->] (a) -- (b);
    }
 \end{scope}

\foreach \x in {0,1,...,6}
    {
        \path (\x,0) node(a) {} (\x,-2) node(b) {};
        \draw[thick,middlearrow={>},dotted] (a) -- (b);

        \path (\x,8) node(a) {} (\x,6) node(b) {};
        \draw[thick,middlearrow={>},dotted] (a) -- (b);
    }


\draw [fill=white,thick] plot [only marks, mark size=2.5, mark=*] coordinates {(-2,8)} node[label={[xshift=-2mm,yshift=0mm] $\sincx$}] {} ;
\path (-2,8) node(a) {} (-1,8) node(b) {}; \draw[magenta,ultra thick,middlearrow={>}] (a) -- (b);
\path (-2,8) node(a) {} (-2,7) node(b) {}; \draw[magenta,ultra thick,middlearrow={>}] (a) -- (b);

\draw [fill=white,thick] plot [only marks, mark size=2.5, mark=*] coordinates {(-2,-1)} node[label={[xshift=-1mm,yshift=-8mm] $\sdeclt$}] {} ;
\path (-2,-1) node(a) {} (-2,0) node(b) {}; \draw[magenta,ultra thick,middlearrow={>}] (a) -- (b);

\draw [fill=white,thick] plot [only marks, mark size=2.5, mark=*] coordinates {(-1,-2)} node[label={[xshift=-5mm,yshift=-3mm] $\tdecgt$}] {} ;
\path (0,-2) node(a) {} (-1,-2) node(b) {}; \draw[magenta,ultra thick,middlearrow={>}] (a) -- (b);

\draw [fill=white,thick] plot [only marks, mark size=2.5, mark=*] coordinates {(8,-2)} node[label={[xshift=0mm,yshift=-6mm] $\tincy$}] {} ;
\path (8,-1) node(a) {} (8,-2) node(b) {}; \draw[magenta,ultra thick,middlearrow={>}] (a) -- (b);
\path (7,-2) node(a) {} (8,-2) node(b) {}; \draw[magenta,ultra thick,middlearrow={>}] (a) -- (b);

\draw [fill=white,thick] plot [only marks, mark size=2.5, mark=*] coordinates {(8,7)} node[label={[xshift=-1mm,yshift=-1mm] $\tdeclt$}] {} ;
\path (8,6) node(a) {} (8,7) node(b) {}; \draw[magenta,ultra thick,middlearrow={>}] (a) -- (b);

\draw [fill=white,thick] plot [only marks, mark size=2.5, mark=*] coordinates {(7,8)} node[label={[xshift=6mm,yshift=-2mm] $\sdecgt$}] {} ;
\path (6,8) node(a) {} (7,8) node(b) {}; \draw[magenta,ultra thick,middlearrow={<}] (a) -- (b);


\draw [black] plot [only marks, mark size=0, mark=*] coordinates {(-2,0)} node[label={[xshift=4mm,yshift=-7mm] $x^{i,j}_{n}$}] {} ;
\draw [black] plot [only marks, mark size=0, mark=*] coordinates {(-2,7)} node[label={[xshift=5mm,yshift=-5mm] $x^{i,j}_{0}$}] {} ;

\draw [black] plot [only marks, mark size=0, mark=*] coordinates {(0,-2)} node[label={[xshift=-3mm,yshift=-1mm] $y^{j,i}_{0}$}] {} ;
\draw [black] plot [only marks, mark size=0, mark=*] coordinates {(7,-2)} node[label={[xshift=0mm,yshift=-1mm] $y^{j,i}_{n}$}] {} ;

\draw [black] plot [only marks, mark size=0, mark=*] coordinates {(8,-1)} node[label={[xshift=-4mm,yshift=-4mm] $y^{i,j}_{n}$}] {} ;
\draw [black] plot [only marks, mark size=0, mark=*] coordinates {(8,6)} node[label={[xshift=4mm,yshift=-3mm] $y^{i,j}_{0}$}] {} ;

\draw [black] plot [only marks, mark size=0, mark=*] coordinates {(6,8)} node[label={[xshift=4mm,yshift=-8mm] $x^{j,i}_{n}$}] {} ;
\draw [black] plot [only marks, mark size=0, mark=*] coordinates {(-1,8)} node[label={[xshift=0mm,yshift=-8mm] $x^{j,i}_{0}$}] {} ;


\draw [decorate,decoration={brace,amplitude=10pt},xshift=-0.6cm,yshift=0pt] (-2,0) -- (-2,7) node [rotate=90,black,midway,yshift=0.75cm] {\footnotesize $x$-path for the pair $(i,j)$};

\draw [decorate,decoration={brace,amplitude=10pt},yshift=-0.6cm,xshift=0pt] (7,-2) -- (0,-2) node [black,midway,yshift=-0.75cm] {\footnotesize $y$-path for the pair $(j,i)$};

\draw [decorate,decoration={brace,amplitude=10pt},xshift=0.6cm,yshift=0pt] (8,6) -- (8,-1) node [rotate=270,black,midway,yshift=0.75cm] {\footnotesize $y$-path for the pair $(i,j)$};

\draw [decorate,decoration={brace,amplitude=10pt},yshift=0.7cm,xshift=0pt] (-1,8) -- (6,8) node [black,midway,yshift=0.75cm] {\footnotesize $x$-path for the pair $(j,i)$};


\draw [red,thick] (3,2) circle [radius=0.3];

\draw [red,thick] (-2.5,2.5) circle [radius=0.3];
\draw [red,thick] (8.5,1.5) circle [radius=0.3];

\draw [red,thick] (3.5,-2.5) circle [radius=0.3];
\draw [red,thick] (2.5,8.5) circle [radius=0.3];


\end{tikzpicture}

\caption{Illustration of the reduction for \dirmcfour. For $1\leq i<j\leq \ell$, the grid $P_{i,j}$ is surrounded by the bidirectional paths $X_{i,j}$  on the left, $X_{j,i}$ on the top, $Y_{i,j}$ on the right and $Y_{j,i}$ on the bottom. Edges incident on terminals are shown in magenta. Green vertices are medium, orange vertices are heavy and white vertices are super-heavy. A desired solution is marked by red circles. \label{fig:multicut-big-picture}}

\end{figure}
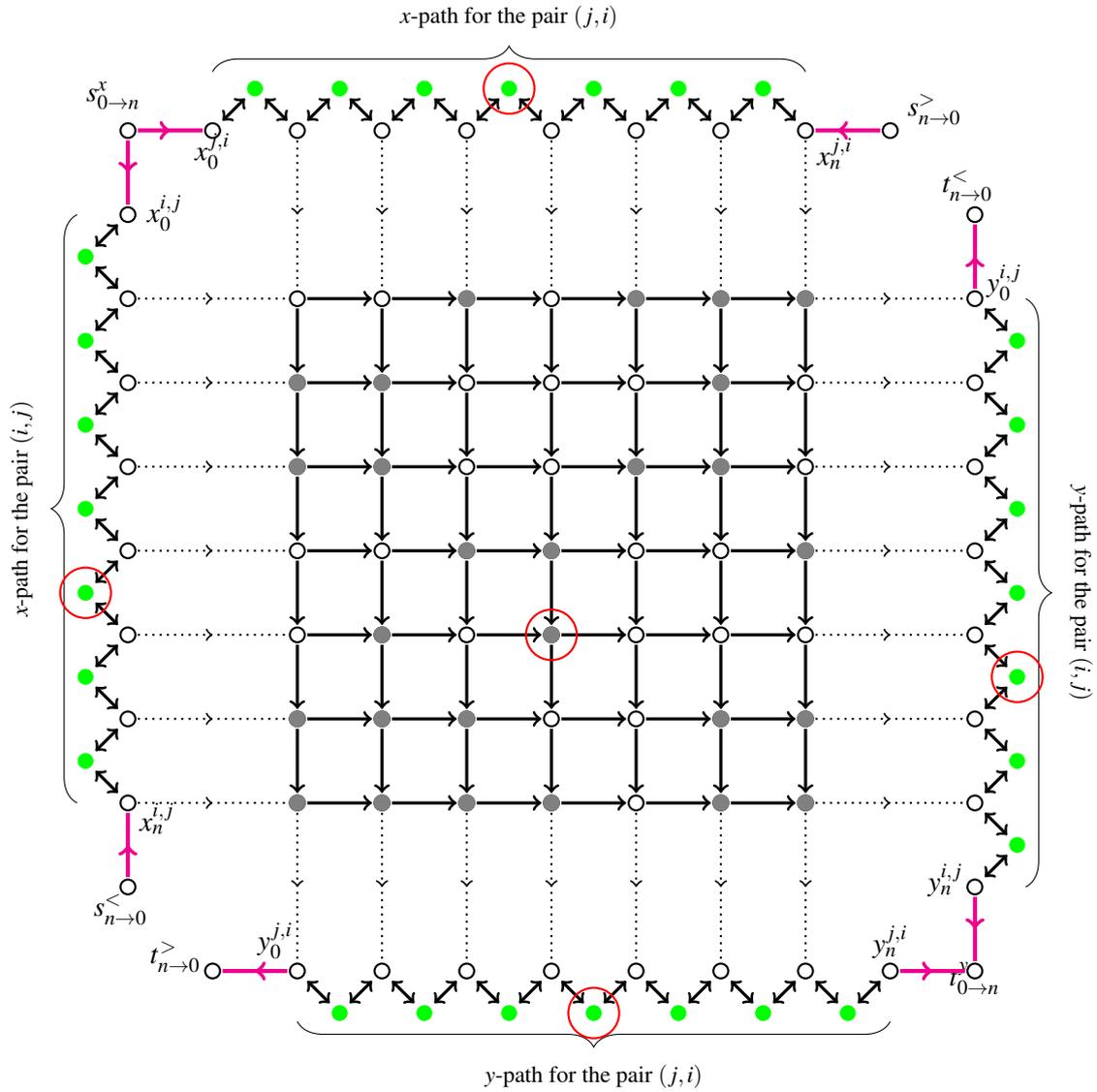

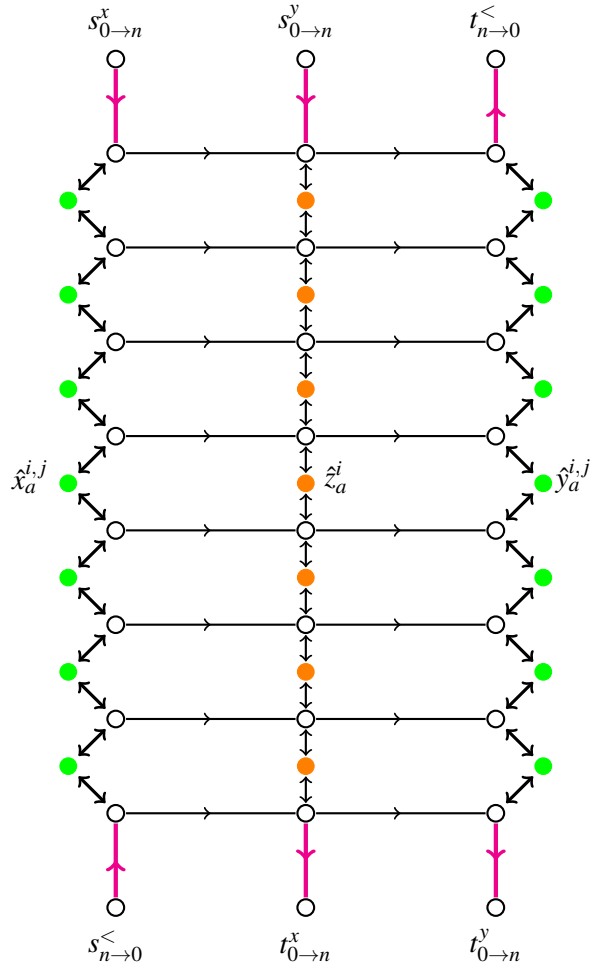
\begin{figure}

\centering

\begin{tikzpicture}[scale=1.25]


\foreach \y in {0,1,2,3,4,5,6,7}
    {
        \draw [fill=white,thick] plot [only marks, mark size=2.5, mark=*] coordinates {(0,\y)};
    }
\foreach \y in {0,1,2,3,4,5,6}
    {
    \draw [green] plot [only marks, mark size=2.5, mark=*] coordinates {(-0.5,\y+0.5)};
    }
\foreach \y in {0,1,...,6}
    {
        \path (0,\y) node(a) {} (-0.5,\y+0.5) node(b) {};
        \draw[very thick,<->] (a) -- (b);

        \path (-0.5,\y+0.5) node(a) {} (0,\y+1) node(b) {};
        \draw[very thick,<->] (a) -- (b);
    }

\foreach \y in {0,1,2,3,4,5,6,7}
    {
        \draw [fill=white,thick] plot [only marks, mark size=2.5, mark=*] coordinates {(2,\y)};
    }
\foreach \y in {0,1,2,3,4,5,6}
    {
    \draw [orange] plot [only marks, mark size=2.5, mark=*] coordinates {(2,\y+0.5)};
    }
\foreach \y in {0,1,...,6}
    {
        \path (2,\y) node(a) {} (2,\y+0.5) node(b) {};
        \draw[thick,<->] (a) -- (b);

        \path (2,\y+0.5) node(a) {} (2,\y+1) node(b) {};
        \draw[thick,<->] (a) -- (b);
    }

\foreach \y in {0,1,2,3,4,5,6,7}
    {
        \draw [fill=white,thick] plot [only marks, mark size=2.5, mark=*] coordinates {(4,\y)};
    }
\foreach \y in {0,1,2,3,4,5,6}
    {
    \draw [green] plot [only marks, mark size=2.5, mark=*] coordinates {(4.5,\y+0.5)};
    }
\foreach \y in {0,1,...,6}
    {
        \path (4,\y) node(a) {} (4.5,\y+0.5) node(b) {};
        \draw[very thick,<->] (a) -- (b);

        \path (4.5,\y+0.5) node(a) {} (4,\y+1) node(b) {};
        \draw[very thick,<->] (a) -- (b);
    }

\foreach \y in {0,1,2,3,4,5,6,7}
    {
        \path (0,\y) node(a) {} (2,\y) node(b) {};
        \draw[thick,middlearrow={>}] (a) -- (b);

        \path (2,\y) node(a) {} (4,\y) node(b) {};
        \draw[thick,middlearrow={>}] (a) -- (b);
    }

\draw [fill=white,thick] plot [only marks, mark size=2.5, mark=*] coordinates {(0,8)} node[label={[xshift=0mm,yshift=0mm] $\sincx$}] {} ;
\path (0,8) node(a) {} (0,7) node(b) {}; \draw[magenta,ultra thick,middlearrow={>}] (a) -- (b);

\draw [fill=white,thick] plot [only marks, mark size=2.5, mark=*] coordinates {(2,8)} node[label={[xshift=0mm,yshift=0mm] $\sincy$}] {} ;
\path (2,8) node(a) {} (2,7) node(b) {}; \draw[magenta,ultra thick,middlearrow={>}] (a) -- (b);

\draw [fill=white,thick] plot [only marks, mark size=2.5, mark=*] coordinates {(4,8)} node[label={[xshift=0mm,yshift=0mm] $\tdeclt$}] {} ;
\path (4,7) node(a) {} (4,8) node(b) {}; \draw[magenta,ultra thick,middlearrow={>}] (a) -- (b);

\draw [fill=white,thick] plot [only marks, mark size=2.5, mark=*] coordinates {(0,-1)} node[label={[xshift=0mm,yshift=-10mm] $\sdeclt$}] {} ;
\path (0,-1) node(a) {} (0,0) node(b) {}; \draw[magenta,ultra thick,middlearrow={>}] (a) -- (b);

\draw [fill=white,thick] plot [only marks, mark size=2.5, mark=*] coordinates {(2,-1)} node[label={[xshift=0mm,yshift=-10mm] $\tincx$}] {} ;
\path (2,0) node(a) {} (2,-1) node(b) {}; \draw[magenta,ultra thick,middlearrow={>}] (a) -- (b);

\draw [fill=white,thick] plot [only marks, mark size=2.5, mark=*] coordinates {(4,-1)} node[label={[xshift=0mm,yshift=-10mm] $\tincy$}] {} ;
\path (4,0) node(a) {} (4,-1) node(b) {}; \draw[magenta,ultra thick,middlearrow={>}] (a) -- (b);


\draw [black] plot [only marks, mark size=0, mark=*] coordinates {(-0.5,3.5)} node[label={[xshift=-5mm,yshift=-4mm] $\hat{x}^{i,j}_{a}$}] {} ;

\draw [black] plot [only marks, mark size=0, mark=*] coordinates {(2,3.5)} node[label={[xshift=4mm,yshift=-4mm] $\hat{z}^{i}_{a}$}] {} ;

\draw [black] plot [only marks, mark size=0, mark=*] coordinates {(4.5,3.5)} node[label={[xshift=4mm,yshift=-4mm] $\hat{y}^{i,j}_{a}$}] {} ;

\end{tikzpicture}

\caption{Illustration of the reduction for \dirmcfour. For every $1\leq i<j\leq \ell$, the z-path $Z_{i}$ corresponding to the color class $i$ is surrounded by the bidirectional paths $X_{i,j}$  on the left and $Y_{i,j}$ on the right. Edges incident on terminals are shown in magenta. Green vertices are medium,orange vertices are heavy and white vertices are super-heavy. \label{fig:multicut-zoomed-in}}

\end{figure}

This concludes the construction of the instance $(G',\mathcal{T}')$ of \dirmcfour. Note that $|V(G')|=(n+\ell)^{O(1)}$, and also $G'$ can be constructed in $(n+\ell)^{O(1)}$ time.

\subsubsection{Completeness of Lemma~\ref{lem:dirmc-4}: $\val(\Gamma)=1 \Rightarrow$ Multicut of cost $\leq 29\ell^2$}

Suppose that $\val(\Gamma) = 1$, i.e., $G$ has a $\ell$-clique which has exactly one vertex in each $V_i$ for $1\leq i\leq \ell$.  Let this clique be given by $\{v^i_{\alpha(i)}: 1 \leq i \leq \ell\}$.
Define
$$X = \{\hat{x}^{i,j}_{\alpha(i)}, \hat{y}^{i,j}_{\alpha(i)} : 1 \leq i,j \leq \ell, i \neq j\}
\cup \{\hat{z}^i_{\alpha(i)} : 1 \leq i \leq \ell \}
\cup \{p^{i,j}_{\alpha(i),\alpha(j)} : 1 \leq i < j \leq \ell\}.$$
Note that $X$ consists of exactly $\ell$ heavy $\hat{z}^i_{\alpha(i)}$
vertices, $4\binom{\ell}{2}$ medium $\hat{x}^{i,j}_{\alpha(i)}$ and
$\hat{y}^{i,j}_{\alpha(i)}$ vertices, and $\binom{\ell}{2}$ light
$p^{i,j}_{\alpha(i),\alpha(j)}$ vertices (the fact that
$p^{i,j}_{\alpha(i),\alpha(j)}$ is light for every $1 \leq i < j \leq \ell$ follows from the assumption that the vertices $v^i_{\alpha(i)}$ induce a clique in $G$).
Hence, the weight of $X$ is exactly $\ell\cdot 20\ell + \binom{\ell}{2}\cdot
(4\cdot 2B) + \binom{\ell}{2}\cdot B= 20\ell^2 + \binom{\ell}{2}\cdot 9B =
29\ell^2$. As shown in~\cite{marcin-magnus-4-mulitcut}, this set $X$ is a
cutset for the instance $(G',\mathcal{T}')$ of \dirmcfour. For the sake of completeness, we repeat the arguments in Section~\ref{app:completeness}.

\subsubsection{Soundness of Lemma~\ref{lem:dirmc-4}: Multicut of cost $\leq 29.5\ell^2 \Rightarrow \val(\Gamma)\geq \frac{1}{10}$}

Let $\mathcal{X}$ be a solution to the instance $(G',\mathcal{T}')$ of \dirmcfour such that weight of $\mathcal{X}$ is $29.5\ell^2$. We now show that $\val(\Gamma)\geq \frac{1}{10}$.

\begin{observation}
Note that every super-heavy vertex has weight $100\ell^2$ and hence $\mathcal{X}$ cannot contain any super-heavy vertex.
\label{obs:no-super-heavy}
\end{observation}

\begin{lemma}
 For each $i\in [\ell]$, the solution $\mathcal{X}$ contains at least one heavy vertex from $Z_i$.
\label{lem:dirmc-at-least-one-Z}
\end{lemma}
\begin{proof}
Note that there is a  $\sincx \leadsto \tincx$ path as follows:
\begin{itemize}
  \item $\sincx \rightarrow x^{i,j}_{0}\rightarrow z^{i}_{0}$
  \item Use the $z$-path for color class $i$ in one direction from $z^{i}_{0}$ to $z^{i}_{n}$
  \item $z^{i}_{n}\rightarrow \tincx$
\end{itemize}

From Observation~\ref{obs:no-super-heavy}, we know that $\mathcal{X}$ cannot contain any super-heavy vertex. Each vertex from the set $\{\sincx, x^{i,j}_{0}, \tincx\}$ is super-heavy. Hence, $X$ must contain at least one heavy vertex from $Z_i$.
\end{proof}

\begin{lemma}
For each $1\leq i\neq j\leq \ell$, the solution $\mathcal{X}$ contains at least one medium vertex from $X_{i,j}$ and at least one medium vertex from $Y_{i,j}$.
\label{lem:dirmc-at-least-one-X-Y}
\end{lemma}
\begin{proof}
There is a path from $\sincx$ to $\tincx$ that traverses the entire $x$-path for the pair $(i,j)$ up to the vertex $x^{i,j}_n$, and then uses the arc $(x^{i,j}_n,z^i_n)$ to reach
$\tincx$. From Observation~\ref{obs:no-super-heavy}, we know that $\mathcal{X}$ cannot contain any super-heavy vertex. Each vertex from the set $\{\sincx, x^{i,j}_{n}, \tincx\}$ is super-heavy. Hence, $X$ must contain at least one medium vertex from $X_{i,j}$.

There is a path from $\sincy$ to $\tincy$ that starts with using the arc $(z^i_0,y^{i,j}_0)$, and then traverses the $y$-path for the pair $(i,j)$
up to the vertex $y^{i,j}_n$. From Observation~\ref{obs:no-super-heavy}, we know that $\mathcal{X}$ cannot contain any super-heavy vertex. Each vertex from the set $\{\sincy, z^{i}_{0}, \tincy\}$ is super-heavy. Hence, $X$ must contain at least one medium vertex from $Y_{i,j}$.
\end{proof}

\begin{definition}
An integer $i\in [\ell]$ is \emph{good} if $\mathcal{X}$ contains exactly one heavy vertex from the $z$-path for the color class $i$, i.e., $|\mathcal{X}\cap Z_i|=1$. In this case, we say that $v^{i}_{\beta_i}$ be the unique vertex from the $z$-path for class $i$ in the solution $\mathcal{X}$.
\end{definition}

\begin{lemma}
\label{lem:dirmc-nr-good}
Let $\good=\{i\in [\ell]\ :\ i\ \text{is good}\}$. Then $|\good|\geq \frac{37\ell}{40}$
\end{lemma}
\begin{proof}
From Lemma~\ref{lem:dirmc-at-least-one-Z} we have a contribution of at least
$\ell\cdot 20\ell =20\ell^2$ towards weight of $\mathcal{X}$ by heavy vertices.
From Lemma~\ref{lem:dirmc-at-least-one-X-Y} we have a contribution of at least
$2\ell(\ell-1)\cdot 2B =8\ell^2$ towards weight of $\mathcal{X}$ by medium
vertices.

By Lemma~\ref{lem:dirmc-at-least-one-Z}, every $i\notin \good$ must contribute at least two heavy vertices to $\mathcal{X}$. Hence, we have
$$29.5\ell^2 \geq\ \text{weight of}\ \mathcal{X} \geq 20\ell^2 + 8\ell^2 + (\ell-|\good|)\cdot (20\ell)$$
Hence, $|\good|\geq \frac{37\ell}{40}$.
\end{proof}

\begin{definition}
Let $1\leq i< j\leq \ell$. We say that the pair $(i,j)$ is \emph{great} if $\mathcal{X}$ contains
\begin{itemize}
  \item exactly one medium vertex from the $x$-path for the pair $(i,j)$
  \item exactly one medium vertex from the $y$-path for the pair $(i,j)$
  \item exactly one medium vertex from the $x$-path for the pair $(j,i)$
  \item exactly one medium vertex from the $y$-path for the pair $(j,i)$
  \item exactly one light vertex from the $p$-grid for the pair $(i,j)$
\end{itemize}
\label{defn:great}
\end{definition}

Let $\goodpairs=\{(i,j)\ :\ 1\leq i<j\leq \ell, \ i,j\in \good\}$

\begin{lemma}
Let $1\leq i< j\leq \ell$. If both $i$ and $j$ are good, and the pair $(i,j)$ is great then $v^{i}_{\beta_i}-v^{j}_{\beta_j} \in E(G)$.
\label{lem:use-of-good-great}
\end{lemma}
\begin{proof}
Fix a pair $(i,j)$ with $1 \leq i<j \leq \ell$. Since $i,j\in \good$ we have that $\mathcal{X}\cap Z_i =
\{\hat{z}^{i}_{\beta_i}\}$ and $\mathcal{X}\cap Z_j =\{\hat{z}^{j}_{\beta_j}\}$.
Since $(i,j)$ is great, let
\begin{itemize}
  \item $\mathcal{X}\cap X_{i,j} = $ for some $\hat{x}^{i,j}_{a}$
  \item $\mathcal{X}\cap Y_{i,j} = $ for some $\hat{y}^{i,j}_{c}$
  \item $\mathcal{X}\cap X_{j,i} = $ for some $\hat{x}^{i,j}_{a'}$
  \item $\mathcal{X}\cap Y_{j,i} = $ for some $\hat{y}^{i,j}_{c'}$
  \item $\mathcal{X}\cap P_{i,j} = $ for some $p^{i,j}_{\mu,\delta}$
\end{itemize}

Observe that $a \leq \beta_i$, as otherwise the path from $\sincx$ to $\tincx$
that traverses the $x$-path for the pair $(i,j)$ up to the vertex
$x^{i,j}_{a-1}$, uses the arc $(x^{i,j}_{a-1},z^i_{a-1})$, and traverses the
$z$-path for the color class $i$ up to the endpoint $z^i_n$ is not cut by
$\mathcal{X}$, a contradiction.
A similar argument for the terminal pair $(\sincy,\tincy)$ implies that $\beta_i \leq c$.
However, if $a < \beta_i$, then the path from $\sdeclt$ to $\tdeclt$ that traverses the $x$ path for the pair $(i,j)$ up to the vertex $x^{i,j}_a$, uses the arc
$(x^{i,j}_a,z^i_a)$, traverses the $z$-path for the color class $i$ up to the
endpoint $z^i_0$, and finally uses the arc $(z^i_0,y^{i,j}_0)$, is not cut by
$\mathcal{X}$, a contradiction. A similar argument gives a contradiction
if $\beta_i < c$. Hence, we have that $a=\beta_i = c$. Similarly, we can show
that $a'=\beta_j=c'$.

Recall that $\mathcal{X}\cap P_{i,j}=p^{i,j}_{\mu,\delta}$. Hence, this vertex $p^{i,j}_{\mu,\delta}$ must hit each of the following two paths which were not cut by the heavy or medium vertices in $\mathcal{X}$:
\begin{itemize}
\item A path $P_1$ from $\sdeclt$ to $\tdeclt$ that traverses the $x$-path for the pair $(i,j)$ up to the vertex $x^{i,j}_{\alpha(i)}$, uses the arc $(x^{i,j}_{\alpha(i)},p^{i,j}_{\alpha(i),1})$,
traverses the $\alpha(i)$-th row of the $p$-grid for the pair $(i,j)$ up to the
vertex $p^{i,j}_{\alpha(i),n}$, uses the arc $(p^{i,j}_{\alpha(i),n},
y^{i,j}_{\alpha(i)-1})$, and traverses the $y$-path for the pair $(i,j)$ up to
the endpoint $y^{i,j}_0$.
\item A path $P_2$ from $\sdecgt$ to $\tdecgt$ that traverses the $x$-path for the pair $(j,i)$ up to the vertex $x^{j,i}_{\alpha(j)}$, uses the arc $(x^{j,i}_{\alpha(j)},p^{i,j}_{1,\alpha(j)})$,
traverses the $\alpha(j)$-th column of the $p$-grid for the pair $(i,j)$ up to the vertex $p^{i,j}_{n,\alpha(j)}$, uses the arc $(p^{i,j}_{n,\alpha(j)}, y^{j,i}_{\alpha(j)-1})$, and traverses the $y$-path for the pair $(j,i)$ up to the endpoint $y^{j,i}_0$.
\end{itemize}

However, the only vertex in common of the two aforementioned paths for a fixed choice of $(i,j)$, $1 \leq i < j \leq n$, is the vertex $p^{i,j}_{\beta(i),\beta(j)}$. Hence, $\mu=\beta_i$ and $\delta=\beta_j$. By Observation~\ref{obs:no-super-heavy}, it follows that $p^{i,j}_{\beta_i,\beta_j}$ is light, i.e., $v^i_{\alpha(i)} v^j_{\alpha(j)} \in E(G)$.
\end{proof}

\begin{definition}
Let $1\leq i<j\leq \ell$. We define $\mathcal{X}_{i,j} = \mathcal{X}\cap (X_{i,j}\cup X_{j,i}\cup Y_{i,j}\cup Y_{j,i}\cup P_{i,j})$
\end{definition}

\begin{lemma}
Let $1\leq i<j\leq \ell$ be such that $i,j\in \good$. Then either
\begin{itemize}
  \item the pair $(i,j)$ is great and weight of $\mathcal{X}_{i,j}$ is exactly $9B$, or
  \item weight of $\mathcal{X}_{i,j}$ is at least $10B$
\end{itemize}
\label{lem:edge-cost-4+gamma-non-edge-higher-cost}
\end{lemma}
\begin{proof}
By Lemma~\ref{lem:dirmc-at-least-one-X-Y}, we know that $\mathcal{X}_{i,j}$
contains at least one medium vertex from each of the four paths $X_{i,j},
X_{j,i}, Y_{i,j}$ and $Y_{i,j}$. Hence, $\mathcal{X}_{i,j}$ contains at least
four medium vertices. If $\mathcal{X}_{i,j}$ contains at least 5 medium vertices
then its weight is at least $5(2B) = 10B$. Hence, suppose that
$\mathcal{X}_{i,j}$ contains exactly four medium vertices. We now consider how
many light vertices from the $p$-grid for the pair $(i,j)$ are present in
$\mathcal{X}_{i,j}$:
\begin{itemize}
  \item $\mathcal{X}_{i,j}$ contains exactly one light vertex from the $p$-grid
for the pair $(i,j)$, i.e., the pair $(i,j)$ is great. Note that in this case
the weight of $\mathcal{X}_{i,j}$ is exactly $4\cdot(2B)+B = 9B$
  \item Otherwise $\mathcal{X}_{i,j}$ contains at least two light vertices from
the $p$-grid for the pair $(i,j)$. In this case, the weight of
$\mathcal{X}_{i,j}$ is at least $4(2B) + B + B = 10B$.\qedhere
\end{itemize}
\end{proof}

\begin{lemma}
Let $\mathcal{E}=\{1\leq i<j\leq \ell\ :\ i,j\in \good\ \text{and}\ (i,j)\
\text{is great}\}$. Then $|\mathcal{E}|\geq \frac{1}{10}\cdot
\binom{\ell}{2}$
\label{lem:gamma-lower-bound}
\end{lemma}
\begin{proof}
From Lemma~\ref{lem:dirmc-at-least-one-Z} we have a contribution of at least
$\ell\cdot 20\ell =20\ell^2$ towards weight of $\mathcal{X}$. From
Lemma~\ref{lem:dirmc-at-least-one-X-Y} we have a contribution of at least
$\binom{\ell}{2}\cdot 8B =8\ell^2$ towards weight of $\mathcal{X}$, i.e, for
each $1\leq i\neq j\leq \ell$ the set $\mathcal{X}_{i,j}$ has at least four
medium vertices.

We now count how much additional weight can be charged to $\mathcal{X}$. Since
$|\good|\geq \frac{37\ell}{40}$ by Lemma~\ref{lem:dirmc-nr-good} we have that
$|\goodpairs|=\binom{37\ell/40}{2}$ pairs $(i,j)$ such that $1\leq i<j\leq \ell$
and $i,j\in \good$. By Lemma~\ref{lem:edge-cost-4+gamma-non-edge-higher-cost},
we have that:
\begin{itemize}
  \item each pair $(i,j)$ in $\mathcal{E}$ contributes an additional cost of $B$ to $\mathcal{X}_{i,j}$ through a light vertex
  \item each pair $(i,j)\in (\goodpairs\setminus \mathcal{E})$ contributes an
additional cost of at least $2B$ to $\mathcal{X}_{i,j}$ through either a medium
vertex or two light vertices
\end{itemize}
Hence, we have that
$$ 29.5\ell^2 \geq \text{weight of}\ \mathcal{X}\geq 20\ell^2 + 8\ell^2 +
|\mathcal{E}|\cdot B + \Big(\binom{37\ell/40}{2}-|\mathcal{E}|\Big)\cdot 2B$$
%
Rearranging we have that
$$|\mathcal{E}| \geq 2\binom{37\ell/40}{2} - \frac{3}{2}\cdot
\binom{\ell}{2}$$
For $\ell\geq 3$, one can easily verify that $2\binom{37\ell/40}{2} - \frac{3}{2}\cdot \binom{\ell}{2}\geq \frac{1}{10} \binom{\ell}{2}$. Hence, we have
$$|\mathcal{E}| \geq 2\binom{37\ell/40}{2} - \frac{3}{2}\cdot \binom{\ell}{2}
\geq \frac{1}{10} \binom{\ell}{2} \qedhere$$
\end{proof}

Consider the following $\ell$-vertex subgraph $C$: for each $i\in [\ell]$
\begin{itemize}
  \item if $i\in [\ell]$ is good then add $v^{i}_{\beta_i}$ to $C$,
  \item otherwise add any vertex from $V_i$ into $C$.
\end{itemize}
From Lemma~\ref{lem:gamma-lower-bound} it follows that there are at least $\frac{1}{10}\cdot \binom{\ell}{2}$ edges in $G$ which have both endpoints in $C$, and hence $\val(\Gamma)\geq \frac{1}{10}$



\subsection{Finishing the proof of Theorem~\ref{thm:lb-approx-dirmc-4}}

We again prove by contrapositive. Suppose that, for some constant $\varepsilon
> 0$ and for some computable function $f(p)$ independent of $n$, there exists an $f(p)
\polyn$-time $(\frac{59}{58} - \varepsilon)$-approximation algorithm for
\dirmc. Let us call this algorithm $\bA$.
%

We create an algorithm $\bB$ that can
distinguish between the two cases of Corollary~\ref{crl:inapprox-colored-dks} with
$h(\ell) = 1- \frac{\log (10)}{\log \ell}= o(1)$. Our new algorithm $\bB$ works as follows.
Given an instance $(G, H, V_1 \cup \cdots \cup V_{\ell})$ of \pname{MCSI} where
$H=K_{\ell}$, the algorithm $\bB$ uses the reduction from Lemma~\ref{lem:dirmc-4} to
create a \dirmcfour instance $(G',\mathcal{T}')$ with $4$ terminal pairs. $\bB$ then runs
$\bA$ on this instance with $p=29\ell^2$; if $\bA$ returns a solution $N$ of
cost less than $29.5\ell^2$, then $\bB$ returns YES. Otherwise, $\bB$
returns NO.

To see that algorithm $\bB$ can indeed distinguish between the YES and NO
cases, first observe that, in the YES case the completeness property of Lemma~\ref{lem:dirmc-4} guarantees
that the optimal solution has cost at most $29\ell^2$. Since $\bA$ is a $(\frac{59}{58} - \varepsilon)$-approximation algorithm, it
returns a solution of cost at most $(\frac{59}{58} - \varepsilon)\cdot 29\ell^2 < 29.5\ell^2$: this means that $\bB$
outputs YES. On the other hand, if $(G, H, V_1 \cup \cdots \cup V_\ell)$ is a NO
instance, i..e,  $\val(\Gamma)<\frac{1}{10} = \ell^{h(\ell)-1}$, then the soundness property of Lemma~\ref{lem:dsnp-inapprox} guarantees
that the optimal solution in $G'$ has cost more than $29.5\ell^2$ (which is greater than $(\frac{59}{58} - \varepsilon)\cdot 29\ell^2$) and hence $\bB$ correctly outputs NO.

Finally, observe that the running time of $\bB$ is $f(p)\cdot |V(G')|^{O(1)}$ plus the $(|V(G)|+\ell)^{O(1)}$ time needed to construct $G'$. Since $|V(G')|=(|V(G)+\ell|)^{O(1)}$ and $p=O(\ell^2)$ it follows that the total running time is $g(\ell)\cdot |V(G)|$ for some computable function $g$. Hence, from Corollary~\ref{crl:inapprox-colored-dks}, Gap-ETH is violated.

\section{FPT inapproximability for \dsnP}

\subsection{$(2-\eps)$-hardness for FPT approximation under Gap-ETH}
\label{subsec:dsnp-gap-eth}


The goal of this section is to show the following theorem:

\begin{reptheorem}{thm:lb-approx-DSNP}
Under Gap-ETH, for any $\eps > 0$ and any computable function $f$, there is
no $f(k)\cdot n^{O(1)}$ time algorithm that computes a $(2-\eps)$-approximation
for \dsnP.
\end{reptheorem}

\subsubsection{Reduction from Colored Biclique to \dsnP}

\begin{lemma} \label{lem:dsnp-inapprox}
For every constant $\gamma > 0$, there exists a polynomial time reduction that,
given an instance $\Gamma = (G, H, V_1 \cup \cdots \cup V_\ell,W_{1}, W_{2}, \ldots, W_{\ell})$ of \pname{MCSI}
where the supergraph $H$ is $K_{\ell,\ell}$, produces an instance $(G',\mc{D}')$
of \dsnP, such that
\begin{itemize}
\item \textbf{\emph{(Completeness)}} If $\val(\Gamma) = 1$, then there exists a planar network $N \subseteq G'$ of cost $2(1 + \gamma^{1/5})$ that satisfies all demands.
\item \textbf{\emph{(Soundness)}} If $\val(\Gamma) < \gamma$, then every network $N \subseteq G'$ that satisfies all demands has cost more than $2(2 -
4\gamma^{1/5})$.
\item (Parameter Dependency) The number of demand pairs $k=|\mc{D}'|$ is $2\ell$.
\end{itemize}
\end{lemma}
Lemma~\ref{lem:dsnp-inapprox} is proven as follows: we construct the \dsnP instance in Section~\ref{subsubsubsec:construction}. The proofs of completeness and soundness of the reduction are deferred Section~\ref{subsubsubsec:completeness} and Section~\ref{subsubsubsec:soundness} respectively.
First, we construct a ``path gadget" which we use repeatedly in our construction.

\paragraph{Constructing a directed ``path'' gadget}
\label{sec:path-uniqueness-gadget}

For every integer $n$ we define the following gadget $\mathcal{P}_n$ which contains
$2n$ vertices (see \autoref{fig:path-uniqueness}). Since we need many of these
gadgets later on, we will denote vertices of $\mathcal{P}_n$ by $\mathcal{P}_n(v)$ etc., in order
to be able to distinguish vertices of different gadgets. All edges will have
the same weight $B$, which we will fix later during the reductions. The gadget
$\mathcal{P}_n$ is constructed as follows: $\mathcal{P}_n$ has a directed path
of one edge corresponding to each $i\in [n]$. This is given by
$\mathcal{P}_{n}(0_{i})\rightarrow \mathcal{P}_{n}(1_{i})$


%

\begin{figure}[h]

\centering
\begin{tikzpicture}[scale=0.3]


\foreach \j in {0,1,2,3,4}
{
\begin{scope}[shift={(0,5*\j)}]

\foreach \i in {0,1}
{
\draw [black] plot [only marks, mark size=8, mark=*] coordinates {(10*\i,0)};
}

\foreach \i in {0}
{
\path (\i, 0) node(a) {} (10+\i,0) node(b) {};
        \draw[very thick, black, middlearrow={>}] (a) -- (b);
}

\end{scope}
}

\draw [black] plot [only marks, mark size=0, mark=*] coordinates {(0,20)}
node[label={[xshift=-2mm,yshift=-1mm] $\mathcal{P}_{n}(0_1)$}] {} ;

\draw [black] plot [only marks, mark size=0, mark=*] coordinates {(10,20)}
node[label={[xshift=-2mm,yshift=-1mm] $\mathcal{P}_{n}(1_1)$}] {} ;

%

\draw [black] plot [only marks, mark size=0, mark=*] coordinates {(0,10)}
node[label={[xshift=-2mm,yshift=-1mm] $\mathcal{P}_{n}(0_i)$}] {} ;

\draw [black] plot [only marks, mark size=0, mark=*] coordinates {(10,10)}
node[label={[xshift=-2mm,yshift=-1mm] $\mathcal{P}_{n}(1_i)$}] {} ;

%

\draw [black] plot [only marks, mark size=0, mark=*] coordinates {(0,0)}
node[label={[xshift=-2mm,yshift=-1mm] $\mathcal{P}_{n}(0_n)$}] {} ;

\draw [black] plot [only marks, mark size=0, mark=*] coordinates {(10,0)}
node[label={[xshift=-2mm,yshift=-1mm] $\mathcal{P}_{n}(1_n)$}] {} ;

%

\end{tikzpicture}

\caption{The construction of the path gadget for $\mathcal{P}_n$. Note that the
gadget has $2n$ vertices. Each edge of $\mathcal{P}_n$ has the same weight $B$
 \label{fig:path-uniqueness}}
\end{figure}
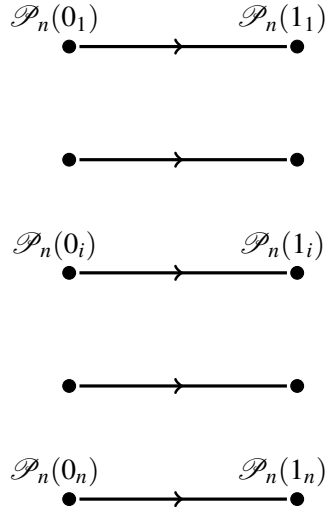

\paragraph{Construction of the \dsnP instance}
\label{subsubsubsec:construction}

We give a reduction which transforms an instance $G=(V,E)$ of \mcsik  into an
instance of \dsn which has $2\ell$ demand pairs and an optimum
which is planar. Let the partition of $V$ into color classes be given by $\{V_1,
V_2, \ldots, V_{\ell}, W_{1}, W_2, \ldots, W_\ell\}$. Without loss of generality
(by adding isolated vertices if necessary), we can assume that each color class
has the same number of vertices. Let $|V_i|=|W_i|=n'$ for each $1\leq
i\leq \ell$. Then $n=|V(G)|=2n'\ell$. For each $1\leq i,j\leq \ell$ we denote
by $E_{i,j}$ the set of edges with one end-point in $V_i$ and other in $W_j$.

We design two types of gadgets: the \emph{main gadget} and the \emph{secondary
gadget}. The reduction from \mcsik represents each edge set $E_{i,j}$ with
a main gadget $M_{i,j}$. This is done as follows: each main gadget is a copy
of the path gadget $\mathcal{P}_{|E_{i,j}|}$ from
Section~\ref{sec:path-uniqueness-gadget}
with $B=\frac{2}{\ell^2}$, i.e., there is a row in
$M_{i,j}$ corresponding to each edge in $E_{i,j}$. Each main gadget is
surrounded by four secondary gadgets:
on the top, right, bottom and left. Each of these gadgets are copies of the path
gadget from Section~\ref{sec:path-uniqueness-gadget} with $B=0$:
\begin{itemize}
  \item For each $1\leq i\leq \ell+1, 1\leq j\leq \ell$ the \emph{horizontal}
gadget $HS_{i,j}$ is a copy of $\mathcal{P}_{|W_j|}$
  \item For each $1\leq i\leq \ell, 1\leq j\leq \ell+1$ the \emph{vertical}
gadget $VS_{i,j}$ is a copy of $\mathcal{P}_{|V_i|}$
\end{itemize}

We refer to Figure~\ref{fig:dsnp-big-picture} (bird's-eye view) and Figure~\ref{fig:dsnp-zoomed-in} (zoomed-in view) for an illustration of the reduction.
Fix some $1\leq i,j\leq \ell$. The main gadget $M_{i,j}$ has four secondary
gadgets surrounding it:
\begin{itemize}
\item Above $M_{i,j}$ is the vertical secondary gadget $VS_{i,j+1}$
\item On the right of $M_{i,j}$ is the horizontal secondary gadget $HS_{i+1,j}$
\item Below $M_{i,j}$ is the vertical secondary gadget $VS_{i,j}$
\item On the left of $M_{i,j}$ is the horizontal secondary gadget $HS_{i,j}$
\end{itemize}
Hence, there are $\ell(\ell+1)$ horizontal secondary gadgets and $\ell(\ell+1)$ vertical secondary gadgets.

\textbf{Red intra-gadget edges}: Fix $(i,j)$ such that $1\leq i,j\leq \ell$. Recall that $M_{i,j}$ is a copy of $\mathcal{P}_{|E_{i,j}|}$ with $B=\frac{2}{\ell^2}$ and each of the secondary
gadgets are copies of $\mathcal{P}_{n'}$ with $B=0$. With slight abuse of
notation, we assume that the rows of $M_{i,j}$ are indexed by the set $\{(x,y)\
:\ (x,y)\in E_{i,j}, x\in W_i, y\in V_j\}$.
We add the following edges (in \textcolor[rgb]{1.00,0.00,0.00}{red} color) of
weight $0$: for each $(x,y)\in E_{i,j}$
\begin{itemize}
\item Add the edge $VS_{i,j+1}(1_x)\rightarrow M_{i,j}(0_{(x,y)})$. These edges
are called \emph{top-red} edges incident on $M_{i,j}$.
\item Add the edge $HS_{i,j}(1_y) \rightarrow M_{i,j}(0_{(x,y)})$. These edges
are called \emph{left-red} edges incident on $M_{i,j}$.
\item Add the edge $M_{i,j}(1_{(x,y)}) \rightarrow HS_{i+1,j}(0_y)$. These
edges are called \emph{right-red} edges incident on $M_{i,j}$.
\item Add the edge $M_{i,j}(1_{(x,y)}) \rightarrow VS_{i,j}(0_x)$. These edges
are called \emph{bottom-red} edges incident on $M_{i,j}$.
\end{itemize}
These are called the \emph{intra-gadget} edges incident on $M_{i,j}$.

Introduce the following $4\ell$ vertices (which we call \emph{border} vertices):
\begin{itemize}
\item $a_1, a_2, \ldots, a_\ell$
\item $b_1, b_2, \ldots, b_\ell$
\item $c_1, c_2, \ldots, c_\ell$
\item $d_1, d_2, \ldots, d_\ell$
\end{itemize}
\textbf{Orange edges}: For each $i\in [\ell]$ add the following edges (shown as \textcolor{orange}{orange} in
Figure~\ref{fig:dsnp-big-picture}) with weight~$\frac{2\gamma^{1/5}}{4\ell}$:
\begin{itemize}
\item $a_{i} \rightarrow VS_{i,\ell+1}(0_{v})$ for each $v\in V_i$. These are
called \emph{top-orange} edges.
\item $VS_{i,1}(1_{v}) \rightarrow b_{i}$ for each $v\in V_i$. These are called
\emph{bottom-orange} edges.
\item $c_{j} \rightarrow HS_{1,j}(0_{w})$ for each $w\in W_j$. These are called
\emph{left-orange} edges.
\item $HS_{\ell+1,j}(1_{w}) \rightarrow d_{j}$ for each $w\in W_j$. These are
called \emph{right-orange} edges.
\end{itemize}

%

Finally, the set of demand pairs $\mathcal{D'}$ is given by:
\begin{itemize}
\item \underline{Type I}: the pairs $(a_i, b_i)$ for each $1\leq i\leq \ell$.

\item \underline{Type II}: the pairs $(c_j, d_j)$ for each $1\leq j\leq \ell$.
\end{itemize}


Clearly, the total number of demand pairs is $k=|\mathcal{D'}|=2\ell$. Let the final graph constructed be $G'$. Note that $G'$ has size $N=(n+\ell)^{O(1)}$ and can be constructed in $(n+\ell)^{O(1)}$ time. It is also easy to see that $G'$ is actually a DAG.

 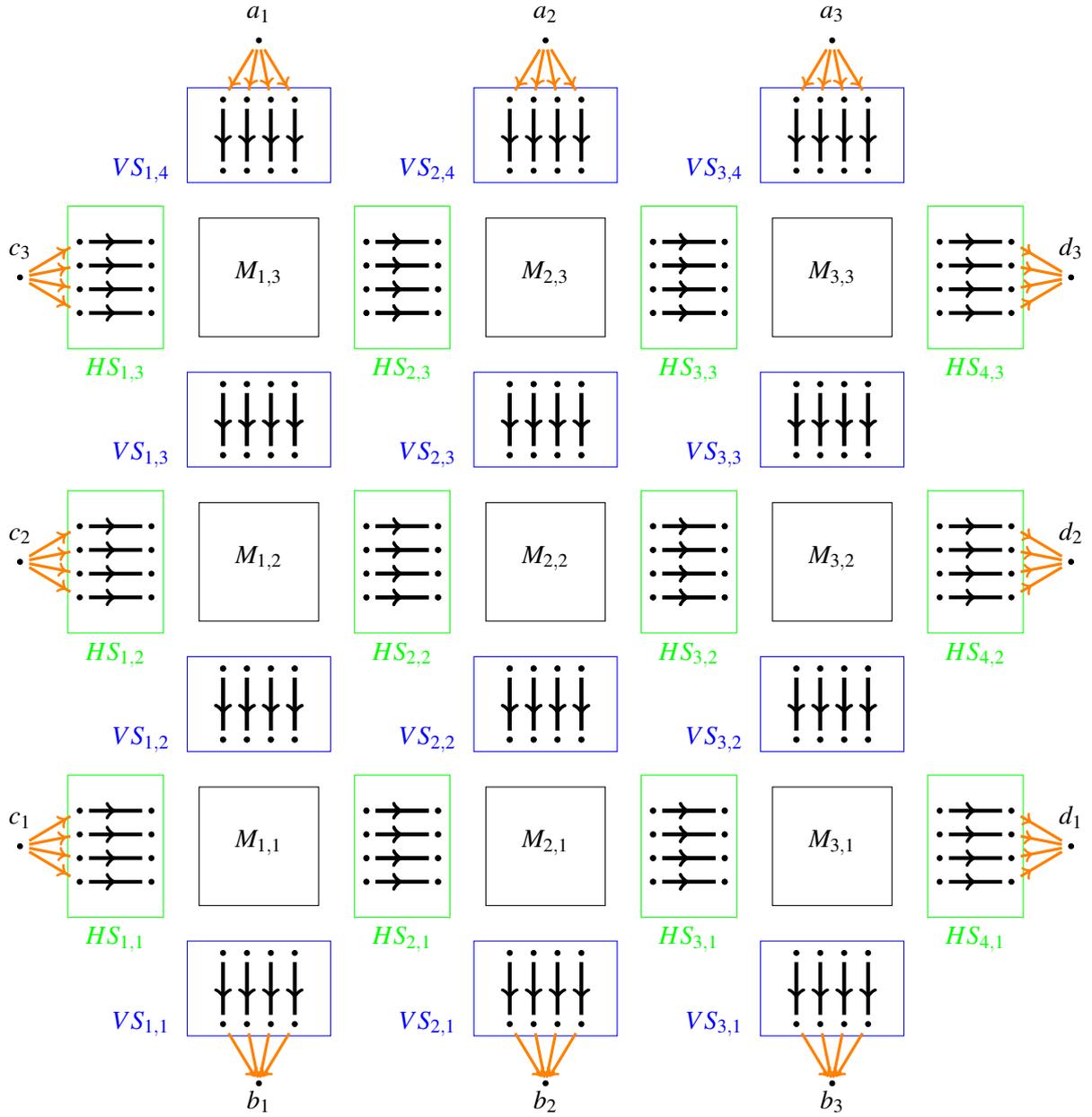
\begin{figure}[hp]

 \centering
 \begin{tikzpicture}[scale=0.35]

 \foreach \j in {0,1,2}
 {
 \begin{scope}[shift={(0,12*\j)}]

 \foreach \j in {0,1,2}
 {
 \begin{scope}[shift={(12*\j,0)}]

     \draw[rectangle] (0,0) rectangle (5,5);

     \foreach \j in {0,1}
     {
     \begin{scope}[shift={(-12*\j,0)}]
     \foreach \j in {0,1,2,3}
     {
     \begin{scope}[shift={(0,\j)}]
     \foreach \i in {0,3}
    \draw [black] plot [only marks, mark size=3, mark=*] coordinates {(7+\i,1)};
     \end{scope}
     }
     \end{scope}
     }

     \foreach \j in {0,1}
     {
     \begin{scope}[shift={(-12*\j,0)}]
     \foreach \j in {0,1,2,3}
     {
     \begin{scope}[shift={(0,\j)}]
     \foreach \i in {0}
     {
     \path (7+\i, 1) node(a) {} (7+\i+3, 1) node(b) {};
         \draw[ultra thick, black,middlearrow={>}] (a) -- (b);
     }

     \end{scope}
     }
     \end{scope}
     }

     \foreach \j in {0,1}
     {
     \begin{scope}[shift={(0,-12*\j)}]
     \foreach \j in {0,1,2,3}
     {
     \begin{scope}[shift={(\j,0)}]
     \foreach \i in {0,3}
    \draw [black] plot [only marks, mark size=3, mark=*] coordinates {(1,7+\i)};
     \end{scope}
     }
     \end{scope}
     }

     \foreach \j in {0,1}
     {
     \begin{scope}[shift={(0,-12*\j)}]
     \foreach \j in {0,1,2,3}
     {
     \begin{scope}[shift={(\j,0)}]
     \foreach \i in {0}
     {
     \path (1,7+\i) node(a) {} (1,7+\i+3) node(b) {};
         \draw[ultra thick, black,middlearrow={<}] (a) -- (b);
     }
     \end{scope}
     }
     \end{scope}
     }
 \end{scope}
 }

 \end{scope}
 }


 \foreach \i in {0,1,2}
 {
 \begin{scope}[shift={(0,12*\i)}]

 \foreach \j in {-1,0,1,2}
     {
     \begin{scope}[shift={(12*\j,0)}]

 \draw[green] (6.5,-0.5) rectangle (10.5,5.5) ;





 \end{scope}
 }
 \end{scope}
 }


 \foreach \i in {0,1,2}
 {
 \begin{scope}[shift={(12*\i,0)}]

 \foreach \j in {0,1,2,3}
     {
     \begin{scope}[shift={(0,12*\j)}]

 \draw[blue] (-0.5,-5.5) rectangle (5.5,-1.5) ;





 \end{scope}
 }
 \end{scope}
 }


 \foreach \i in {1,2,3}
 {
 \foreach \j in {1,2,3}
 {
 \draw [black] plot [only marks, mark size=0, mark=*] coordinates
{(-9.5+12*\i,-9.5+12*\j)}
 node[label={[xshift=0mm,yshift=-4mm] \textbf{$M_{\i,\j}$}}] {} ;
 }
 }


 \draw [black] plot [only marks, mark size=3, mark=*] coordinates {(-7.5,2.5)}
node[label={[xshift=0mm,yshift=0mm] $c_1$}] {};
 \draw [black] plot [only marks, mark size=3, mark=*] coordinates {(-7.5,14.5)}
node[label={[xshift=0mm,yshift=0mm] $c_2$}] {};
 \draw [black] plot [only marks, mark size=3, mark=*] coordinates {(-7.5,26.5)}
node[label={[xshift=0mm,yshift=0mm] $c_3$}] {};

 \foreach \i in {0,1,2}
 {
 \begin{scope}[shift={(0,12*\i)}]
 \foreach \i in {0,1,2,3}
 {
 \path (-7.5,2.5) node(a) {} (-5,\i+1) node(b) {};
         \draw[very thick,orange, endarrow={>}] (a) -- (b);
 }
 \end{scope}
 }

 \draw [black] plot [only marks, mark size=3, mark=*] coordinates {(36.5,2.5)}
node[label={[xshift=0mm,yshift=0mm] $d_1$}] {};
 \draw [black] plot [only marks, mark size=3, mark=*] coordinates {(36.5,14.5)}
node[label={[xshift=0mm,yshift=0mm] $d_2$}] {};
 \draw [black] plot [only marks, mark size=3, mark=*] coordinates {(36.5,26.5)}
node[label={[xshift=0mm,yshift=0mm] $d_3$}] {};

 \foreach \i in {0,1,2}
 {
 \begin{scope}[shift={(0,12*\i)}]
 \foreach \i in {0,1,2,3}
 {
 \path (36.5,2.5) node(a) {} (34,\i+1) node(b) {};
         \draw[very thick,orange,endarrow={<}] (a) -- (b);
 }
 \end{scope}
 }

 \draw [black] plot [only marks, mark size=3, mark=*] coordinates {(2.5,-7.5)}
node[label={[xshift=0mm,yshift=-7mm] $b_1$}] {};
 \draw [black] plot [only marks, mark size=3, mark=*] coordinates {(14.5,-7.5)}
node[label={[xshift=0mm,yshift=-7mm] $b_2$}] {};
 \draw [black] plot [only marks, mark size=3, mark=*] coordinates {(26.5,-7.5)}
node[label={[xshift=0mm,yshift=-7mm] $b_3$}] {};

 \foreach \i in {0,1,2}
 {
 \begin{scope}[shift={(12*\i,0)}]
 \foreach \i in {0,1,2,3}
 {
 \path (2.5,-7.5) node(a) {} (\i+1,-5) node(b) {};
         \draw[very thick,orange ,startarrow={<}] (a) -- (b);
 }
 \end{scope}
 }

 \draw [black] plot [only marks, mark size=3, mark=*] coordinates {(2.5,36.5)}
node[label={[xshift=0mm,yshift=0mm] $a_1$}] {};
 \draw [black] plot [only marks, mark size=3, mark=*] coordinates {(14.5,36.5)}
node[label={[xshift=0mm,yshift=0mm] $a_2$}] {};
 \draw [black] plot [only marks, mark size=3, mark=*] coordinates {(26.5,36.5)}
node[label={[xshift=0mm,yshift=0mm] $a_3$}] {};

 \foreach \i in {0,1,2}
 {
 \begin{scope}[shift={(12*\i,0)}]
 \foreach \i in {0,1,2,3}
 {
 \path (2.5,36.5) node(a) {} (\i+1,34) node(b) {};
         \draw[orange, very thick, endarrow={>}] (a) -- (b);
 }
 \end{scope}
 }


 \foreach \j in {1,2,3}
 {
 \foreach \i in {1,2,3,4}
 {
 \draw [green] plot [only marks, mark size=0, mark=*] coordinates
{(-15.5+12*\i,-13+12*\j)} node[label={[xshift=0mm,yshift=-6mm] \textbf{$HS_{\i,\j}$}}]
{};

 }
 }

 \foreach \j in {1,2,3,4}
 {
 \foreach \i in {1,2,3}
 {
 \draw [blue] plot [only marks, mark size=0, mark=*] coordinates
{(-11+12*\i,-18+12*\j)} node[label={[xshift=-12mm,yshift=-1mm] \textbf{$VS_{\i,\j}$}}]
{};
 }
 }

 \end{tikzpicture}

 \caption{A bird's-eye view of the instance of $G'$ with $\ell=3$ and $n'=4$ (see
Figure~\ref{fig:dsnp-zoomed-in} for a zoomed-in view).
Additionally we have some red edges between each main gadget and the four
secondary gadgets surrounding it which are omitted in this figure for clarity
(they are shown in Figure~\ref{fig:dsnp-zoomed-in} which
gives a more zoomed-in view).
  \label{fig:dsnp-big-picture}}
 \end{figure}

 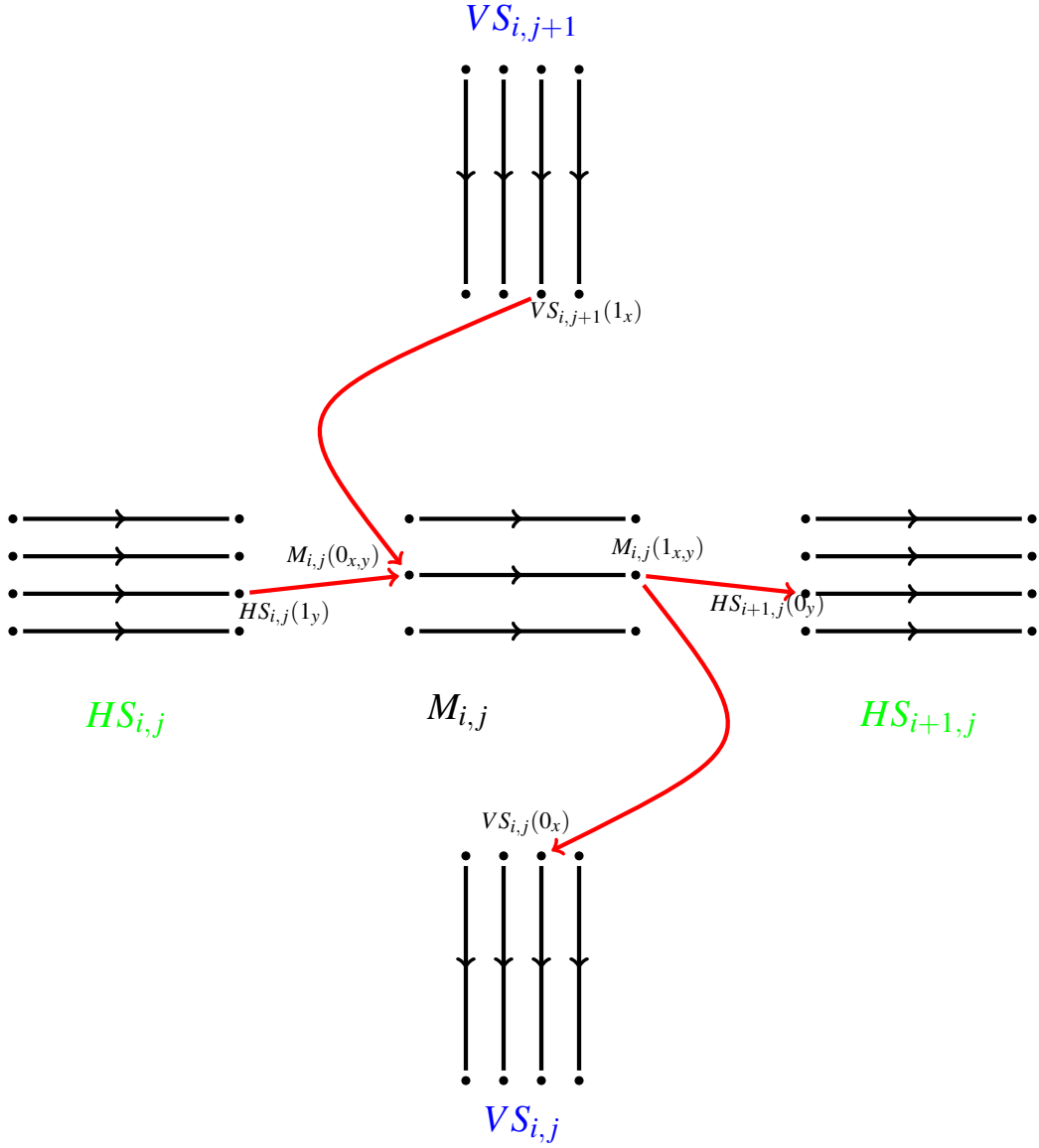
\begin{figure}[hp]

 \centering
 \begin{tikzpicture}[scale=0.5]



 \foreach \j in {0,1}
 {
 \begin{scope}[shift={(-21*\j,0)}]

 \foreach \j in {0,1,2,3}
 {
 \begin{scope}[shift={(0,\j)}]
 \foreach \i in {0,3}
 {
 \draw [black] plot [only marks, mark size=3, mark=*] coordinates {(12+2*\i,3)};
 }

 \foreach \i in {0}
 {
 \path (12+2*\i,3) node(a) {} (12+2*\i+6,3) node(b) {};
         \draw[ultra thick,black, middlearrow={>}] (a) -- (b);
 }
 \end{scope}
 }



%
%
%

%
%
%

 \end{scope}
 }


 \foreach \j in {0,1}
 {
 \begin{scope}[shift={(0,21*\j)}]

 \foreach \j in {0,1,2,3}
 {
 \begin{scope}[shift={(\j,0)}]
 \foreach \i in {0,3}
 {
 \draw [black] plot [only marks, mark size=3, mark=*] coordinates {(3,-3-2*\i)};
 }

 \foreach \i in {0}
 {
 \path (3,-3-2*\i) node(a) {} (3,-3-2*\i-6) node(b) {};
         \draw[ultra thick,black,middlearrow={>}] (a) -- (b);
 }
 \end{scope}
 }

 \end{scope}
 }



 \foreach \j in {0,1,2}
 {
 \begin{scope}[shift={(0,1.5*\j)}]
 \foreach \i in {0,3}
 {
\draw [black] plot [only marks, mark size=3, mark=*] coordinates {(1.5+2*\i,3)};
 }

 \foreach \i in {0}
 {
 \path (1.5+2*\i,3) node(a) {} (1.5+2*\i+6,3) node(b) {};
         \draw[ultra thick,black,middlearrow={>}] (a) -- (b);
 }
 \end{scope}
 }


%
%
%
%
%
%
%
%
%
%


 \path (1.5,4.5) node(a) {} (-3,4) node(b) {};
  \draw [ultra thick,red,->] (b) to (a);

 \path (7.5,4.5) node(a) {} (12,4) node(b) {};
  \draw [ultra thick,red,->] (a) to (b);

 \path (1.5,4.5) node(a) {} (5,12) node(b) {};
  \draw [ultra thick,red,->] (b) .. controls (-2,9) .. (a);

 \path (7.5,4.5) node(a) {} (5,-3) node(b) {};
  \draw [ultra thick,red,->] (a) .. controls (11,0) .. (b);



 \draw [green] plot [only marks, mark size=0, mark=*] coordinates {(-6,1)}
node[label={[xshift=0mm,yshift=-7mm] \Large{$HS_{i,j}$}}] {};

 \draw [green] plot [only marks, mark size=0, mark=*] coordinates {(15,1)}
node[label={[xshift=0mm,yshift=-7mm] \Large{$HS_{i+1,j}$}}] {};

 \draw [blue] plot [only marks, mark size=0, mark=*] coordinates {(4.5,-10)}
node[label={[xshift=0mm,yshift=-6mm] \Large{$VS_{i,j}$}}] {};

 \draw [blue] plot [only marks, mark size=0, mark=*] coordinates {(4.5,19)}
node[label={[xshift=0mm,yshift=-4mm] \Large{$VS_{i,j+1}$}}] {};

 \draw [black] plot [only marks, mark size=0, mark=*] coordinates {(2,0)}
node[label={[xshift=4mm,yshift=-1mm] \Large{$M_{i,j}$}}] {};

 \draw [black] plot [only marks, mark size=0, mark=*] coordinates {(1.5,4.5)}
node[label={[xshift=-10mm,yshift=-2mm] \footnotesize{$M_{i,j}(0_{x,y})$}}] {};

 \draw [black] plot [only marks, mark size=0, mark=*] coordinates {(7.5,4.5)}
node[label={[xshift=3mm,yshift=-1mm] \footnotesize{$M_{i,j}(1_{x,y})$}}] {};

 \draw [black] plot [only marks, mark size=0, mark=*] coordinates {(5,12)}
node[label={[xshift=6mm,yshift=-7mm] \footnotesize{$VS_{i,j+1}(1_{x})$}}] {};

 \draw [black] plot [only marks, mark size=0, mark=*] coordinates {(-3,4)}
node[label={[xshift=6mm,yshift=-7mm] \footnotesize{$HS_{i,j}(1_{y})$}}] {};

 \draw [black] plot [only marks, mark size=0, mark=*] coordinates {(5,-3)}
node[label={[xshift=-2mm,yshift=0mm] \footnotesize{$VS_{i,j}(0_{x})$}}] {};

 \draw [black] plot [only marks, mark size=0, mark=*] coordinates {(12,4)}
node[label={[xshift=-5mm,yshift=-6mm] \footnotesize{$HS_{i+1,j}(0_{y})$}}] {};

 \end{tikzpicture}

 \caption{A zoomed-in view of the main gadget $M_{i,j}$ surrounded by four
secondary gadgets: vertical gadget $VS_{i,j+1}$ on the top, horizontal gadget
$HS_{i,j}$ on the left, vertical gadget $VS_{i,j}$ on the bottom and horizontal
gadget $HS_{i+1,j}$ on the right. Each of the secondary gadgets is a copy of
the uniqueness gadget $U_n$ (see Section~\ref{sec:path-uniqueness-gadget}) and the main
gadget $M_{i,j}$ is a copy of the uniqueness gadget $U_{|S_{i,j}|}$. The only
inter-gadget edges are the red edges: they have one end-point in a main gadget
and the other end-point in a secondary gadget. We have shown four such red
edges which are introduced for every $(x,y)\in E_{i,j}$.
  \label{fig:dsnp-zoomed-in}}
 \end{figure}

\paragraph{Completeness of the reduction in Lemma~\ref{lem:dsnp-inapprox}}
\label{subsubsubsec:completeness}

If $\val(\Gamma) = 1$, then there exist $(v_1, v_2,\dots,
v_\ell) \in V_1 \times V_2 \times \cdots \times V_\ell$ and $(w_1, w_2,\dots,
w_\ell) \in W_1 \times W_2 \times \cdots \times W_\ell$ that induce a
$K_{\ell,\ell}$. We now build a planar solution $N$ for the \dsnP instance
$(G',\mathcal{D}')$ as follows:
\begin{itemize}
  \item For each $i\in [\ell]$ pick the orange edges $a_i \rightarrow
VS_{i,\ell+1}(0_{v_i})$ and $VS_{i,1}(1_{v_i}) \rightarrow b_i$.
  \item For each $j\in [\ell]$ pick the orange edges $c_j \rightarrow
HS_{1,j}(0_{w_j})$ and $HS_{\ell+1,j}(1_{w_j}) \rightarrow d_j$.
  \item For each $1\leq i,j\leq \ell$ pick the black edge $MS_{i,j}(0_{v_i,w_j}) \rightarrow MS_{i,j}(1_{v_i,w_j})$ in the main gadget. This edge is guaranteed to exist since $v_i - w_j \in E_{i,j}$. Also pick the four red edges with one endpoint in $M_{i,j}$ given by $VS_{i,j+1}(1_{v_i})\rightarrow MS_{i,j}(0_{v_i,w_j}), HS_{i,j}(1_{w_j})\rightarrow MS_{i,j}(0_{v_i,w_j}), MS_{i,j}(1_{v_i,w_j}) \rightarrow HS_{i+1,j}(0_{w_j}) $ and $MS_{i,j}(1_{v_i,w_j}) \rightarrow VS_{i,j}(0_{v_i})$.
  \item For each $1\leq i\leq \ell, 1\leq j\leq \ell+1$ pick the edge $VS_{i,j}(0_{v_i})\rightarrow VS_{i,j}(1_{v_i})$ .
  \item For each $1\leq i\leq \ell+1, 1\leq j\leq \ell$ pick the edge $HS_{i,j}(0_{w_j})\rightarrow HS_{i,j}(1_{w_j})$.
\end{itemize}

Note that red edges and edges in secondary gadgets have weight 0. Hence, the
weight of $N$ is $4\ell\cdot (\frac{2\gamma^{1/5}}{4\ell}) + \ell^{2}\cdot
(\frac{2}{\ell^2}) = 2(1+\gamma^{1/5})$ since we pick $4\ell$ orange edges and
one black edge from each of the $\ell^2$ main gadgets.

We next show that $N$ satisfies all the demand pairs. Consider the pair $(c_j, d_j)$ for some $j\in
[\ell]$. There is a $c_j \leadsto d_j$ path as follows: start with the edge $c_j
\rightarrow HS_{1,j}(0_{w_j})$. Then for each $1\leq i\leq \ell$ use the
following edges in this order:
\begin{itemize}
  \item Traverse the gadget $HS_{i,j}$ using the edge $HS_{i,j}(0_{w_j}) \rightarrow HS_{i,j}(1_{w_j})$
  \item Reach the main gadget $MS_{i,j}$ using the edge $HS_{i,j}(1_{w_j}) \rightarrow MS_{i,j}(0_{v_i,w_j})$
  \item Traverse the main gadget $MS_{i,j}$ using the edge $MS_{i,j}(0_{v_i,w_j})\rightarrow MS_{i,j}(1_{v_i,w_j})$
  \item Reach the gadget $HS_{i+1,j}$ using the edge $MS_{i,j}(1_{v_i,w_j}) \rightarrow HS_{i+1,j}(0_{w_j})$
\end{itemize}
This way we have reached the vertex $HS_{\ell+1, j}(0_{w_j})$. Finally use the two edges $HS_{\ell+1, j}(0_{w_j})\rightarrow HS_{\ell+1, j}(1_{w_j})$ and $HS_{\ell+1, j}(1_{w_j})\rightarrow d_j$. The proof for the satisfiability of $a_i\leadsto b_i$ paths is similar.

Finally, we show that $N$ is planar. It is easy to see that removing the red edges from $G'$ leads to
a planar graph (see Figure~\ref{fig:dsnp-big-picture} for a planar embedding of
this graph). It remains to show that the red edges we add in $N$ do not destroy
planarity. For any main gadget $M_{i,j}$: the only red edges from $G'$ which are
added in $N$ are as follows: one left-red edge and one top-red edge incident on
the same $0$-vertex of $M_{i,j}$ and one bottom-red edge and one right-red edge
incident on the same $1$-vertex of $M_{i,j}$. This can clearly be done while
preserving planarity: the only 4 red edges retained in $N$ are shown as in
Figure~\ref{fig:dsnp-zoomed-in} (note that Figure~\ref{fig:dsnp-zoomed-in} is
actually supposed to have many more red edges which are omitted for clarity).

\paragraph{Soundness of the reduction in Lemma~\ref{lem:dsnp-inapprox}}
\label{subsubsubsec:soundness}

Our soundness proof will be by contrapositive. Suppose that there exists a planar solution $N$ of cost $\rho \leq 2(2 - 4\gamma^{1/5})$ that satisfies all the demand pairs. We first define the following sets:
\begin{itemize}
\item $L_j := \{w\in W_j\ :\ c_j \rightarrow HS_{1,j}(0_{w})\in E(N)\}$ for
each $j\in [\ell]$
\item $L := \cup_{j=1}^{\ell} L_j$

\item $R_j := \{w\in W_j\ :\ HS_{\ell+1,j}(1_{w}) \rightarrow d_j\in E(N)\}$
for each $j\in [\ell]$
\item $R := \cup_{j=1}^{\ell} R_j$

\item $T_i := \{v\in V_i\ :\ a_i \rightarrow VS_{i,\ell+1}(0_{v})\in E(N)\}$
for each $i\in [\ell]$
\item $T := \cup_{i=1}^{\ell} T_i$

\item $B_i := \{v\in V_i\ :\ VS_{i,1}(1_{v}) \rightarrow b_i\in E(N)\}$ for
each $i\in [\ell]$
\item $B := \cup_{i=1}^{\ell} B_i$

\item $\cH_{i,j} := \{MS_{i,j}(0_{x,y}) \rightarrow MS_{i,j}(1_{x,y})\in E(N)\
:\ x\in V_i, y\in W_j, x-y\in E_{i,j}\}$ for each $1\leq i,j\leq \ell$
\end{itemize}
Let $\alpha_W = |L|+|R|$ and $\alpha_V = |T| + |B|$. Since each orange edge has
weight $\frac{2\gamma^{1/5}}{4\ell}$ it follows that
$$ \max\{\alpha_V, \alpha_W\}\leq \alpha_W + \alpha_V  \leq
\frac{\rho}{(2\gamma^{1/5}/4\ell)} \leq
\frac{4}{(2\gamma^{1/5}/4\ell)} \leq 8\ell\gamma^{-1/5}$$

\begin{claim}
For each $1\leq i,j\leq \ell$ we have
\begin{itemize}
  \item $L_j\cap R_j \neq \emptyset$
  \item $T_i\cap B_i \neq \emptyset$
  \item $\cH_{i,j} \neq \emptyset$
\end{itemize}
\label{claim:dsnp-non-empty-sets}
\end{claim}
\begin{proof}
Fix any $j\in [\ell]$, and let $P$ be any $c_j\leadsto d_j$ path in $N$. By orientation of the edges of $G'$, it is easy to see that $P$ cannot contain any vertex of the vertical gadget $VS_{i,j+1}$ for any $i\in [\ell]$. Also, $P$ cannot contain any vertex of a vertical gadget $VS_{i,j}$ for any $i\in [\ell]$: due to the orientation of the edges, there is no path from $VS_{i,j}$ to $d_j$. Hence, each internal (i.e., non-orange) edge of $P$ has both end-points in the vertex set given by $(\bigcup_{i=1}^{\ell+1} HS_{i,j}) \cup (\bigcup_{i=1}^{\ell} MS_{i,j}$).

Since $P$ is a $c_j\leadsto d_j$ path, the first edge of $P$ is an orange edge
incident on $c_j$. Let this edge be to the vertex $HS_{1,j}(0_{w^*})$ for some
$w^* \in W_j$. Recall that edges which have one endpoint in a main gadget and
other the end-point in a horizontal gadget (see
Figure~\ref{fig:dsnp-zoomed-in}) are of one of the following two types:
\begin{itemize}
\item For each $1\leq i\leq \ell$, we have the edge $HS_{i,j}(1_y)\rightarrow MS_{i,j}(0_{x,y})$ for some $y\in W_j, x\in V_i$ such that $x-y\in E_{i,j}$

\item For each $1\leq i\leq \ell$, we have the edge $MS_{1_{x,y}}\rightarrow HS_{i+1,j}(0_y)$ for some $y\in W_j, x\in V_i$ such that $x-y\in E_{i,j}$
\end{itemize}
Moreover, every $0$-vertex of a horizontal gadget or main gadget has exactly
one out-neighbor, which is its corresponding $1$-vertex. Hence, each edge in
$P$ with both end-points in a horizontal gadget must correspond to $w^*
\in W_j$, which implies that the last edge of $P$ must be
$HS_{\ell+1,j}(1_{w^*})\rightarrow d_j$. Therefore, $w^* \in L_j\cap R_j$, and
so $L_j\cap R_j\neq \emptyset$. Similarly, for each $i\in [\ell]$ we have
$T_i\cap B_i \neq \emptyset$.

The argument above shows that any $c_j\leadsto d_j$ path in $G'$ uses an edge from each $M_{i,j}$ for $1\leq i\leq \ell$. Hence, $\cH_{i,j} \neq \emptyset$ for each $1\leq i,j\leq \ell$.
\end{proof}
%
Next, recall that each edge in a main gadget has weight $\frac{2}{\ell^2}$. Since $N$ has cost $\rho$, we have
$$
\sum_{1 \leq i,j \leq \ell} |\cH_{i, j}| \leq \frac{\ell^2}{2} \cdot \rho \hspace{3mm} \Rightarrow \hspace{3mm}
\sum_{1 \leq i,j \leq \ell} (|\cH_{i, j}| - 1) \leq \frac{\ell^2}{2} \cdot \left(\rho -
2\right)
$$
Since $\cH_{i, j} \ne \emptyset$ for every $1 \leq i, j \leq \ell$, the above
inequality implies that, for at least $\ell^2-\frac{\ell^2}{2} \cdot
(\rho-2)=\frac{\ell^2}{2} \cdot (4 - \rho)
\geq 8\gamma^{1/5}\frac{\ell^2}{2} = 4\gamma^{1/5}\ell^2$ pairs of $(i, j)$'s we have $|\cH_{i, j}| = 1$. Let
$\cP_{\uni}$ be the set of all such pairs of $(i, j)$'s.

We will argue that a random assignment defined from picking one vertex from each $L_j\cap R_j$ (for $j\in [\ell]$) uniformly
independently at random and one vertex from each $T_i\cap B_i$ (for $i\in [\ell]$) uniformly
independently at random covers many superedges in expectation. To do this, we
need to first show that, for many $(i, j)$'s, there exist $x \in (T_i\cap B_i)$
and $y \in (L_j\cap R_j)$ such that $x-y \in E_G$. In fact, we can show this
for every $(i, j) \in \cP_{\uni}$ as stated below.

\begin{claim} \label{claim:dsnp-edgeins}
Let $1\leq i,j\leq \ell$. If $(i, j) \in \cP_{\uni}$, then there exists $x \in (T_i\cap B_i)$ and $y \in (L_j\cap R_j)$ such that $x-y \in E_{i,j}\subseteq E_G$.
\end{claim}
\begin{proof}
Since  $(i, j) \in \cP_{\uni}$, the solution $N$ contains exactly one edge from the main gadget $M_{i, j}$. Let this edge be $M_{i,j}(0_{x,y})\rightarrow M_{i,j}(1_{x,y})$. We now claim that $y\in (L_j\cap R_j)$ and $x\in (T_i\cap B_i)$. This would complete the proof since the existence of the edge $M_{i,j}(0_{x,y})\rightarrow M_{i,j}(1_{x,y})$ implies $x-y\in E_{i,j}$.

We now show that $y\in L_j\cap R_j$. The proof for $x\in T_i\cap B_i$ is
similar. Recall that in the proof of Claim~\ref{claim:dsnp-non-empty-sets}, we
have (implicitly) shown that the only set of edges from $M_{i,j}$ that can be
used as part of the solution $N$ is a subset of the set
$\{M_{i,j}(0_{v,w})\rightarrow M_{i,j}(1_{v,w})\ :\ v\in V_i, w\in W_j, v-w \in
E_{i,j}\}$. Since the only edge from $M_{i,j}$ in the solution $N$ is
$M_{i,j}(0_{x,y})\rightarrow M_{i,j}(1_{x,y})$, it follows that $y\in L_j$.
Similarly, $y\in R_j$ and hence $y\in L_j\cap R_j$.
\end{proof}

Now, let $\phi: [2\ell]=V_H \to V_G$ be a random assignment obtained as follows:
\begin{itemize}
  \item For each $1\leq i\leq \ell$, choose $\phi(i)$ independently uniformly at random from $T_i\cap B_i$
  \item For each $1\leq j\leq \ell$, choose $\phi(j+\ell)$ independently uniformly at random from $L_j\cap R_j$
\end{itemize}

By Claim~\ref{claim:dsnp-edgeins}, for every $(i, j) \in \cP_{\uni}$, there exists $v \in (T_i\cap B_i)$ and $w \in (L_j\cap R_j)$ such that $v-w \in E_{i,j}\subseteq E_G$. This means that, for such $(i, j)$, the probability that the superedge $i-j \in E_H$ is covered is at least the probability that $\phi(i) = v$ and $\phi(j) = w$, which is equal to $\frac{1}{|T_i\cap B_i|\cdot |L_j\cap R_j|}$. As a result, the expected number of superedges covered by $\phi$ is at least
\begin{align*}
&\sum_{(i, j) \in \cP_{\uni}} \frac{1}{|T_i\cap B_i|\cdot |L_j\cap R_j|} \\
&\geq \sum_{(i, j) \in \cP_{\uni}} \frac{1}{|T_i|\cdot |L_j|} \\
&\geq \frac{|\cP_{\uni}|^3}{\left(\sum_{(i, j) \in \cP_{\uni}} |T_i|\right)\left(\sum_{(i, j) \in \cP_{\uni}} |L_j|\right)} \tag*{(From H\"{o}lder's inequality)} \\
&\geq \frac{|\cP_{\uni}|^3}{(\ell\cdot |T|)\cdot (\ell\cdot |L|)}  \tag*{(Since
$T=\cup_{i=1}^{\ell} T_i$ and $L=\cup_{j=1}^{\ell} L_j$)}\\
&\geq \frac{|\cP_{\uni}|^3}{\ell^2 \cdot (|T|+|B|)\cdot (|L|+|R|)} \\
&= \frac{|\cP_{\uni}|^3}{\ell^2 \cdot (\alpha_V)\cdot (\alpha_W)} \tag*{(Since $\alpha_V = |T|+|B|$ and $\alpha_W=|L|+|R|$)} \\
&\geq \frac{|\cP_{\uni}|^3}{\ell^2 \cdot (8\ell\gamma^{-1/5})\cdot (8\ell\gamma^{-1/5})} \tag*{(Since $\max\{\alpha_V,\alpha_W\}\leq 8\ell\gamma^{-1/5})$)} \\
&\geq \frac{(4\gamma^{1/5}\ell^2)^3}{\ell^2 \cdot (8\ell\gamma^{-1/5})\cdot (8\ell\gamma^{-1/5})} \tag*{(Since $|\cP_{\uni}| \geq 4\gamma^{1/5}\ell^2$)}\\
&\geq \ell^{2} \cdot \gamma,
\end{align*}
Hence, there exists an assignment of $\Gamma$ with value at least $\gamma$, which implies that $\val(\Gamma) \geq \gamma$. This completes the proof of Lemma~\ref{lem:dsnp-inapprox}.

\subsubsection{Finishing the proof of Theorem~\ref{thm:lb-approx-DSNP}}

We can now easily prove Theorem~\ref{thm:lb-approx-DSNP} by combining
Lemma~\ref{lem:dsnp-inapprox} and Corollary~\ref{crl:inapprox-colored-biclique}.

\begin{proof}[Proof of Theorem~\ref{thm:lb-approx-DSNP}]
We again prove by contrapositive. Suppose that, for some constant $\varepsilon
> 0$ and for some function $f(k)$ independent of $n$, there exists an $f(k)\cdot
N^{O(1)}$-time $(2 - \varepsilon)$-approximation algorithm for \dsnP where $k$
is the number of terminal pairs and $N$ is the size of the instance. Let us
call this algorithm $\bA$.

Given $\epsilon>0$, it is easy to see that there exists a sufficiently small $\gamma^* =
\gamma^*(\varepsilon)$ such that $\frac{2(2 - 4{\gamma^*}^{1/5)}}{2(1 +{\gamma^*}^{1/5})} \geq (2 - \varepsilon)$. We create an algorithm $\bB$ that can
distinguish between the two cases of Corollary~\ref{crl:inapprox-colored-biclique} with
$h(\ell) = 1- \frac{\log (1/\gamma^*)}{\log \ell}=o(1)$. Our new algorithm $\bB$ works as follows.
Given an instance $(G, H, V_1 \cup \cdots \cup V_\ell, W_1 \cup \cdots \cup W_\ell)$ of \pname{MCSI} of size $n$ where
$H=K_{\ell,\ell}$, the algorithm $\bB$ uses the reduction from Lemma~\ref{lem:dsnp-inapprox} to
create in $(n+\ell)^{O(1)}$ time a \dsnP instance on the graph $G'$ with $k=2\ell$ terminal pairs and size $N=(\ell+n)^{O(1)}$. The algorithm $\bB$ then runs
$\bA$ on this instance; if $\bA$ returns a solution $N$ of cost at most $2(2 -4{\gamma^*}^{1/5})$, then $\bB$ returns YES. Otherwise, $\bB$
returns NO.

We now show that the algorithm $\bB$ can indeed distinguish between the YES and NO
cases of Corollary~\ref{crl:inapprox-colored-biclique}. In the YES case, i.e., $\val(\Gamma)=1$, the completeness property of Lemma~\ref{lem:dsnp-inapprox} guarantees that the optimal planar solution has cost at most $2(1 + {\gamma^*}^{1/5})$. Since $\bA$ is a $(2 - \varepsilon)$-approximation algorithm, it
returns a solution of cost at most $2(1 + {\gamma^*}^{1/5})
\cdot (2 - \varepsilon) \leq 2(2 - 4{\gamma^*}^{1/5})$ where
the inequality comes from our choice of $\gamma^*$; this means that $\bB$
outputs YES. On the other hand, in the NO case, i.e., $\val(\Gamma)<\gamma$, the soundness property of Lemma~\ref{lem:dsnp-inapprox} guarantees
that the optimal solution \textcolor[rgb]{0.00,0.00,0.00}{(and hence the planar optimal solution as well, if it exists)} in $G'$ has cost more than $2(2 - 4{\gamma^*}^{1/5})$, which implies that $\bB$ outputs NO.

Finally, observe that the running time of $\bB$ is $f(k)\cdot N^{O(1)} + \text{poly}(\ell+n)^{O(1)}$ which is bounded by $f'(\ell)\cdot n^{O(1)}$ for some computable function $f'$ since $k=2\ell$ and $N=(n+\ell)^{O(1)}$. Hence, from Corollary~\ref{crl:inapprox-colored-biclique}, Gap-ETH is violated.
\end{proof}

\subsection{Lower Bounds for FPT Approximation Schemes for \dsnP}

We obtain the following result regarding the parameterized complexity of \dsnP
parameterized by~$k+p$.

\begin{reptheorem}{thm:dsnp-no-epas} 
The \dsnP problem is \textup{W[1]}-hard parameterized by $p+k$. Moreover, under ETH, for any computable function $f$ and any $\eps>0$
\begin{itemize}
 \item There is no $f(k,p)\cdot n^{o(k+\sqrt {p})}$ time algorithm for \dsnP, and
 \item There is no $f(k,\eps,p)\cdot n^{o(k+\sqrt{p+1/\epsilon})}$ time
algorithm which computes a $(1+\eps)$-approximation for \dsnP
\end{itemize}
\end{reptheorem}

We reduce from the \gt problem:

\begin{center}
\noindent\framebox{\begin{minipage}{5.00in}
\textbf{$(\ell,n)-\gt$}\\
\emph{Input }: Integers $\ell, n$, and $\ell^2$ non-empty sets $S_{i,j}\subseteq
[n]\times [n]$ where $1\leq i, j\leq \ell$\\
\emph{Question}: For each $1\leq i, j\leq \ell$ does there exist a value
$\gamma_{i,j}\in S_{i,j}$ such that
\begin{itemize}
\item If $\gamma_{i,j}=(x,y)$ and $\gamma_{i,j+1}=(x',y')$ then $x=x'$.
\item If $\gamma_{i,j}=(x,y)$ and $\gamma_{i+1,j}=(x',y')$ then $y=y'$.
\end{itemize}
\end{minipage}}
\end{center}

\begin{figure}[!htb]
\centering
\vspace{-5mm}
\includegraphics[width=6in]{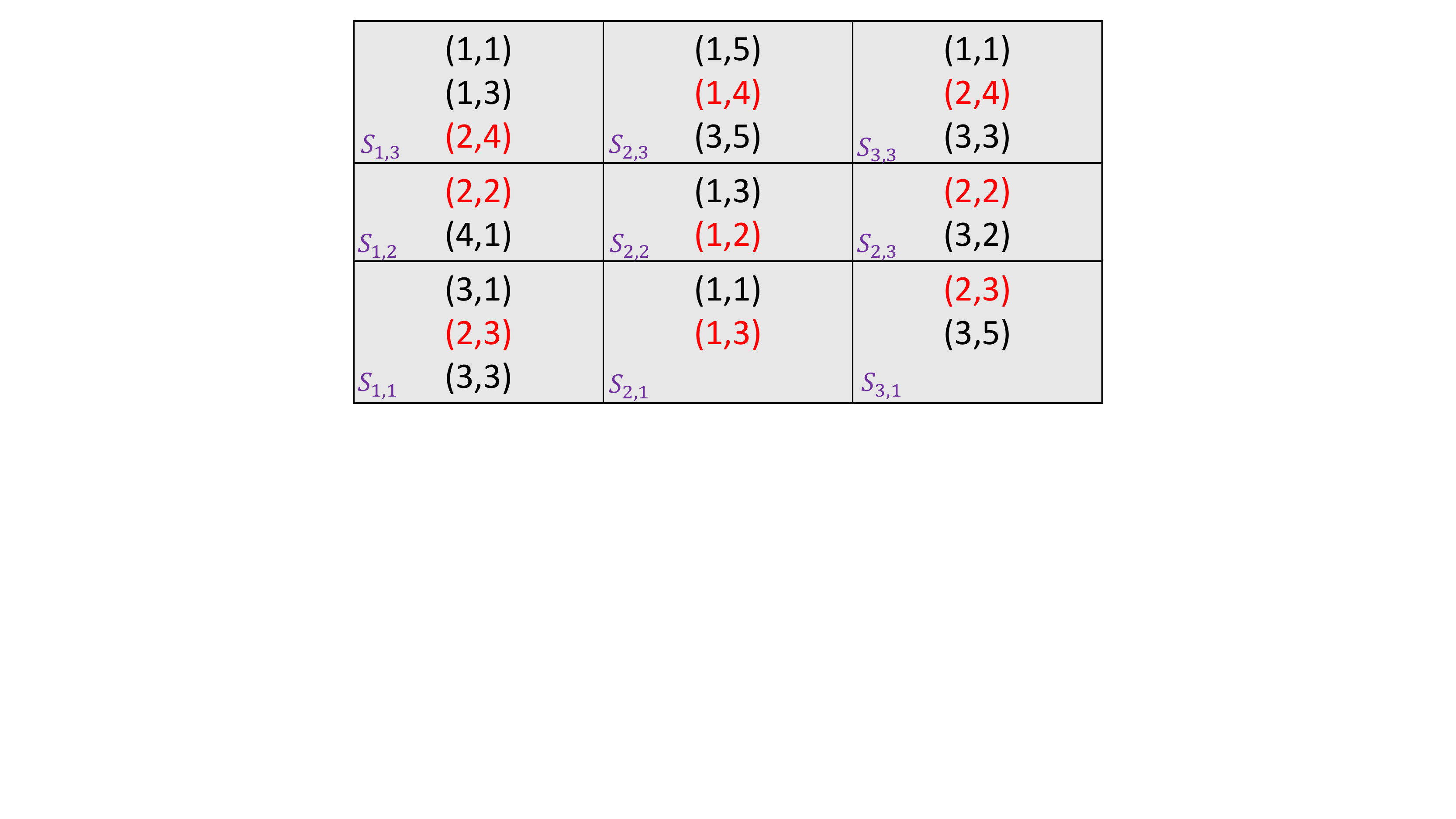}
\vspace{-35mm}
\caption{An instance of \gt with $\ell=3, n=5$ with a solution highlighted in red.
Note that in a solution, all entries from a row agree in the second coordinate
and all entries from a column agree in the first coordinate.}
\label{fig:gt}
\end{figure}

See Figure~\ref{fig:gt} for an example of an instance of \gt. We use the same construction as for Lemma~\ref{lem:dsnp-inapprox}, but with different weights.
We design two types of gadgets: the \emph{main gadget} and the \emph{secondary
gadget}. We represent each set $S_{i,j}$ with
a main gadget $M_{i,j}$ as follows: each main gadget is a copy
of the path gadget $\mathcal{P}_{|S_{i,j}|}$ from
Section~\ref{sec:path-uniqueness-gadget}
with $B=1$, i.e., there is a row in
$M_{i,j}$ corresponding to each element from $S_{i,j}$. Each main gadget is
surrounded by four secondary gadgets:
on the top, right, bottom and left. Each of these gadgets are copies of the path
gadget from Section~\ref{sec:path-uniqueness-gadget}:
\begin{itemize}
  \item For each $1\leq i\leq \ell+1, 1\leq j\leq \ell$ the \emph{horizontal} gadget $HS_{i,j}$ is a copy of $\mathcal{P}_{n}$ with $B=1$
  \item For each $1\leq i\leq \ell, 1\leq j\leq \ell+1$ the \emph{vertical} gadget $VS_{i,j}$ is a copy of $\mathcal{P}_{n}$ with $B=1$
\end{itemize}

We refer to Figure~\ref{fig:dsnp-big-picture} (bird's-eye view) and Figure~\ref{fig:dsnp-zoomed-in} (zoomed-in view) for an illustration of the reduction.
Fix some $1\leq i,j\leq \ell$. The main gadget $M_{i,j}$ has four secondary
gadgets surrounding it:
\begin{itemize}
\item Above $M_{i,j}$ is the vertical secondary gadget $VS_{i,j+1}$
\item On the right of $M_{i,j}$ is the horizontal secondary gadget $HS_{i+1,j}$
\item Below $M_{i,j}$ is the vertical secondary gadget $VS_{i,j}$
\item On the left of $M_{i,j}$ is the horizontal secondary gadget $HS_{i,j}$
\end{itemize}
Hence, there are $\ell(\ell+1)$ horizontal secondary gadgets and $\ell(\ell+1)$ vertical secondary gadgets.

\textbf{Red intra-gadget edges}: Fix $(i,j)$ such that $1\leq i,j\leq \ell$. Recall that $M_{i,j}$ is a copy of $\mathcal{P}_{|S_{i,j}|}$ with $B=1$ and each of the secondary
gadgets are copies of $\mathcal{P}_{n}$ with $B=1$. With slight abuse of
notation, we assume that the rows of $M_{i,j}$ are indexed by the set $\{(x,y)\in S_{i,j}\
:\ x,y\in [n]\}$.
We add the following edges (in \textcolor[rgb]{1.00,0.00,0.00}{red} color) of
weight $1$: for each $(x,y)\in S_{i,j}$
\begin{itemize}
\item Add the edge $VS_{i,j+1}(1_x)\rightarrow M_{i,j}(0_{(x,y)})$. These edges
are called \emph{top-red} edges incident on $M_{i,j}$.
\item Add the edge $HS_{i,j}(1_y) \rightarrow M_{i,j}(0_{(x,y)})$. These edges
are called \emph{left-red} edges incident on $M_{i,j}$.
\item Add the edge $M_{i,j}(1_{(x,y)}) \rightarrow HS_{i+1,j}(0_y)$. These
edges are called \emph{right-red} edges incident on $M_{i,j}$.
\item Add the edge $M_{i,j}(1_{(x,y)}) \rightarrow VS_{i,j}(0_x)$. These edges
are called \emph{bottom-red} edges incident on $M_{i,j}$.
\end{itemize}
These are called the \emph{intra-gadget} edges incident on $M_{i,j}$.

Introduce the following $4\ell$ vertices (which we call \emph{border} vertices):
\begin{itemize}
\item $a_1, a_2, \ldots, a_\ell$
\item $b_1, b_2, \ldots, b_\ell$
\item $c_1, c_2, \ldots, c_\ell$
\item $d_1, d_2, \ldots, d_\ell$
\end{itemize}
\textbf{Orange edges}: For each $i\in [\ell]$ add the following edges (shown as \textcolor{orange}{orange} in
Figure~\ref{fig:dsnp-big-picture}) with weight~$1$:
\begin{itemize}
\item $a_{i} \rightarrow VS_{i,\ell+1}(0_{v})$ for each $v\in V_i$. These are
called \emph{top-orange} edges.
\item $VS_{i,1}(1_{v}) \rightarrow b_{i}$ for each $v\in V_i$. These are called
\emph{bottom-orange} edges.
\item $c_{j} \rightarrow HS_{1,j}(0_{w})$ for each $w\in W_j$. These are called
\emph{left-orange} edges.
\item $HS_{\ell+1,j}(1_{w}) \rightarrow d_{j}$ for each $w\in W_j$. These are
called \emph{right-orange} edges.
\end{itemize}

%

Finally, the set of demand pairs $\mathcal{D'}$ is given by:
\begin{itemize}
\item \underline{Type I}: the pairs $(a_i, b_i)$ for each $1\leq i\leq \ell$.

\item \underline{Type II}: the pairs $(c_j, d_j)$ for each $1\leq j\leq \ell$.
\end{itemize}

%
%
%
%
%
%
%
%


Let the final graph constructed be $G'$. Note that $G'$ has size $N=(n+\ell)^{O(1)}$ and can be constructed in $(n+\ell)^{O(1)}$ time. It is also easy to see that $G'$ is actually a DAG.

Fix the budget $B^*= 6\ell+ 7\ell^2 = O(\ell^2)$. We now show that the instance
$(\ell,n, \{S_{i,j}\ : i,j\in [\ell]\})$  of $(\ell,n)$-\gt answers YES if and only the
instance $(G',\mathcal{D}')$ of \dsnP has a solution of cost at most $B^*$.


\subsubsection{\textcolor[rgb]{0.00,0.00,0.00}{\gt answers YES $\Rightarrow$
instance $(G',\mathcal{D}')$ of \dsnP has a planar solution of cost~$\leq B^*$}}
\label{subsec:dsnp-epas-easy}

Suppose that \gt has a solution, i.e., for each $1\leq i,j\leq \ell$ there is a
value $(x_{i,j}, y_{i,j})=\gamma_{i,j}\in S_{i,j}$ such that
\begin{itemize}
\item for every $i\in [\ell]$, we have $x_{i,1}=x_{i,2}=x_{i,3}=\ldots=x_{i,\ell} =
\alpha_i$, and
\item for every $j\in [\ell]$, we have $y_{1,j}=y_{2,j}=y_{3,j}=\ldots=y_{\ell,j} =
\beta_j$.
\end{itemize}
~\\
We now build a \emph{planar} solution $N$ for the \bidsn instance $(G',
\mc{D}')$ and show that it has weight at most $B^*$. In the edge set $N$, we
take the following edges:
\begin{enumerate}
\item For each $i\in [\ell]$ pick the edges
        \begin{itemize}
          \item Top-orange edge $a_i \rightarrow VS_{i,\ell+1}(0_{\alpha_i})$
          \item Bottom-orange edge $VS_{i,1}(1_{\alpha_i})\rightarrow b_i$
          \item Left-orange edge $c_j \rightarrow HS_{1,j}(0_{\beta_j})$
          \item Right-orange edge $HS_{\ell+1,j}(1_{\beta_i})\rightarrow d_j$
        \end{itemize}
        This incurs a cost of $4\ell$ since each orange edge has cost $1$.

%
%

\item For each $1\leq i,j\leq \ell$ for the main gadget $M_{i,j}$, pick the edge
$MS_{i,j}(0_{\alpha_i,\beta_j})\rightarrow MS_{i,j}(1_{\alpha_i,\beta_j})$ of
weight $1$. Note that this edge exists because $(\alpha_i,\beta_j)\in S_{i,j}$
for each $1\leq i,j\leq \ell$ because \gt answers YES. Additionally we also pick
the following four \textcolor[rgb]{1.00,0.00,0.00}{red} edges (each of which has
weight 1):
        \begin{itemize}
        \item $VS_{i,j+1}(1_{\alpha_i})\rightarrow M_{i,j}(0_{\alpha_i, \beta_j})$
        \item $HS_{i,j}(1_{\beta_j})\rightarrow M_{i,j}(0_{\alpha_i, \beta_j})$
        \item $VS_{i,j}(0_{\alpha_i})\leftarrow M_{i,j}(1_{\alpha_i, \beta_j})$
        \item $HS_{i+1,j}(0_{\beta_j})\leftarrow M_{i,j}(1_{\alpha_i, \beta_j})$
        \end{itemize}
        This incurs a cost of $\ell^2 +4\ell^2 = 5\ell^2$.

\item For each $1\leq j\leq \ell+1$ and $1\leq i\leq \ell$ for the vertical secondary gadget $VS_{i,j}$, pick the edge $VS_{i,j}(0_{\alpha_i})\rightarrow VS_{i,j}(1_{\alpha_i})$ which has weight $1$. This incurs a cost of $\ell(\ell+1)$.

\item For each $1\leq j\leq \ell$ and $1\leq i\leq \ell+1$ for the vertical secondary gadget $VS_{i,j}$, pick the edge $HS_{i,j}(0_{\beta_j})\rightarrow HS_{i,j}(1_{\beta_j})$ which has weight $1$. This incurs a cost of $\ell(\ell+1)$.
\end{enumerate}

Hence, the weight of $N$ is $4\ell+ \ell^2 + 4\ell^2 + \ell(\ell+1) +
\ell(\ell+1) = B^*$. We now argue that $N$ is planar. It is easy to see that removing the red edges from $G'$ leads to
a planar graph (see Figure~\ref{fig:dsnp-big-picture} for a planar embedding of
this graph). It remains to show that the red edges we add in $N$ do not destroy
planarity. For any main gadget $M_{i,j}$: the only red edges from $G'$ which are
added in $N$ are as follows: one left-red edge and one top-red edge incident on
the same $0$-vertex of $M_{i,j}$ and one bottom-red edge and one right-red edge
incident on the same $1$-vertex of $M_{i,j}$. This can clearly be done while
preserving planarity: the only 4 red edges retained in $N$ are shown as in
Figure~\ref{fig:dsnp-zoomed-in} (note that Figure~\ref{fig:dsnp-zoomed-in} is
actually supposed to have many more red edges which are omitted for clarity).

%
%

It remains to show that $N$ is indeed a solution for the \dsnP instance $(G',\mathcal{D}')$. We
show that each demand pair of Type I is satisfied. Fix $i\in [\ell]$. Then there
is an $a_i\leadsto b_i$ path in $N$ given by the following edges:
\begin{itemize}
  \item $a_i\rightarrow VS_{i,\ell+1}(0_{\alpha_i})$
  \item For each $\ell+1\geq r\geq 2$ use the path
            \begin{itemize}
              \item $VS_{i,r}(0_{\alpha_i})\rightarrow VS_{i,r}(1_{\alpha_i})\rightarrow M_{i,r-1}(0_{\alpha_i, \beta_{r-1}})\rightarrow M_{i,r-1}(1_{\alpha_i, \beta_{r-1}})\rightarrow VS_{i,r-1}(0_{\alpha_i})$
            \end{itemize}
  \item Finally use the path $VS_{i,1}(0_{\alpha_i})\rightarrow VS_{i,1}(1_{\alpha_i})\rightarrow b_i$
\end{itemize}
The argument showing that each demand pair of Type II is satisfied in $N$ is
very similar, and we omit the details here.

\subsubsection{Instance $(G',\mathcal{D}')$ of \dsnP has a solution of cost $\leq B^*$ $\Rightarrow$ \gt answers YES}
\label{subsec:dsnp-epas-hard}

\textcolor[rgb]{0.00,0.00,0.00}{Suppose that the instance $(G',\mathcal{D}')$ of \dsnP has a solution $N$ of cost at most $B^*= 6\ell+ 7\ell^2$. We will now show that this
implies that \gt answers YES. This implies that if \gt answers NO then the cost
of an optimal solution (and hence the cost of an optimal planar solution, if
one exists) is greater than $B^*$.}


\begin{lemma}
$N$ contains at least $4\ell$ orange edges. In fact, for each $1\leq i\leq \ell$
we have that $N$ contains at least one
\begin{itemize}
\item outgoing orange edge from $a_i$
\item incoming orange edge into $b_i$
\item outgoing orange edge from $c_j$
\item incoming orange edge into $d_j$
\end{itemize}
\label{lem:orange-planar-dsnp}
\end{lemma}
\begin{proof}
The terminal pair $(a_i,b_i)$ is in $\mathcal{D}'$ for each $i\in [\ell]$.
Since the only outgoing edges from $a_i$ are top-orange edges, it follows that
$N$ contains at least one orange edge outgoing from $a_i$.
The other three claims follow by similar arguments.
\end{proof}

For each $j\in [\ell]$, we define $$\horizontal(j) = \{c_j,d_j\}\bigcup \Big(\cup_{i\in [\ell]} M_{i,j}\Big) \bigcup \Big(\cup_{i\in [\ell+1]} HS_{i,j}\Big)$$

\begin{lemma}
For every $j\in [\ell]$ any $c_j \leadsto d_j$ path must have all edges in $G'[\horizontal(j)]$.
\label{lem:every-horizontal-path-contained-dsnp}
\end{lemma}
\begin{proof}
Observe that the graph $G'$ is a DAG. Any $c_j \leadsto d_j$ path starts with $c_j$ which is a vertex from $\horizontal(j)$. The only incoming edges into $\horizontal(j)$ are the in-vertical red edges from the set of vertices $\cup_{i=1}^{\ell} VS_{i,j+1}$ and the only outgoing edges from $\horizontal(j)$ are the out-vertical red edges to the set of vertices $\cup_{i=1}^{\ell} VS_{i,j}$. Hence, no $c_j \leadsto d_j$ path can leave the vertex set $\horizontal(j)$.
\end{proof}

For each $i\in [\ell]$, we define
$$ \verticall(i) = \{a_i,b_j\}\bigcup \Big(\cup_{j\in [\ell]} M_{i,j}\Big) \bigcup \Big(\cup_{j\in [\ell+1]} VS_{i,j}\Big) $$
The proof of the next lemma is analogous to that of Lemma~\ref{lem:every-horizontal-path-contained-dsnp}:
\begin{lemma}
For every $i\in [\ell]$ any $a_i \leadsto b_i$ path must have all edges in $G'[\verticall(i)]$.
\label{lem:every-vertical-path-contained-dsnp}
\end{lemma}

\begin{corollary}
For every $1\leq i,j\leq \ell$ the edge set $N$ contains at least one
intra-gadget edge from the main gadget $M_{i,j}$.
\label{crl:in-out-main-dsnp}
\end{corollary}
\begin{proof}
Fix $j\in [\ell]$. By Lemma~\ref{lem:every-horizontal-path-contained-dsnp}, there is an $c_j \leadsto d_j$ path contained in $G'[\horizontal(j)]$. Hence, this path must contain at least one intra-gadget edge from each main gadget $M_{i,j}$ for each $1\leq i\leq \ell$.
\end{proof}

Analogous lemmas hold also for the horizontal secondary gadgets and the vertical secondary gadgets:
\begin{corollary}
For every $1\leq i\leq \ell+1, 1\leq j\leq \ell$ the edge set $N$ contains at
least one intra-gadget edge from the horizontal secondary gadget $ HS_{i,j}$.
\label{crl:in-out-horizontal-dsnp}
\end{corollary}
\begin{proof}
Fix $j\in [\ell]$. By Lemma~\ref{lem:every-horizontal-path-contained-dsnp}, there is an $c_j \leadsto d_j$ path contained in $G'[\horizontal(j)]$. Hence, this path must contain at least one intra-gadget edge from each horizontal secondary gadget $HS_{i,j}$ for each $1\leq i\leq \ell+1$.
\end{proof}

\begin{corollary}
For every $1\leq i\leq \ell, 1\leq j\leq \ell+1$ the edge set $N$ contains at
least one intra-gadget edge from the vertical secondary gadget $VS_{i,j}$.
\label{crl:in-out-vertical-dsnp}
\end{corollary}
\begin{proof}
Fix $i\in [\ell]$. By Lemma~\ref{lem:every-vertical-path-contained-dsnp}, there is an $a_i \leadsto b_i$ path contained in $G'[\verticall(j)]$. Hence, this path must contain at least one intra-gadget edge from each vertical secondary gadget $VS_{i,j}$ for each $1\leq j\leq \ell+1$.
\end{proof}

\begin{corollary}
For each $1\leq i,j\leq \ell$, the solution $N$ contains at least one
\begin{itemize}
  \item top-red edge incident on $M_{i,j}$
  \item right-red edge incident on $M_{i,j}$
  \item bottom-red edge incident on $M_{i,j}$
  \item left-red edge incident on $M_{i,j}$
\end{itemize}
\label{crl:main-at-least-4-red-planar-dsnp}
\end{corollary}
\begin{proof}
Fix some $1\leq i,j\leq \ell$. By
Lemma~\ref{lem:every-horizontal-path-contained-dsnp}, there is an $c_j \leadsto
d_j$ path contained in $G'[\horizontal(j)]$. The only way to enter $M_{i,j}$ by
edges within $\horizontal(j)$ is via left-red edges incident on $M_{i,j}$, and
the only way to exit $M_{i,j}$ by edges within $\horizontal(j)$ is via right-red
edges incident on $M_{i,j}$. Hence, $N$ contains at least one left-red edge and
at least one right-red edge incident on $M_{i,j}$.

By Lemma~\ref{lem:every-vertical-path-contained-dsnp}, there is an $a_i \leadsto
b_i$ path contained in $G'[\verticall(j)]$. The only way to enter $M_{i,j}$ by
edges within $\verticall(j)$ is via top-red edges incident on $M_{i,j}$, and the
only way to exit $M_{i,j}$ by edges within $\verticall(j)$ is via bottom-red
edges incident on $M_{i,j}$. Hence, $N$ contains at least one top-red edge and
at least one bottom-red edge incident on $M_{i,j}$.
\end{proof}

We show now that there is no slack, i.e., weight of $N$ must be exactly $B^*$.

\begin{lemma}
The weight of $N$ is exactly $B^*$, and hence it is minimal (under edge
deletions) since no edges have zero weights.
\label{lem:exact-B^*-planar-dsnp}
\end{lemma}
\begin{proof}
We have the following collection of pairwise disjoint sets of edges which are
guaranteed to be contained in $N$:
\begin{itemize}
  \item $4\ell$ orange edges (from Lemma~\ref{lem:orange-planar-dsnp}). This
incurs a cost of at least $4\ell$.
  \item A cost of at  least $1$ from intra-gadget edges of vertical secondary
gadgets (from Corollary~\ref{crl:in-out-vertical-dsnp}). This incurs a cost of
at least $\ell(\ell+1)$
  \item A cost of at  least $1$ from intra-gadget edges of horizontal secondary
gadgets (from Corollary~\ref{crl:in-out-horizontal-dsnp}). This incurs a cost of
at least $\ell(\ell+1)$
  \item A cost of at  least $1$ from intra-gadget edges of main gadgets (from
Corollary~\ref{crl:in-out-main-dsnp}). This incurs a cost of $\ell^2$
  \item A cost of $\geq 4$ from inter-gadget edges of main gadgets (from Corollary~\ref{crl:main-at-least-4-red-planar-dsnp}). This incurs a cost of $4\ell^2$
\end{itemize}

Hence, the cost of $N$ is at least $4\ell+\ell(\ell+1)+\ell(\ell+1)+\ell^2 +
\ell^2 = B^*$. But, we are given that cost of $N$ is at most $B^*$. Hence,
the cost of $N$ is exactly $B^*$.
\end{proof}


The following corollary follows from Lemma~\ref{lem:exact-B^*-planar-dsnp}:
\begin{corollary}
The solution $N$ contains exactly one intra-gadget edge from each gadget (main,
vertical secondary or horizontal secondary). Hence,
\begin{itemize}

\item for each $1\leq i\leq \ell+1, 1\leq j\leq \ell$, the unique intra-gadget
edge from the vertical secondary gadget $VS_{i,j}$ in $N$ is
$VS_{i,j}(0_{x_{i,j}})\rightarrow VS_{i,j}(1_{x_{i,j}})$ for some $x_{i,j}\in
[n]$

\item for each $1\leq i\leq \ell, 1\leq j\leq \ell+1$, the unique intra-gadget
edge from the horizontal secondary gadget $HS_{i,j}$ in $N$ is
$HS_{i,j}(0_{y_{i,j}})\rightarrow HS_{i,j}(1_{y_{i,j}})$ for some $y_{i,j}\in
[n]$

\item for each $1\leq i, j\leq \ell$, the unique intra-gadget edge from the main
gadget $M_{i,j}$ in $N$ is $M_{i,j}(0_{\lambda_{i,j},\delta_{i,j}})\rightarrow
M_{i,j}(1_{\lambda_{i,j},\delta_{i,j}})$ for some $(\lambda_{i,j},
\delta_{i,j})\in S_{i,j}$
\end{itemize}
\label{crl:gadgets-representation-dsnp}
\end{corollary}

The following corollary follows from Lemma~\ref{lem:exact-B^*-planar-dsnp}:
\begin{corollary}
For each $1\leq i,j\leq \ell$, the solution $N$ contains exactly one
\begin{itemize}
  \item top-red edge incident on $M_{i,j}$
  \item right-red edge incident on $M_{i,j}$
  \item bottom-red edge incident on $M_{i,j}$
  \item left-red edge incident on $M_{i,j}$
\end{itemize}
\label{crl:main-exactly-4-red-planar-dsnp}
\end{corollary}

Consider a main gadget $M_{i,j}$. The main gadget has four secondary gadgets
surrounding
it: $VS_{i,j}$ below it, $VS_{i,j+1}$ above it, $HS_{i,j}$ to the left and
$HS_{i+1,j}$ to the right.


\begin{lemma}
(\textbf{propagation}) For every main gadget $M_{i, j}$, we have
$x_{i,j}=\lambda_{i,j}=x_{i,j+1}$ and $y_{i,j}=\delta_{i,j}=y_{i+1,j}$.
\label{lem:agreement-tight-planar}
\end{lemma}
\begin{proof}
Due to symmetry, it suffices to only argue that $x_{i,j}=\lambda_{i,j}$. By
Corollary~\ref{crl:gadgets-representation}, the only intra-gadget edge from the
vertical secondary gadget is $VS_{i,j}$ in $N$ is
$VS_{i,j}(0_{x_{i,j}})\rightarrow VS_{i,j}(1_{x_{i,j}})$ and the only
intra-gadget edge from the main gadget $M_{i,j}$ in $N$ is
$M_{i,j}(0_{\lambda_{i,j},\delta_{i,j}})\rightarrow
M_{i,j}(1_{\lambda_{i,j},\delta_{i,j}})$. By
Corollary~\ref{crl:main-exactly-4-red-planar-dsnp}, there is exactly one
bottom-red incident edge on $M_{i,j}$. Moreover, this is the only incoming edge
into $VS_{i,j}$. Hence, it follows that $x_{i,j}=\lambda_{i,j}$.
\end{proof}

\begin{lemma}
The \gt instance $(\ell,n, \{S_{i,j}\ : i,j\in [\ell]\})$ has a solution.
\label{thm:gt-says-yes}
\end{lemma}
\begin{proof}
By Lemma~\ref{lem:agreement-tight-planar}, it follows that for each
$1\leq i,j\leq \ell$ we have $x_{i,j}=\lambda_{i,j}=x_{i,j+1}$ and
$y_{i,j}=\delta_{i,j}=y_{i+1,j}$ in addition to $(\lambda_{i,j},
\delta_{i,j})\in S_{i,j}$ (by the definition of the main gadget). This implies
that \gt has a solution.
\end{proof}

\subsubsection{Finishing the proof of Theorem~\ref{thm:dsnp-no-epas}}

There is a simple reduction~\cite[Theorem 14.28]{fpt-book} from
$\ell$-\textsc{Clique} on $n$ vertex graphs to $(\ell,n)$ -\gt. Combining the two directions
from Section~\ref{subsec:dsnp-epas-easy} and Section~\ref{subsec:dsnp-epas-hard}
gives a parameterized reduction from $(\ell,n)$-\gt to \dsnP on $(n+\ell)^{O(1)}$ vertex graphs with
$k=O(\ell)$ and $p=O(\ell^2)$. Composing the two reductions, we get a
parameterized reduction from an $\ell$-\textsc{Clique} instance on $n$ vertices to an instance \dsnP with $(n+\ell)^{O(1)}$ vertices, $k=O(\ell)$
and $p=O(\ell^2)$. Hence, the W[1]-hardness of \dsnP parameterized by $(k+p)$
follows from the W[1]-hardness of $\ell$-\textsc{Clique} parameterized by
$\ell$. Moreover, Chen et al.~\cite{chen-hardness} showed that, for any function
$f$, the existence of an $f(\ell)\cdot n^{o(\ell)}$ algorithm for
\textsc{Clique} violates ETH.
Hence, we obtain that, under ETH, there is no $f(k,p)\cdot n^{o(k+\sqrt{p})}$
time algorithm for \dsnP.

%

Suppose now that there is an algorithm $\mathbb{A}$ which runs in time
$f(k,p,\eps)\cdot n^{o(k+\sqrt{p+1/\epsilon})}$ (for some computable function $f$) and computes an $(1+\eps)$-approximate
solution for \dsnP. Recall that our reduction works as follows: \gt answers
YES if and only if \dsnP has a solution of cost $B^* \leq 6\ell+ 7\ell^2\leq
13\ell^2 < 14\ell^2$.
Consequently, consider running $\mathbb{A}$ with $\eps=\frac{1}{14\ell^2}$ implies that $(1+\eps)\cdot B^{*} <B^{*}+1$. Every edge of our constructed graph $G$
has weight at least $1$, and hence an $(1+\epsilon)$-approximation is in fact
forced to find a solution of cost at most $B^*$, i.e., $\mathbb{A}$ finds
an optimum solution. Since $k=O(\ell), p=O(\ell^2)$ and $1/\epsilon= O(\ell^2)$ it follows $f(k,p,\eps)\cdot n^{o(k+\sqrt{p+1/\epsilon})}=g(\ell)\cdot n^{o(\ell)}$ for some computable function $g$. By the previous paragraph, this is not possible under ETH.

\section{Lower Bounds for FPT Approximation Schemes for \scssP}

We obtain the following result regarding the parameterized complexity of \scssP
parameterized by~$k+p$.
%
%
\begin{reptheorem}{thm:scssp-no-epas} 
The \scssP problem is \textup{W[1]}-hard parameterized by $p+k$. Moreover, under ETH, for any computable function $f$ and any $\eps>0$
\begin{itemize}
 \item there is no $f(k,p)\cdot n^{o(\sqrt {k+p})}$ time algorithm for \scssP,
and
 \item there is no $f(k,\eps,p)\cdot n^{o(\sqrt{k+p+1/\epsilon})}$ time
algorithm which computes an $(1+\eps)$-approximation for \mbox{\scssP}.
\end{itemize}
\end{reptheorem}

To prove Theorem~\ref{thm:scssp-no-epas}, we give a reduction which transforms
an instance $(\ell,n, \{S_{i,j}\ : i,j\in [\ell]\})$ of $\ell\times \ell$ \gt
into an instance of $(G,\mathcal{T})$ of \scssP which has
$|\mathcal{T}|=O(\ell^2)$ terminals and an optimum
which is planar and has size $O(\ell^2)$. First we construct a ``uniqueness" gadget which is used repeatedly as a building block in our construction.

\subsection{Constructing a ``uniqueness'' gadget}
\label{sec:scssp-uniqueness-gadget}

For every integer $n$ we define the following gadget $\mathcal{U}_n$ which contains
$4n+4$ vertices (see \autoref{fig:uniqueness}). Since we need many of these
gadgets later on, we will denote vertices of $\mathcal{U}_n$ by $\mathcal{U}_n(v)$ etc., in order
to be able to distinguish vertices of different gadgets. All edges will have
the same weight $B$, which we will fix later during the reductions. The gadget
$\mathcal{U}_n$ is constructed as follows (we first construct an undirected graph, and
then bidirect each edge):

        \begin{itemize}
        \item For each $i\in [n]$ introduce four vertices $\mathcal{U}_{n}(0_i),\mathcal{U}_{n}(1_i)$,$\mathcal{U}_{n}(2_i),\mathcal{U}_{n}(3_i)$.
        \item Introduce two terminal vertices $\mathcal{U}_{n}(s_{1})$ and $\mathcal{U}_{n}(s_{2})$,
        \item $\mathcal{U}_n$ has a path of three edges corresponding to each $i\in [n]$.
                \begin{itemize}
                \item Let $i\in [n]$. Then we denote the path in $\mathcal{U}_{n}$ corresponding to $i$ by $P_{\mathcal{U}_{n}}(i):= \mathcal{U}_{n}(0_{i})\rightarrow \mathcal{U}_{n}(1_{i})\rightarrow \mathcal{U}_{n}(2_{i})\rightarrow \mathcal{U}_{n}(3_{i})$.
                \item Each of these edges is called as a ``\emph{base}'' edge and has weight $B$
                 \end{itemize}

        \item We add the following edges:
                \begin{itemize}
                \item $\mathcal{U}_{n}(s_{1})\rightarrow \mathcal{U}_{n}(1_{i})$ for each $i\in [n]$
                \item $\mathcal{U}_{n}(s_{2})\leftarrow \mathcal{U}_{n}(2_{i})$ for each $i\in [n]$
                \item Each of these edges is called a ``\emph{connector}'' edge and has weight $B$.
                \end{itemize}
        \item We also add the edge $\mathcal{U}_{n}(s_2)\rightarrow \mathcal{U}_{n}(s_1)$ with weight $B$ and call it as a ``bridge" edge.

        \end{itemize}

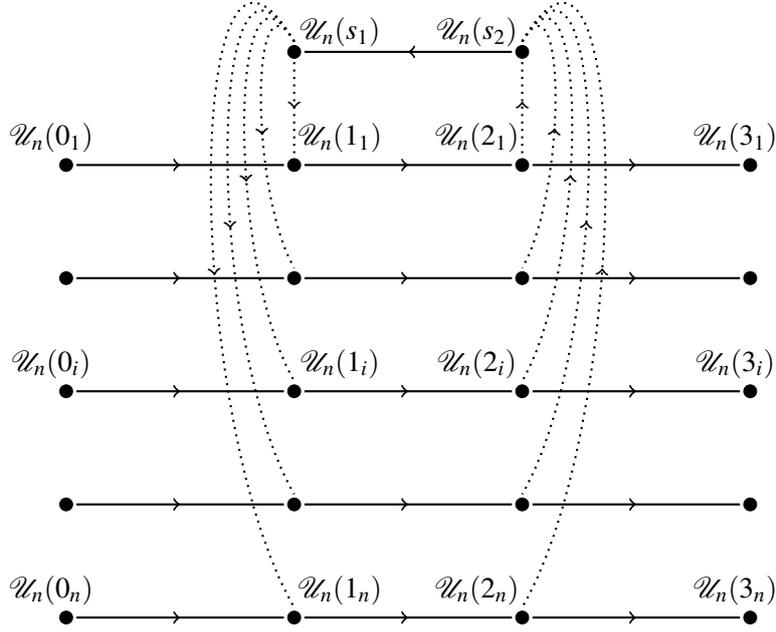
\begin{figure}[h]

\centering
\begin{tikzpicture}[scale=0.3]


\foreach \j in {0,1,2,3,4}
{
\begin{scope}[shift={(0,5*\j)}]

\foreach \i in {0,1,2,3}
{
\draw [black] plot [only marks, mark size=8, mark=*] coordinates {(10*\i,0)};
}

\foreach \i in {0,10,20}
{
\path (\i, 0) node(a) {} (10+\i,0) node(b) {};
        \draw[thick,black,middlearrow={>}] (a) -- (b);
}

\end{scope}
}


\draw [black] plot [only marks, mark size=8, mark=*] coordinates {(10,25)}
node[label={[xshift=6mm,yshift=-2mm] $\mathcal{U}_{n}(s_1)$}] {} ;

\draw [black] plot [only marks, mark size=8, mark=*] coordinates {(20,25)}
node[label={[xshift=-6mm,yshift=-2mm] $\mathcal{U}_{n}(s_2)$}] {} ;

\path (20,25) node(a) {} (10,25) node(b) {};
 \draw [thick,middlearrow={>}] (a) to (b);

%


\foreach \j in {0}
{
\begin{scope}[shift={(10*\j,0)}]

\path (10,25) node(a) {} (10,20) node(b) {};
 \draw [thick,dotted,middlearrow={>}] (a) to (b);

\path (10,25) node(a) {} (10,15) node(b) {};
 \draw [thick,dotted,middlearrow={>}] (a.north) to [out=120,in=120] (b.north);

\path (10,25) node(a) {} (10,10) node(b) {};
 \draw [thick,dotted,middlearrow={>}] (a.north) to [out=120,in=120] (b.north);

\path (10,25) node(a) {} (10,5) node(b) {};
 \draw [thick,dotted,middlearrow={>}] (a.north) to [out=120,in=120] (b.north);

\path (10,25) node(a) {} (10,0) node(b) {};
 \draw [thick,dotted,middlearrow={>}] (a.north) to [out=120,in=120] (b.north);

\end{scope}
}

\foreach \j in {1}
{
\begin{scope}[shift={(10*\j,0)}]

\path (10,25) node(a) {} (10,20) node(b) {};
 \draw [thick,dotted,middlearrow={<}] (a) to (b);

\path (10,25) node(a) {} (10,15) node(b) {};
 \draw [thick,dotted,middlearrow={<}] (a.north) to [out=60,in=60] (b.north);

\path (10,25) node(a) {} (10,10) node(b) {};
 \draw [thick,dotted,middlearrow={<}] (a.north) to [out=60,in=60] (b.north);

\path (10,25) node(a) {} (10,5) node(b) {};
 \draw [thick,dotted,middlearrow={<}] (a.north) to [out=60,in=60] (b.north);

\path (10,25) node(a) {} (10,0) node(b) {};
 \draw [thick,dotted,middlearrow={<}] (a.north) to [out=60,in=60] (b.north);

\end{scope}
}


%
%
%
%
%
%
%
%


\draw [black] plot [only marks, mark size=0, mark=*] coordinates {(0,20)}
node[label={[xshift=-2mm,yshift=-1mm] $\mathcal{U}_{n}(0_1)$}] {} ;

\draw [black] plot [only marks, mark size=0, mark=*] coordinates {(10,20)}
node[label={[xshift=6mm,yshift=-1mm] $\mathcal{U}_{n}(1_1)$}] {} ;

\draw [black] plot [only marks, mark size=0, mark=*] coordinates {(20,20)}
node[label={[xshift=-6mm,yshift=-1mm] $\mathcal{U}_{n}(2_1)$}] {} ;

\draw [black] plot [only marks, mark size=0, mark=*] coordinates {(30,20)}
node[label={[xshift=-2mm,yshift=-1mm] $\mathcal{U}_{n}(3_1)$}] {} ;

\draw [black] plot [only marks, mark size=0, mark=*] coordinates {(0,10)}
node[label={[xshift=-2mm,yshift=-1mm] $\mathcal{U}_{n}(0_i)$}] {} ;

\draw [black] plot [only marks, mark size=0, mark=*] coordinates {(10,10)}
node[label={[xshift=6mm,yshift=-1mm] $\mathcal{U}_{n}(1_i)$}] {} ;

\draw [black] plot [only marks, mark size=0, mark=*] coordinates {(20,10)}
node[label={[xshift=-6mm,yshift=-1mm] $\mathcal{U}_{n}(2_i)$}] {} ;

\draw [black] plot [only marks, mark size=0, mark=*] coordinates {(30,10)}
node[label={[xshift=-2mm,yshift=-1mm] $\mathcal{U}_{n}(3_i)$}] {} ;

\draw [black] plot [only marks, mark size=0, mark=*] coordinates {(0,0)}
node[label={[xshift=-2mm,yshift=-1mm] $\mathcal{U}_{n}(0_n)$}] {} ;

\draw [black] plot [only marks, mark size=0, mark=*] coordinates {(10,0)}
node[label={[xshift=6mm,yshift=-1mm] $\mathcal{U}_{n}(1_n)$}] {} ;

\draw [black] plot [only marks, mark size=0, mark=*] coordinates {(20,0)}
node[label={[xshift=-6mm,yshift=-1mm] $\mathcal{U}_{n}(2_n)$}] {} ;

\draw [black] plot [only marks, mark size=0, mark=*] coordinates {(30,0)}
node[label={[xshift=-2mm,yshift=-1mm] $\mathcal{U}_{n}(3_n)$}] {} ;

\end{tikzpicture}

\caption{The construction of the uniqueness gadget for $\mathcal{U}_n$. Note that the
gadget has $4n+4$ vertices. Each \emph{base} edge is denoted by a filled edge
and each connector edge is denoted by a dotted edge in the figure.
 \label{fig:uniqueness}}
\end{figure}


\begin{definition}
We define the set of \emph{left boundary} vertices of $\mathcal{U}_n$ to be $\bigcup_{i=1}^{n} \mathcal{U}_{n}(0_i)$ and the set of \emph{right boundary} vertices of $\mathcal{U}_n$ to be $\bigcup_{i=1}^{n} \mathcal{U}_{n}(3_i) $
\label{defn-boundary}
\end{definition}

\begin{definition}
A set of edges $E'$ of $\mathcal{U}_n$ satisfies the ``\emph{in-out}'' property if each of
the following four conditions is satisfied
\begin{itemize}
\item $\mathcal{U}_{n}(s_1)$ can reach some right boundary vertex via a path contained in the gadget $\mathcal{U}_n$
\item $\mathcal{U}_{n}(s_1)$ can be reached from some left boundary vertex via a path contained in the gadget $\mathcal{U}_n$
\item $\mathcal{U}_{n}(s_2)$ can reach some right boundary vertex via a path contained in the gadget $\mathcal{U}_n$
\item $\mathcal{U}_{n}(s_2)$ can be reached from some left boundary vertex via a path contained in the gadget $\mathcal{U}_n$
\end{itemize}
\label{defn-in-out}
\end{definition}

\begin{definition}
We say that a set of edges of $\mathcal{U}_n$ is represented by $i\in [n]$ if it contains exactly the following six edges
\begin{itemize}
\item the connector edges $\mathcal{U}_{n}(s_1)\rightarrow \mathcal{U}_{n}(1_{i})$ and $\mathcal{U}_{n}(s_2)\leftarrow \mathcal{U}_{n}(2_{i})$

\item the base edges given by the directed path $P_{\mathcal{U}_n}(i):= \mathcal{U}_{n}(0_{i})\rightarrow \mathcal{U}_{n}(1_{i})\rightarrow \mathcal{U}_{n}(2_{i})\rightarrow \mathcal{U}_{n}(3_{i})$ .
\item the bridge edge $\mathcal{U}_{n}(s_2)\rightarrow \mathcal{U}_{n}(s_1)$
\end{itemize}
We denote this set of edges by $E_{\mathcal{U}_{n}}(i)$.
\label{defn-representation}
\end{definition}

\begin{observation}
For any $i\in [n]$, the set of edges $E_{\mathcal{U}_{n}}(i)$ forms a planar graph.
\label{obs:represented-edges-form-planar-subgraph}
\end{observation}

Note that for any $i\in [n]$ the set of edges $E_{\mathcal{U}_{n}}(i)$ of $\mathcal{U}_n$ represented by $i\in [n]$ satisfies the ``in-out" property. We now show a lower bound on the cost/weight of edges we need to pick from
$\mathcal{U}_n$ to satisfy the ``in-out'' property.

\begin{lemma}
\label{lem:macro-uniqueness-gadget}
Let $E'$ be a set of edges of $\mathcal{U}_n$ which satisfies the ``in-out'' property.
Then we have that either\\
(i) the weight of $E'$ is at least $7B$\\
\textbf{OR}\\
(ii) the weight of $E'$ is exactly $6B$ and there is an integer $i\in [n]$
such that $E'$ is represented by $i$
\end{lemma}
\begin{proof}
First we observe that $E'$ must contain the bridge edge $\mathcal{U}_{n}(s_2)\rightarrow \mathcal{U}_{n}(s_1)$: this is because $\mathcal{U}_{n}(s_1)$ must be reached from some left boundary vertex and the only incoming edge into $\mathcal{U}_{n}(s_1)$ is the bridge edge. This incurs a cost of $B$. We also clearly need at least two connector edges in $N$:
\begin{itemize}
\item One outgoing edge from $\mathcal{U}_{n}(s_1)$ so that it can reach some right boundary vertex.
\item One incoming edge into $\mathcal{U}_{n}(s_2)$ so that it can be reached from some left boundary vertex.
\end{itemize}
This incurs a cost of $2B$ in $E'$. We now see how many base edges must be present in $E'$. We define the following:
\begin{itemize}
\item \underline{``0-1'' edges}: This is the set of edges $\{ \mathcal{U}_{n}(0_{i})\rightarrow \mathcal{U}_{n}(1_i)\ :\ 1\leq i\leq n\}$
\item \underline{``1-2'' edges}: This is the set of edges $\{ \mathcal{U}_{n}(1_{i})\rightarrow \mathcal{U}_{n}(2_i)\ :\ 1\leq i\leq n\}$
\item \underline{``2-3'' edges}: This is the set of edges $\{ \mathcal{U}_{n}(2_{i})\rightarrow \mathcal{U}_{n}(3_i)\ :\ 1\leq i\leq n\}$
\end{itemize}

We have the following cases:
\begin{itemize}
\item \underline{$E'$ has at least one ``0-1'' edge}: This is because the only outgoing edges from the left boundary vertices (i.e., the $0$-vertices) are to the $1$-vertices.

\item \underline{$E'$ has at least one ``2-3'' edge}: This is because the only incoming edges into the right boundary vertices (i.e., the $3$-vertices) are from the $2$-vertices.

\item \underline{$E'$ has at least one ``1-2'' edge}: Note that $\mathcal{U}_{n}(s_1)$ has to reach a right boundary vertex. The only outgoing edges from $\mathcal{U}_{n}(s_1)$ are to $1$-vertices, and the only outgoing edges from the $1$-vertices are to the $2$-vertices.
\end{itemize}


This incurs a cost of $3B$. Therefore, the solution $E'$ has cost $\geq B+2B+3B=6B$. If $E'$ contains one more edge than the ones listed above, then the cost of $E'$ is $\geq 7B$ since each edge of $\mathcal{U}_n$ has weight $B$.

Suppose that the solution $E'$ has cost exactly $6M$. Hence, it follows from the previous arguments that $E'$ contains exactly the following six edges:
\begin{itemize}
  \item A connector edge outgoing from $\mathcal{U}_{n}(s_1)$. Let this edge be $\mathcal{U}_{n}(s_1)\rightarrow \mathcal{U}_{n}(1_{\lambda_1})$ for some $\lambda_1 \in [n]$
  \item A connector edge incoming into $\mathcal{U}_{n}(s_2)$. Let this edge be $\mathcal{U}_{n}(s_2)\leftarrow \mathcal{U}_{n}(1_{\lambda_2})$ for some $\lambda_2 \in [n]$
  \item The bridge edge $\mathcal{U}_{n}(s_2)\rightarrow \mathcal{U}_{n}(s_1)$
  \item An $0-1$ edge given by $\mathcal{U}_{n}(0_{\beta_1})\rightarrow \mathcal{U}_{n}(0_{\beta_1})$ for some $\beta_1 \in [n]$
  \item An $1-2$ edge given by $\mathcal{U}_{n}(0_{\beta_2})\rightarrow \mathcal{U}_{n}(0_{\beta_2})$ for some $\beta_2 \in [n]$
  \item An $2-3$ edge given by $\mathcal{U}_{n}(0_{\beta_3})\rightarrow \mathcal{U}_{n}(0_{\beta_3})$ for some $\beta_3 \in [n]$
\end{itemize}

%

We now show that $\lambda_1=\lambda_2=i=\beta_1=\beta_2=\beta_3$ for some $i\in [n]$.
\begin{itemize}
\item \underline{How does $\mathcal{U}_{n}(s_1)$ reach a right boundary vertex}: This path must use a connector edge outgoing from $\mathcal{U}_{n}(s_1)$ followed by an $1-2$ edge and an $2-3$ edge. This implies that $\lambda_1 = \beta_2 = \beta_3$.

\item \underline{How is $\mathcal{U}_{n}(s_2)$ reached from a left boundary vertex}: This path must use an $0-1$ edge followed by a $1-2$ edge and then a connector edge incoming into $\mathcal{U}_{n}(s_2)$. This implies that $\beta_1 = \beta_2 = \lambda_2$.


\end{itemize}

Hence, we have that $\lambda_1=\lambda_2=i=\beta_1=\beta_2=\beta_3$ for some $i\in [n]$, i.e., $E'$ is represented by some $i\in [n]$.
%
%
\end{proof}

The following corollary follows immediately from the second part of proof of the
previous lemma.

\begin{corollary}
For every $i\in [n]$ there is a set of edges $E_{\mathcal{U}_n}(i)$ of cost exactly $6B$ which represents $i$ (and hence also satisfies the ``in-out'' property).
\label{crl:macro}
\end{corollary}


\subsection{Construction of the instance $(G,\mathcal{T})$ of \scssP}

We design two types of gadgets: the \emph{main gadget} and the \emph{secondary
gadget}. The reduction from \gt represents each cell of the grid with a copy of
the main gadget, and each main gadget is surrounded by four secondary gadgets:
on the top, right, bottom and left. Each of these gadgets are actually
copies of the ``uniqueness gadget'' $\mathcal{U}$ from
Section~\ref{sec:scssp-uniqueness-gadget} with weight of each edge set to
$B=1$: each secondary gadget is a copy
of $\mathcal{U}_{n}$ and for each $1\leq i,j\leq \ell$ the main gadget $M_{i,j}$
(corresponding to the set $S_{i,j}$) is a copy of $\mathcal{U}_{|S_{i,j}|}$. We refer to
Figure~\ref{fig:scssp-big-picture} (bird's-eye view) and
Figure~\ref{fig:scssp-zoomed-in} (zoomed-in view) for an
illustration of the reduction.

 \begin{figure}[hp]

 \centering
 \begin{tikzpicture}[scale=0.33]

 \foreach \j in {0,1,2}
 {
 \begin{scope}[shift={(0,12*\j)}]

 \foreach \j in {0,1,2}
 {
 \begin{scope}[shift={(12*\j,0)}]

     \draw[rectangle] (0,0) rectangle (5,5);

     \foreach \j in {0,1}
     {
     \begin{scope}[shift={(-12*\j,0)}]
     \foreach \j in {0,1,2,3}
     {
     \begin{scope}[shift={(0,\j)}]
     \foreach \i in {0,1,2,3}
    \draw [black] plot [only marks, mark size=3, mark=*] coordinates {(7+\i,1)};
     \end{scope}
     }
     \end{scope}
     }

     \foreach \j in {0,1}
     {
     \begin{scope}[shift={(-12*\j,0)}]
     \foreach \j in {0,1,2,3}
     {
     \begin{scope}[shift={(0,\j)}]
     \foreach \i in {0,1,2}
     {
     \path (7+\i, 1) node(a) {} (7+\i+1, 1) node(b) {};
         \draw[thick,black] (a) -- (b);
     }

     \end{scope}
     }
     \end{scope}
     }

     \foreach \j in {0,1}
     {
     \begin{scope}[shift={(0,-12*\j)}]
     \foreach \j in {0,1,2,3}
     {
     \begin{scope}[shift={(\j,0)}]
     \foreach \i in {0,1,2,3}
    \draw [black] plot [only marks, mark size=3, mark=*] coordinates {(1,7+\i)};
     \end{scope}
     }
     \end{scope}
     }

     \foreach \j in {0,1}
     {
     \begin{scope}[shift={(0,-12*\j)}]
     \foreach \j in {0,1,2,3}
     {
     \begin{scope}[shift={(\j,0)}]
     \foreach \i in {0,1,2}
     {
     \path (1,7+\i) node(a) {} (1,7+\i+1) node(b) {};
         \draw[thick,black] (a) -- (b);
     }
     \end{scope}
     }
     \end{scope}
     }
 \end{scope}
 }

 \end{scope}
 }


 \foreach \i in {0,1,2}
 {
 \begin{scope}[shift={(0,12*\i)}]

 \foreach \j in {-1,0,1,2}
     {
     \begin{scope}[shift={(12*\j,0)}]

 \draw[green] (6.5,-0.5) rectangle (10.5,5.5) ;

 \draw [red] plot [only marks, mark size=3, mark=*] coordinates {(8,5)}
 node[label={[xshift=-3mm,yshift=-4mm] \small{$s_1$}}] {} ;

 \draw [red] plot [only marks, mark size=3, mark=*] coordinates {(9,5)}
 node[label={[xshift=3mm,yshift=-4mm] \small{$s_2$}}] {} ;



 \end{scope}
 }
 \end{scope}
 }


 \foreach \i in {0,1,2}
 {
 \begin{scope}[shift={(12*\i,0)}]

 \foreach \j in {0,1,2,3}
     {
     \begin{scope}[shift={(0,12*\j)}]

 \draw[blue] (-0.5,-5.5) rectangle (5.5,-1.5) ;

 \draw [red] plot [only marks, mark size=3, mark=*] coordinates {(5,-3)}
 node[label={[xshift=0mm,yshift=-1mm] \small{$s_1$}}] {} ;

 \draw [red] plot [only marks, mark size=3, mark=*] coordinates {(5,-4)}
 node[label={[xshift=0mm,yshift=-7mm] \small{$s_2$}}] {} ;

%

 \end{scope}
 }
 \end{scope}
 }


 \foreach \i in {1,2,3}
 {
 \foreach \j in {1,2,3}
 {
 \draw [black] plot [only marks, mark size=0, mark=*] coordinates
{(-9.5+12*\i,-9.5+12*\j)}
 node[label={[xshift=0mm,yshift=-4mm] $M_{\i,\j}$}] {} ;
 }
 }


 \draw [black] plot [only marks, mark size=3, mark=*] coordinates {(-7.5,2.5)}
node[label={[xshift=0mm,yshift=0mm] $c_1$}] {};
 \draw [black] plot [only marks, mark size=3, mark=*] coordinates {(-7.5,14.5)}
node[label={[xshift=0mm,yshift=0mm] $c_2$}] {};
 \draw [black] plot [only marks, mark size=3, mark=*] coordinates {(-7.5,26.5)}
node[label={[xshift=0mm,yshift=0mm] $c_3$}] {};

 \foreach \i in {0,1,2}
 {
 \begin{scope}[shift={(0,12*\i)}]
 \foreach \i in {0,1,2,3}
 {
 \path (-7.5,2.5) node(a) {} (-5,\i+1) node(b) {};
         \draw[ultra thick,orange,endarrow={>}] (a) -- (b);
 }
 \end{scope}
 }

 \draw [black] plot [only marks, mark size=3, mark=*] coordinates {(36.5,2.5)}
node[label={[xshift=0mm,yshift=0mm] $d_1$}] {};
 \draw [black] plot [only marks, mark size=3, mark=*] coordinates {(36.5,14.5)}
node[label={[xshift=0mm,yshift=0mm] $d_2$}] {};
 \draw [black] plot [only marks, mark size=3, mark=*] coordinates {(36.5,26.5)}
node[label={[xshift=0mm,yshift=0mm] $d_3$}] {};

 \foreach \i in {0,1,2}
 {
 \begin{scope}[shift={(0,12*\i)}]
 \foreach \i in {0,1,2,3}
 {
 \path (36.5,2.5) node(a) {} (34,\i+1) node(b) {};
         \draw[ultra thick,orange,onethirdarrow={>}] (b) -- (a);
 }
 \end{scope}
 }

 \draw [black] plot [only marks, mark size=3, mark=*] coordinates {(2.5,-7.5)}
node[label={[xshift=0mm,yshift=-7mm] $b_1$}] {};
 \draw [black] plot [only marks, mark size=3, mark=*] coordinates {(14.5,-7.5)}
node[label={[xshift=0mm,yshift=-7mm] $b_2$}] {};
 \draw [black] plot [only marks, mark size=3, mark=*] coordinates {(26.5,-7.5)}
node[label={[xshift=0mm,yshift=-7mm] $b_3$}] {};

 \foreach \i in {0,1,2}
 {
 \begin{scope}[shift={(12*\i,0)}]
 \foreach \i in {0,1,2,3}
 {
 \path (2.5,-7.5) node(a) {} (\i+1,-5) node(b) {};
         \draw[ultra thick,orange,onethirdarrow={>}] (b) -- (a);
 }
 \end{scope}
 }

 \draw [black] plot [only marks, mark size=3, mark=*] coordinates {(2.5,36.5)}
node[label={[xshift=0mm,yshift=0mm] $a_1$}] {};
 \draw [black] plot [only marks, mark size=3, mark=*] coordinates {(14.5,36.5)}
node[label={[xshift=0mm,yshift=0mm] $a_2$}] {};
 \draw [black] plot [only marks, mark size=3, mark=*] coordinates {(26.5,36.5)}
node[label={[xshift=0mm,yshift=0mm] $a_3$}] {};

 \foreach \i in {0,1,2}
 {
 \begin{scope}[shift={(12*\i,0)}]
 \foreach \i in {0,1,2,3}
 {
 \path (2.5,36.5) node(a) {} (\i+1,34) node(b) {};
         \draw[orange,ultra thick,endarrow={>}] (a) -- (b);
 }
 \end{scope}
 }


 \foreach \j in {1,2,3}
 {
 \foreach \i in {1,2,3,4}
 {
 \draw [green] plot [only marks, mark size=0, mark=*] coordinates
{(-15.5+12*\i,-13+12*\j)} node[label={[xshift=0mm,yshift=-6mm] $HS_{\i,\j}$}]
{};

 }
 }

 \foreach \j in {1,2,3,4}
 {
 \foreach \i in {1,2,3}
 {
 \draw [blue] plot [only marks, mark size=0, mark=*] coordinates
{(-11+12*\i,-18+12*\j)} node[label={[xshift=-12mm,yshift=-1mm] $VS_{\i,\j}$}]
{};
 }
 }

 \end{tikzpicture}

 \caption{A bird's-eye view of the instance of $G^*$ with $\ell=3$ and $n=4$ (see
Figure~\ref{fig:scssp-zoomed-in} for a zoomed-in view). From the secondary gadgets, we have shown only the base edges: the connector edges and bridge edges are omitted for clarity. Similarly,
the vertices and edges within each main gadget are not shown here either.
Additionally we have some red edges between each main gadget and the four
secondary gadgets surrounding it which are omitted in this figure for clarity
(they are shown in Figure~\ref{fig:scssp-zoomed-in}) which
gives a more zoomed-in view. Also missing in this picture are the two special terminals $s^*, t^*$ in addition to the source edges, sink edges and the strong edge.
  \label{fig:scssp-big-picture}}
 \end{figure}
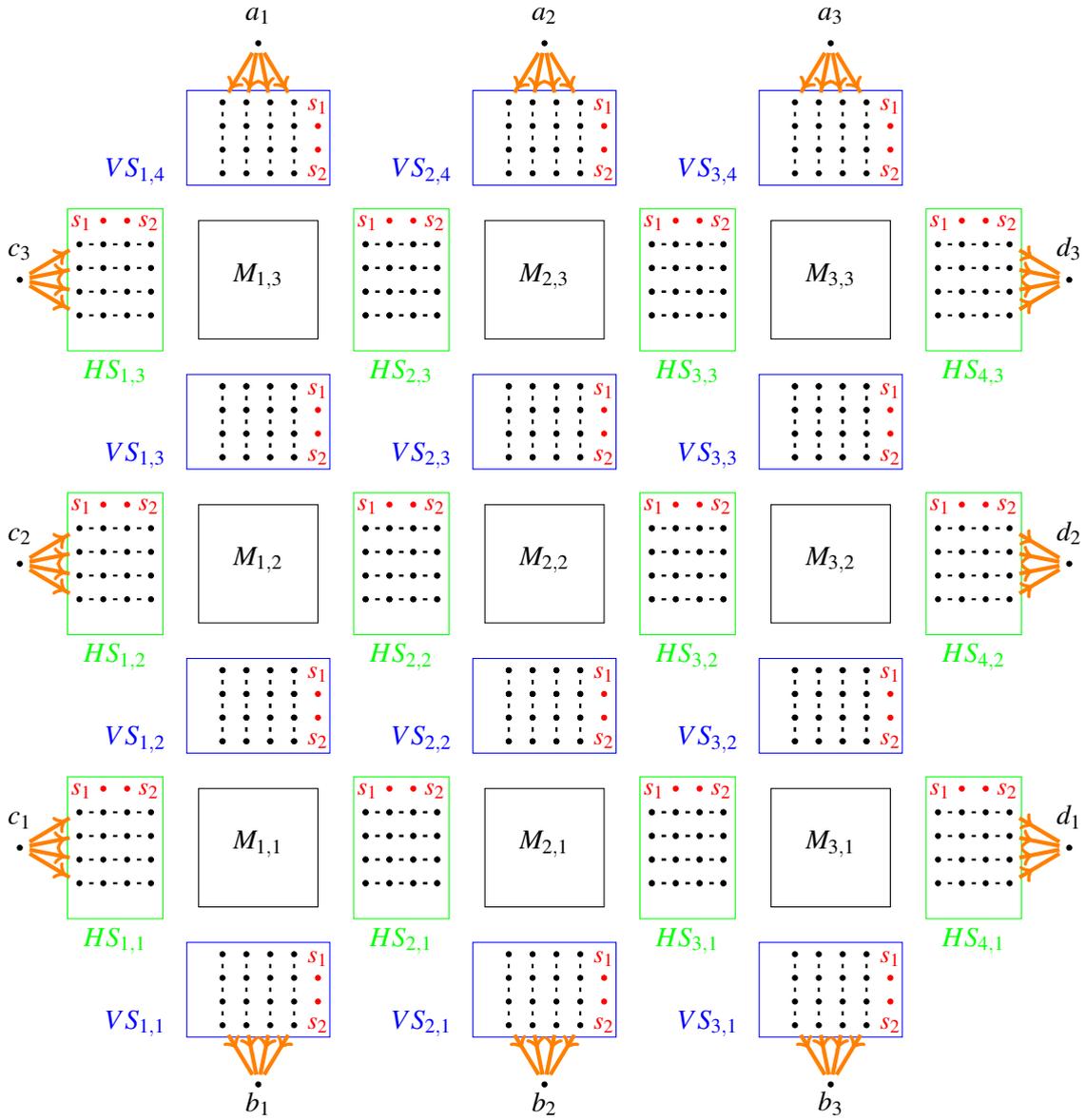


Fix some $1\leq i,j\leq \ell$. The main gadget $M_{i,j}$ has four secondary
gadgets\footnote{Half of the secondary gadgets are called ``horizontal'' since
their base edges are horizontal (as seen by the reader), and the other half of
the secondary gadgets are called ``vertical''.} surrounding it:
\begin{itemize}
\item Above $M_{i,j}$ is the vertical secondary gadget $VS_{i,j+1}$
\item On the right of $M_{i,j}$ is the horizontal secondary gadget $HS_{i+1,j}$
\item Below $M_{i,j}$ is the vertical secondary gadget $VS_{i,j}$
\item On the left of $M_{i,j}$ is the horizontal secondary gadget $HS_{i,j}$
\end{itemize}
Hence, there are $\ell(\ell+1)$ horizontal secondary gadgets and $\ell(\ell+1)$ vertical
secondary gadgets.
Recall that $M_{i,j}$ is a copy of $\mathcal{U}_{|S_{i,j}|}$ and each of the secondary
gadgets are copies of $\mathcal{U}_n$ (with $B=1$). With slight abuse of notation, we
assume that the rows of $M_{i,j}$ are indexed by the set $\{(x,y)\ :\ (x,y)\in
S_{i,j}\}$.
We add the following edges (in \textcolor[rgb]{1.00,0.00,0.00}{red} color) of
weight $1$: for each $1\leq i,j\leq \ell$ and each $(x,y)\in S_{i,j}$
\begin{itemize}
\item Add the edge $VS_{i,j+1}(3_x)\rightarrow M_{i,j}(0_{(x,y)})$. These edges
are called \emph{top-red} edges incident on $M_{i,j}$.
\item Add the edge $HS_{i,j}(3_y) \rightarrow M_{i,j}(0_{(x,y)})$. These edges
are called \emph{left-red} edges incident on $M_{i,j}$.
\item Add the edge $M_{i,j}(3_{(x,y)}) \rightarrow HS_{i+1,j}(0_y)$. These
edges are called \emph{right-red} edges incident on $M_{i,j}$.
\item Add the edge $M_{i,j}(3_{(x,y)}) \rightarrow VS_{i,j}(0_x) $. These edges
are called \emph{bottom-red} edges incident on $M_{i,j}$.
\end{itemize}
~\\
Introduce the following $4\ell$ vertices (which we call \emph{border} vertices):
\begin{itemize}
\item $a_1, a_2, \ldots, a_\ell$
\item $b_1, b_2, \ldots, b_\ell$
\item $c_1, c_2, \ldots, c_\ell$
\item $d_1, d_2, \ldots, d_\ell$
\end{itemize}
~\\
For each $i\in [\ell]$ add the following edges (shown using \textcolor{orange}{orange} color in Figure~\ref{fig:scssp-big-picture}) with weight 1:
\begin{itemize}
\item $a_{i} \rightarrow VS_{i,\ell+1}(0_{j})$ for each $j\in [n]$. We call
these edges \emph{top-orange} edges.
\item $b_{i} \rightarrow VS_{i,1}(3_{j})$ for each $j\in [n]$. We call these
edges \emph{bottom-orange} edges.
\item $c_{i} \rightarrow HS_{1,i}(0_{j})$ for each $j\in [n]$. We call these
edges \emph{left-orange} edges.
\item $d_{i} \rightarrow HS_{\ell+1,i}(3_{j})$ for each $j\in [n]$. We call
these edges \emph{right-orange} edges.
\end{itemize}
~\\
Introduce two new vertices $s^*, t^*$ and add the following edges (not shown in
Figure~\ref{fig:scssp-big-picture}) with weight~$1$ each:
\begin{itemize}
\item $s^* \rightarrow a_i$ and $s^* \rightarrow c_i$ for each $i\in [\ell]$. We
call these edges \emph{source} edges
\item $d_i\rightarrow t^*$ and $b_i\rightarrow t^*$ for each $i\in [\ell]$. We
call these edges \emph{sink} edges
\item The edge $t^* \rightarrow s^*$. We call this edge the \emph{strong} edge.
\end{itemize}

This completes the construction of the graph $G$. Note that $G$ has size $N=(n+\ell)^{O(1)}$ and can be constructed in $(n+\ell)^{O(1)}$ time.
%
%
%
%
%
%
%
%
%
%
%
%
%
 \begin{figure}[hp]

 \centering
 \begin{tikzpicture}[scale=0.5]



 \foreach \j in {0,1}
 {
 \begin{scope}[shift={(-21*\j,0)}]

 \foreach \j in {0,1.5,3,4.5}
 {
 \begin{scope}[shift={(0,\j)}]
 \foreach \i in {0,1,2,3}
 {
 \draw [black] plot [only marks, mark size=3, mark=*] coordinates {(12+2*\i,1.5)};
 }

 \foreach \i in {0,1,2}
 {
 \path (12+2*\i,1.5) node(a) {} (12+2*\i+2,1.5) node(b) {};
         \draw[ultra thick,black, middlearrow={>}] (a) -- (b);
 }
 \end{scope}
 }

 \draw [red] plot [only marks, mark size=3, mark=*] coordinates {(14,8)} ;
 \draw [red] plot [only marks, mark size=3, mark=*] coordinates {(16,8)} ;

 \foreach \i in {0}
 {
 \begin{scope}[shift={(2*\i,0)}]
 \path (14,8) node(a) {} (14,6) node(b) {};
  \draw [thick,dotted,middlearrow={>}] (a) to (b);

 \path (14,8) node(a) {} (14,4.5) node(b) {};
  \draw [thick,dotted,endarrow={>}] (a.north) to [out=135,in=135] (b.north);

 \path (14,8) node(a) {} (14,3) node(b) {};
  \draw [thick,dotted,endarrow={>}] (a.north) to [out=135,in=135] (b.north);

 \path (14,8) node(a) {} (14,1.5) node(b) {};
  \draw [thick,dotted,endarrow={>}] (a.north) to [out=135,in=135] (b.north);

%
%
%

 \end{scope}
 }

 \foreach \i in {0}
 {
 \begin{scope}[shift={(2*\i,0)}]
 \path (16,6) node(a) {} (16,8) node(b) {};
  \draw [thick,dotted,middlearrow={>}] (a) to (b);

 \path (16,4.5) node(a) {} (16,8) node(b) {};
  \draw [thick,dotted,startarrow={>}] (a.north) to [out=60,in=60] (b.north);

 \path (16,3) node(a) {} (16,8) node(b) {};
  \draw [thick,dotted,startarrow={>}] (a.north) to [out=60,in=60] (b.north);

 \path (16,1.5) node(a) {} (16,8) node(b) {};
  \draw [thick,dotted,startarrow={>}] (a.north) to [out=60,in=60] (b.north);
\end{scope}
}

 \end{scope}
 }

\path (16,8) node(a) {} (14,8) node(b) {};
  \draw [ultra thick,middlearrow={>}] (a) to (b);

\path (-5,8) node(a) {} (-7,8) node(b) {};
  \draw [ultra thick,middlearrow={>}] (a) to (b);


 \begin{scope}[shift={(-0.75,0)}]

 \foreach \j in {0,1}
 {
 \begin{scope}[shift={(0,21*\j)}]

 \foreach \j in {0,1.5,3,4.5}
 {
 \begin{scope}[shift={(\j,0)}]
 \foreach \i in {0,1,2,3}
 {
 \draw [black] plot [only marks, mark size=3, mark=*] coordinates {(3,-3-2*\i)};
 }

 \foreach \i in {0,1,2}
 {
 \path (3,-3-2*\i) node(a) {} (3,-3-2*\i-2) node(b) {};
         \draw[ultra thick,black,middlearrow={>}] (a) -- (b);
 }
 \end{scope}
 }

 \draw [red] plot [only marks, mark size=3, mark=*] coordinates {(9.5,-5)} ;
 \draw [red] plot [only marks, mark size=3, mark=*] coordinates {(9.5,-7)} ;

 \foreach \i in {0}
 {
 \begin{scope}[shift={(0,-2*\i)}]
 \path (9.5,-5) node(a) {} (7.5,-5) node(b) {};
  \draw [thick,dotted,middlearrow={>}] (a) to (b);

 \path (9.5,-5) node(a) {} (6,-5) node(b) {};
  \draw [thick,dotted,endarrow={>}] (a.north) to [out=135,in=30] (b.east);

 \path (9.5,-5) node(a) {} (4.5,-5) node(b) {};
  \draw [thick,dotted,endarrow={>}] (a.north) to [out=135,in=30] (b.east);

 \path (9.5,-5) node(a) {} (3,-5) node(b) {};
  \draw [thick,dotted,endarrow={>}] (a.north) to [out=135,in=30] (b.east);

%
%
%

 \end{scope}
 }

 \foreach \i in {0}
 {
 \begin{scope}[shift={(0,-2*\i)}]
 \path (7.5,-7) node(a) {} (9.5,-7) node(b) {};
  \draw [thick,dotted,middlearrow={>}] (a) to (b);

 \path (6,-7) node(a) {} (9.5,-7) node(b) {};
  \draw [thick,dotted,startarrow={>}] (a.east) to [out=-30,in=210] (b.south);

 \path (4.5,-7) node(a) {} (9.5,-7) node(b) {};
  \draw [thick,dotted,startarrow={>}] (a.east) to [out=-30,in=210] (b.south);

 \path (3,-7) node(a) {} (9.5,-7) node(b) {};
  \draw [thick,dotted,startarrow={>}] (a.east) to [out=-30,in=210] (b.south);

%
%
%

 \end{scope}
 }

 \end{scope}
 }

 \path (9.5,-7) node(a) {} (9.5,-5) node(b) {};
  \draw [ultra thick,middlearrow={>}] (a) to (b);

 \path (9.5,14) node(a) {} (9.5,16) node(b) {};
  \draw [ultra thick,middlearrow={>}] (a) to (b);

\end{scope}



 \begin{scope}[shift={(0,-0.75)}]
 \foreach \j in {0,1,2}
 {
 \begin{scope}[shift={(0,1.5*\j)}]
 \foreach \i in {0,1,2,3}
 {
\draw [black] plot [only marks, mark size=3, mark=*] coordinates {(1.5+2*\i,3)};
 }

 \foreach \i in {0,1,2}
 {
 \path (1.5+2*\i,3) node(a) {} (1.5+2*\i+2,3) node(b) {};
         \draw[ultra thick,black,middlearrow={>}] (a) -- (b);
 }
 \end{scope}
 }

 \draw [red] plot [only marks, mark size=3, mark=*] coordinates {(3.5,7.5)} ;
 \draw [red] plot [only marks, mark size=3, mark=*] coordinates {(5.5,7.5)} ;

 \foreach \i in {0}
 {
 \begin{scope}[shift={(2*\i,0)}]
 \path (3.5,7.5) node(a) {} (3.5,6) node(b) {};
  \draw [thick,dotted,->] (a) to (b);

 \path (3.5,7.5) node(a) {} (3.5,4.5) node(b) {};
  \draw [thick,dotted,->] (a.north) to [out=120,in=120] (b.north);

 \path (3.5,7.5) node(a) {} (3.5,3) node(b) {};
  \draw [thick,dotted,->] (a.north) to [out=120,in=120] (b.north);

%
%

 \end{scope}
 }

 \foreach \i in {0}
 {
 \begin{scope}[shift={(2*\i,0)}]
 \path (5.5,6) node(a) {} (5.5,7.5) node(b) {};
  \draw [thick,dotted,->] (a) to (b);

 \path (5.5,4.5) node(a) {} (5.5,7.5) node(b) {};
  \draw [thick,dotted,startarrow={>}] (a.north) to [out=60,in=60] (b.north);

 \path (5.5,3) node(a) {} (5.5,7.5) node(b) {};
  \draw [thick,dotted,startarrow={>}] (a.north) to [out=60,in=60] (b.north);

%
%

 \end{scope}
 }

 \end{scope}

  \path (5.5,6.75) node(a) {} (3.5,6.75) node(b) {};
\draw [ultra thick,middlearrow={>}] (a) to (b);


 \path  (-3,3) node(a) {} (1.5,3.75) node(b) {};
  \draw [ultra thick,red,middlearrow={>}] (a) to (b);

 \path (7.5,3.75) node(a) {} (12,3) node(b) {};
  \draw [ultra thick,red,middlearrow={>}] (a) to (b);

 \path (1.5,3.75) node(a) {} (5.25,12) node(b) {};
  \draw [ultra thick,red,middlearrow={<}] (a) .. controls (-2,9) .. (b);

 \path (7.5,4) node(a) {} (5,-3) node(b) {};
  \draw [ultra thick,red,middlearrow={>}] (a) .. controls (11,0) .. (b);



 \draw [green] plot [only marks, mark size=0, mark=*] coordinates {(-6,0)}
node[label={[xshift=0mm,yshift=-7mm] \LARGE{$HS_{i,j}$}}] {};

 \draw [green] plot [only marks, mark size=0, mark=*] coordinates {(15,0)}
node[label={[xshift=0mm,yshift=-7mm] \LARGE{$HS_{i+1,j}$}}] {};

 \draw [blue] plot [only marks, mark size=0, mark=*] coordinates {(4.5,-10)}
node[label={[xshift=0mm,yshift=-6mm] \LARGE{$VS_{i,j}$}}] {};

 \draw [blue] plot [only marks, mark size=0, mark=*] coordinates {(4.5,19)}
node[label={[xshift=0mm,yshift=-4mm] \LARGE{$VS_{i,j+1}$}}] {};

 \draw [black] plot [only marks, mark size=0, mark=*] coordinates {(4,0)}
node[label={[xshift=4mm,yshift=-1mm] \LARGE{$M_{i,j}$}}] {};

 \draw [black] plot [only marks, mark size=0, mark=*] coordinates {(1.75,3.75)}
node[label={[xshift=-10mm,yshift=-2mm] \footnotesize{$M_{i,j}(0_{x,y})$}}] {};

 \draw [black] plot [only marks, mark size=0, mark=*] coordinates {(7.75,3.75)}
node[label={[xshift=3mm,yshift=-1mm] \footnotesize{$M_{i,j}(3_{x,y})$}}] {};

 \draw [black] plot [only marks, mark size=0, mark=*] coordinates {(4.75,11.75)}
node[label={[xshift=6mm,yshift=-7mm] \footnotesize{$VS_{i,j+1}(3_{x})$}}] {};

 \draw [black] plot [only marks, mark size=0, mark=*] coordinates {(-2.75,3.25)}
node[label={[xshift=6mm,yshift=-7mm] \footnotesize{$HS_{i,j}(3_{y})$}}] {};

 \draw [black] plot [only marks, mark size=0, mark=*] coordinates {(5.25,-2.75)}
node[label={[xshift=-2mm,yshift=0mm] \footnotesize{$VS_{i,j}(0_{x})$}}] {};

 \draw [black] plot [only marks, mark size=0, mark=*] coordinates {(11.5,3)}
node[label={[xshift=-5mm,yshift=-6mm] \footnotesize{$HS_{i+1,j}(0_{y})$}}] {};

 \end{tikzpicture}

 \caption{A zoomed-in view of the main gadget $M_{i,j}$ surrounded by four
secondary gadgets: horizontal gadget $HS_{i,j+1}$ on the top, vertical gadget
$VS_{i,j}$ on the left, horizontal gadget $HS_{i,j}$ on the bottom and vertical
gadget $VS_{i+1,j}$ on the right. Each of the secondary gadgets is a copy of
the uniqueness gadget $\mathcal{U}_n$ (see Section~\ref{sec:scssp-uniqueness-gadget}) and the main
gadget $M_{i,j}$ is a copy of the uniqueness gadget $\mathcal{U}_{|S_{i,j}|}$. The only
inter-gadget edges are the red edges: they have one end-point in a main gadget
and the other end-point in a secondary gadget. We have shown four such red
edges which are introduced for every $(x,y)\in S_{i,j}$.
  \label{fig:scssp-zoomed-in}}
 \end{figure}
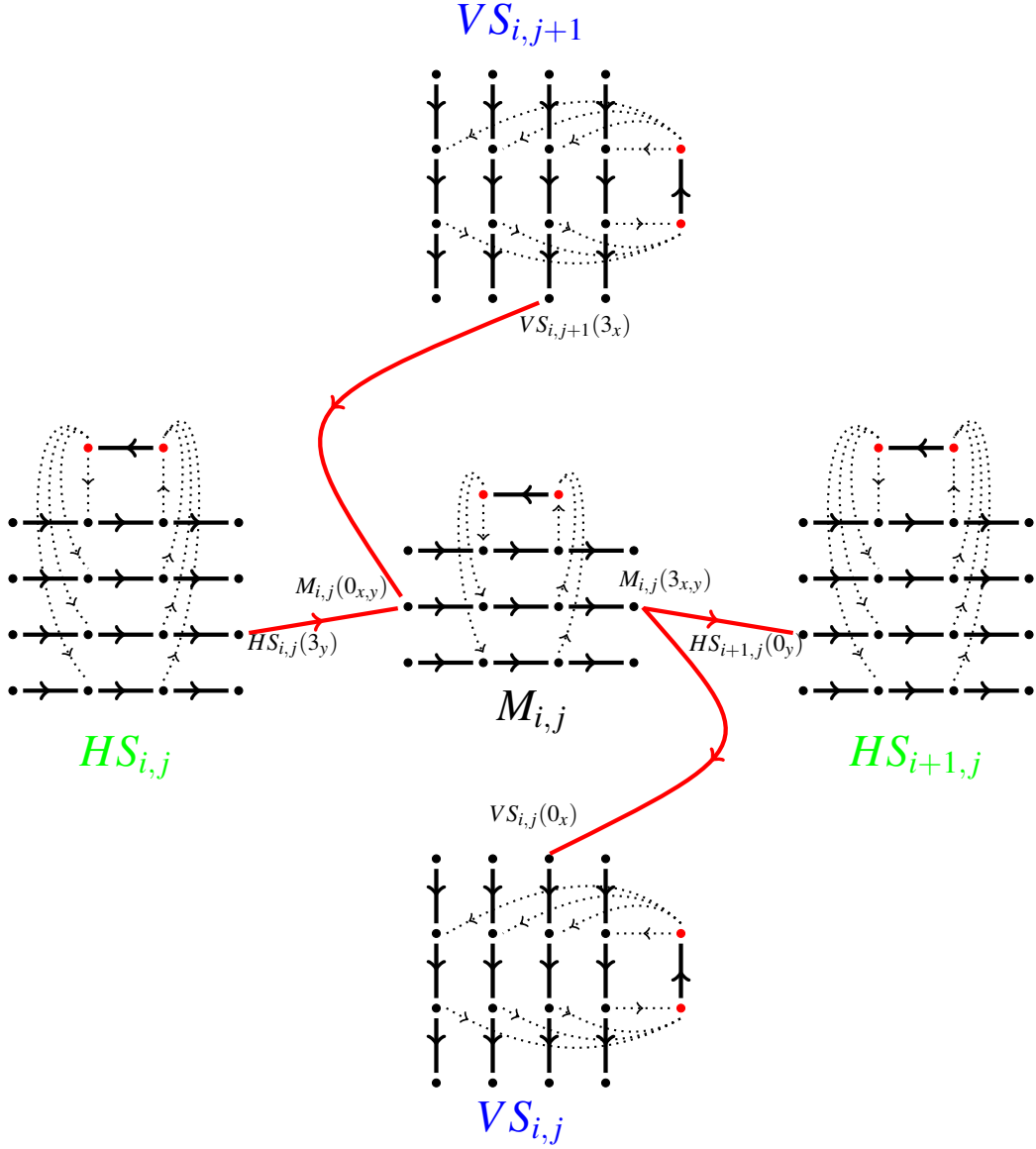
%
%
%
%
%
%
%
%
%
%
%
%
We now define the set of terminals $\mathcal{T}$ as follows:
\begin{itemize}
\item \underline{Vertical terminals}: The set of vertical terminals is given by $\bigcup_{1\leq i\leq \ell,1\leq j\leq \ell+1} \{VS_{i,j}(s_1), VS_{i,j}(s_2)\}$

\item \underline{Horizontal terminals}: The set of horizontal terminals is given by $\bigcup_{1\leq i\leq \ell+1,1\leq j\leq \ell} \{HS_{i,j}(s_1), HS_{i,j}(s_2)\}$

\item \underline{Main terminals}: The set of main terminals is given by $\bigcup_{1\leq i,j\leq \ell} \{M_{i,j}(s_1), M_{i,j}(s_2)\}$

\item \underline{Special terminals}: There are only two special terminals,
namely $s^*$ and $t^*$.
\end{itemize}

We have $\ell^2$ main gadgets, $\ell(\ell+1)$ vertical secondary gadgets and
$\ell(\ell+1)$ horizontal secondary gadgets. In addition to the two special
terminals $s^*$ and $t^*$, we add two terminals corresponding to each of these
gadgets. Hence, the total number of terminals is $k=|\mathcal{T}|=O(\ell^2)$.

Fix the budget $B^*= 1+ 20\ell + 24\ell^{2} = O(\ell^2)$. We will show that any
solution has weight at least $B^*$ using the following intuition:
\begin{itemize}
  \item The strong edge $t^* \rightarrow s^*$ must be present
  \item $2\ell$ source edges and $2\ell$ sink edges must be present
  \item $4\ell$ orange edges must be present (one for each boundary vertex)
  \item Each of the $\ell(\ell+1)$ vertical secondary gadgets must satisfy ``in-out" property and hence have weight at least $6$
  \item Each of the $\ell(\ell+1)$ horizontal secondary gadgets must satisfy ``in-out" property and hence have weight at least $6$
  \item Each of the $\ell^2$ main gadgets must satisfy ``in-out" property and hence have weight at least $6$
  \item Each of the $\ell^2$ main gadgets must contribute at least four red edges (each of which have exactly one endpoint in the main gadget)
\end{itemize}
In the other direction, we will show that a solution having cost exactly $B^*$ forces enough structure to allow to us to conclude that \gt answers YES.


\subsection{\textcolor[rgb]{0.00,0.00,0.00}{\mbox{{\altgt} answers YES $\Rightarrow$ instance $(G,\mathcal{T})$ of {\altscssP} has a planar}
solution of cost $\leq B^*$}}
\label{sec:easy-gt}

Suppose that \gt has a solution, i.e., for each $1\leq i,j\leq \ell$ there is a value $(x_{i,j}, y_{i,j})=\gamma_{i,j}\in S_{i,j}$ such that
\begin{itemize}
\item for every $i\in [\ell]$, we have $x_{i,1}=x_{i,2}=x_{i,3}=\ldots=x_{i,\ell} =\alpha_i$, and
\item for every $j\in [\ell]$, we have $y_{1,j}=y_{2,j}=y_{3,j}=\ldots=y_{\ell,j} =\beta_j$.
\end{itemize}
~\\
We now build a \emph{planar} solution $N$ for the instance $(G^*, \mc{T})$ of
\scss. In the edge set $N$, we
take the following edges:
\begin{enumerate}
\item The strong edge $t^* \rightarrow s^*$. This incurs a cost of $1$.

\item For each $i\in [\ell]$ add the source edges $s^* \rightarrow a_i$ and $s^* \rightarrow c_i$. This incurs a cost of $2\ell$ since each source edge has weight $1$.

\item For each $i\in [\ell]$ add the sink edges $b_i \rightarrow t^*$ and
$d_i \rightarrow t^*$. This incurs a cost of $2\ell$ since each sink edge has
weight $1$.

\item For each $i\in [\ell]$ add the top orange edge $a_i \rightarrow VS_{i,\ell+1}(0_{\alpha_i})$ and the bottom orange edge $VS_{i,1}(0_{\alpha_i})\rightarrow b_i$. This incurs a cost of $2\ell$ since each orange edge has weight $1$.

\item For each $j\in [\ell]$ add the left orange edge $c_j \rightarrow HS_{1,j}(0_{\alpha_i})$ and the right orange edge $HS_{\ell+1,j}(0_{\alpha_i})\rightarrow d_j$. This incurs a cost of $2\ell$ since each orange edge has weight $1$.

\item For each $1\leq i\leq \ell+1, 1\leq j\leq \ell$, use Corollary~\ref{crl:macro} to add the set of edges $E_{HS_{i,j}}(\beta_j)$ from $HS_{i,j}$ of weight $6$ which represents $\beta_j$. This incurs a cost of $6\ell(\ell+1)$ since there are $\ell(\ell+1)$ horizontal secondary gadgets.

\item For each $1\leq i\leq \ell, 1\leq j\leq \ell+1$, use Corollary~\ref{crl:macro} to add the set of edges $E_{VS_{i,j}}(\alpha_i)$ from $VS_{i,j}$ of weight $6$ which represents $\alpha_i$. This incurs a cost of $6\ell(\ell+1)$ since there are $\ell(\ell+1)$ vertical secondary gadgets.

\item For each $1\leq i,j\leq \ell$, use Corollary~\ref{crl:macro} to add the set of edges $E_{M_{i,j}}((\alpha_i,\beta_j))$ from $M_{i,j}$ of weight $6$ which represents $(\alpha_i,\beta_j)$. Note that this is possible since the solution of the \gt instance guarantees that $(\alpha_i,\beta_j)\in S_{i,j}$ for each $1\leq i,j\leq \ell$. This incurs a cost of $6\ell^2$ since there are $\ell^2$ main gadgets.

\item For each $1\leq i,j\leq \ell$, add the four edges (each of which has weight $1$)
            \begin{itemize}
              \item $VS_{i,j+1}(3_{\alpha_i})\rightarrow M_{i,j}(0_{(\alpha_i,\beta_j)})$
              \item $HS_{i,j}(3_{\beta_j}) \rightarrow M_{i,j}(0_{(\alpha_i,\beta_j)})$
              \item $M_{i,j}(3_{(\alpha_i,\beta_j)}) \rightarrow
HS_{i+1,j}(0_{\alpha_i})$
              \item $M_{i,j}(3_{(\alpha_i,\beta_j)}) \rightarrow
VS_{i,j}(0_{\beta_j})$
            \end{itemize}
         Note that this is possible since the solution of the \gt instance guarantees that $(\alpha_i,\beta_j)\in S_{i,j}$ for each $1\leq i,j\leq \ell$. This incurs a cost of $4\ell^2$ since there are $\ell^2$ main gadgets.

%
%
%
\end{enumerate}

It follows that the weight of $N$ is exactly $ 1+ 4\ell+ 4\ell + 6\ell(\ell+1) +
6\ell(\ell+1) + 6\ell^{2} + 4\ell^{2}= B^*$. We next show that $N$ is in fact
a planar solution of the instance $(G,\mathcal{T})$ of \scss.

\subsubsection{$N$ is planar}

We use the following two definitions to help us argue about planarity of $N$;

\begin{definition}
We call the set of edges $E_{\text{intra}}$ which have both end-points in the same gadget (either main, vertical secondary or vertical secondary) as intra-gadget edges. We call the set of edges $E_{\text{inter}}$ which one end-point in a main gadget and other end-point in a  secondary gadget as inter-gadget edges.
\end{definition}

The source edges, sinks edges and the strong edge can be drawn on the ``outside"
of Figure~\ref{fig:scssp-big-picture}. Hence, Figure~\ref{fig:scssp-big-picture}
gives a planar embedding of $G\setminus (E_{\text{intra}}\cup
E_{\text{inter}})$. We now consider the set of edges $E_{\text{intra}}\cap
E(N)$. For any gadget (either main, horizontal secondary or vertical
secondary), Corollary~\ref{crl:macro} implies that the edges which have both
end-points in this gadget form a planar graph (see
Observation~\ref{obs:represented-edges-form-planar-subgraph}). Finally, we now
consider the set of edges $E_{\text{inter}}\cap N$. For each main gadget $M$,
there are exactly four inter-gadget edges incident on $M$: one each to the four
secondary gadgets surrounding it. These four (red) inter-gadget edges do not
destroy planarity either since two of them are incident on one $0$-vertex of the
main gadget and other two are incident on another $3$-vertex of the main gadget
(see Figure~\ref{fig:scssp-zoomed-in}).
Hence, $N$ is planar.

%
%

\subsubsection{$N$ is a solution for the instance $(G,\mathcal{T})$ of \scss}

Note that there are only two special terminals: $s^*$ and $t^*$. For each
non-special terminal $x$, i.e., $x\in \mathcal{T}\setminus \{s^*, t^*\}$, we
will show below that $N$ contains an $s^* \leadsto x$ path\footnote{Here, by
paths we technically mean walks since vertices and edges may repeat (all we care
about is directed connectivity)} and an $x\leadsto t^*$ path. Since $(t^*,
s^*)\in N$, this is sufficient to show that $N$ is indeed a solution for the
instance $(G,\mathcal{T})$ of \scss because:
\begin{itemize}
  \item The strong edge $(t^*, s^*)$ gives a $t^* \rightarrow s^*$ path. There
are many $s^* \leadsto t^*$ paths: choose any non-special terminal and
concatenate the $s^* \leadsto x$ path with the $x\leadsto t^*$ path.
  \item Let $x$ be any non-special terminal. Then we are guaranteed existence of
an $s^*\leadsto x$ path and an $x\leadsto t^*$ path. Since the strong edge
$(t^*, s^*)$ is in $N$ it follows that there also exists an $x\leadsto s^*$
path and an $t^* \leadsto x$ path.
  \item For any two terminals $x,y$ there is an $x\leadsto y$ path as follows:
take the $x\leadsto t^*$ path followed by the strong edge $(t^*, s^*)$ followed
by the $s^* \leadsto y$ path.
\end{itemize}


Hence, it remains to show that for any non-special terminal $x$ the set $N$
contains an $s^* \leadsto x$ path and an $x\leadsto t^*$ path. We have three
cases:
\begin{enumerate}
  \item \textbf{$x$ is a vertical terminal}: Suppose $x$ is a terminal in $VS_{i,j}$ for some $1\leq i\leq \ell, 1\leq j\leq \ell+1$.
                    \begin{itemize}
                      \item We first show the existence of an $VS_{i,j}(s_2)\leadsto t^*$ path which also contains the vertex $VS_{i,j}(s_1)$:
                                \begin{itemize}
                                  \item $VS_{i,j}(s_2)\rightarrow VS_{i,j}(s_1) $
                                  \item $VS_{i,j}(s_1)\rightarrow VS_{i,j}(1_{\alpha_i})\rightarrow VS_{i,j}(2_{\alpha_i})\rightarrow VS_{i,j}(3_{\alpha_i})$
                                  \item If $j\neq 1$, then use the $VS_{i,j}(3_{\alpha_i})\leadsto VS_{i,j-1}(3_{\alpha_i})$  path given by concatenating the following edges/paths
                                                \begin{itemize}
                                                  \item $VS_{i,j}(3_{\alpha_i})\rightarrow M_{i,j-1}(0_{\alpha_i, \beta_{j-1}})$
                                                  \item The $ M_{i,j-1}(0_{\alpha_i, \beta_{j-1}})\leadsto M_{i,j-1}(3_{\alpha_i, \beta_{j-1}})$ path given by $P_{M_{i,j-1}}(\alpha_i,\beta_j)$
                                                  \item $M_{i,j-1}(3_{\alpha_i, \beta_{j-1}})\rightarrow VS_{i,j-1}(0_{\alpha_i})$
                                                  \item The $VS_{i,j-1}(0_{\alpha_i})\leadsto VS_{i,j-1}(3_{\alpha_i})$ given by $P_{VS_{i,j-1}}(\alpha_i)$
                                                \end{itemize}
                                  We do this for each $j$ (decreasing $j$ by 1
each time) each time until we reach the vertex $VS_{i,1}(0_{\alpha_i})$
                                  \item Then use the path $VS_{i,1}(0_{\alpha_i})\rightarrow VS_{i,1}(1_{\alpha_i})\rightarrow VS_{i,1}(2_{\alpha_i})\rightarrow VS_{i,1}(3_{\alpha_i})\rightarrow b_i \rightarrow t^*$
                                \end{itemize}
                      \item We now show existence of an $s^* \leadsto VS_{i,j}(s_1)$ path which also contains the vertex $VS_{i,j}(s_2)$:
                                \begin{itemize}
                                  \item $s^* \rightarrow a_i \rightarrow VS_{i,\ell+1}(0_{\alpha_i})$
                                  \item If $j\neq \ell+1$, then use the $VS_{i,j}(3_{\alpha_i})\leadsto VS_{i,j-1}(3_{\alpha_i})$  path given by concatenating the following edges/paths
                                                \begin{itemize}
                                                  \item $VS_{i,j}(3_{\alpha_i})\rightarrow M_{i,j-1}(0_{\alpha_i, \beta_{j-1}})$
                                                  \item The $ M_{i,j-1}(0_{\alpha_i, \beta_{j-1}})\leadsto M_{i,j-1}(3_{\alpha_i, \beta_{j-1}})$ path given by $P_{M_{i,j-1}}(\alpha_i,\beta_j)$
                                                  \item $M_{i,j-1}(3_{\alpha_i, \beta_{j-1}})\rightarrow VS_{i,j-1}(0_{\alpha_i})$
                                                  \item The $VS_{i,j-1}(0_{\alpha_i})\leadsto VS_{i,j-1}(3_{\alpha_i})$ given by $P_{VS_{i,j-1}}(\alpha_i)$
                                                \end{itemize}

                                  \item $VS_{i,j}(s_2)\rightarrow VS_{i,j}(s_1) $
                                  \item $VS_{i,j}(s_1)\rightarrow VS_{i,j}(1_{\alpha_i})\rightarrow VS_{i,j}(2_{\alpha_i})\rightarrow VS_{i,j}(3_{\alpha_i})$
                                  \item If $j\neq 1$, then use the $VS_{i,j}(3_{\alpha_i})\leadsto VS_{i,j-1}(3_{\alpha_i})$  path given by concatenating the following edges/paths
                                                \begin{itemize}
                                                  \item $VS_{i,j}(3_{\alpha_i})\rightarrow M_{i,j-1}(0_{\alpha_i, \beta_{j-1}})$
                                                  \item The $ M_{i,j-1}(0_{\alpha_i, \beta_{j-1}})\leadsto M_{i,j-1}(3_{\alpha_i, \beta_{j-1}})$ path given by $P_{M_{i,j-1}}(\alpha_i,\beta_j)$
                                                  \item $M_{i,j-1}(3_{\alpha_i, \beta_{j-1}})\rightarrow VS_{i,j-1}(0_{\alpha_i})$
                                                  \item The $VS_{i,j-1}(0_{\alpha_i})\leadsto VS_{i,j-1}(3_{\alpha_i})$ given by $P_{VS_{i,j-1}}(\alpha_i)$
                                                \end{itemize}
                                  We do this for each $j$ (decreasing $j$ by 1
each time) each time until we reach the vertex $VS_{i,1}(0_{\alpha_i})$
                                  \item Then use the path $VS_{i,1}(0_{\alpha_i})\rightarrow VS_{i,1}(1_{\alpha_i})\rightarrow VS_{i,1}(2_{\alpha_i})\rightarrow VS_{i,1}(3_{\alpha_i})\rightarrow b_i \rightarrow t^*$
                                \end{itemize}
                    \end{itemize}

  \item \textbf{$x$ is a horizontal terminal}: Suppose $x$ is a terminal in $HS_{i,j}$ for some $1\leq i\leq \ell+1, 1\leq j\leq \ell$.
                    \begin{itemize}
                      \item We first show the existence of an $HS_{i,j}(s_2)\leadsto t^*$ path which also contains the vertex $HS_{i,j}(s_1)$. If $i=\ell+1$ then we can use the path $HS_{\ell+1,j}(s_2)\rightarrow
                      HS_{\ell+1,j}(s_1)\rightarrow HS_{\ell+1,j}(1_{\beta_j})\rightarrow HS_{\ell+1,j}(2_{\beta_j})\rightarrow HS_{\ell+1,j}(3_{\beta_j})\rightarrow d_j \rightarrow t^*$. Otherwise if $i<\ell+1$ then we use the path obtained by concatenating the following paths (in order):
                                \begin{itemize}
                                  \item $HS_{i,j}(s_2)\rightarrow HS_{i,j}(s_1) $
                                  \item $HS_{i,j}(s_1)\rightarrow HS_{i,j}(1_{\beta_j})\rightarrow HS_{i,j}(2_{\beta_j})\rightarrow HS_{i,j}(3_{\beta_j})$
                                  \item $HS_{i,j}(3_{\beta_j})\rightarrow M_{i,j}(0_{\alpha_i, \beta_j})\rightarrow M_{i,j}(1_{\alpha_i, \beta_j})\rightarrow M_{i,j}(2_{\alpha_i, \beta_j})\rightarrow M_{i,j}(s_2)$
                                  \item Now use the $M_{i,j}(s_2)\leadsto t^*$ path guaranteed by Case 1
                                \end{itemize}
                      \item We now show existence of an $s^* \leadsto HS_{i,j}(s_1)$ path which also contains the vertex $HS_{i,j}(s_2)$. If $i=1$ then we can use the path $s^* \rightarrow c_j \rightarrow HS_{1,j}(0_{\beta_j})\rightarrow HS_{1,j}(1_{\beta_j})\rightarrow HS_{1,j}(2_{\beta_j})\rightarrow HS_{1,j}(s_2)\rightarrow HS_{1,j}(s_1)$. Otherwise if $i>1$ then we can use the path obtained by concatenating the following paths (in order):
                                \begin{itemize}
                                   \item The $s^* \leadsto VS_{i-1,j+1}(s_1)$ path guaranteed by Case 1
                                   \item $VS_{i-1,j+1}(s_1)\leadsto VS_{i-1,j+1}(1_{\alpha_{i-1}})\rightarrow VS_{i-1,j+1}(2_{\alpha_{i-1}})\rightarrow VS_{i-1,j+1}(3_{\alpha_{i-1}})$
                                    \item $VS_{i-1,j+1}(3_{\alpha_{i-1}})\rightarrow MS_{i-1,j}(0_{\alpha_{i-1},\beta_j})\rightarrow MS_{i-1,j}(1_{\alpha_{i-1},\beta_j})\rightarrow MS_{i-1,j}(2_{\alpha_{i-1},\beta_j})\rightarrow MS_{i-1,j}(3_{\alpha_{i-1},\beta_j})$

                                    \item $MS_{i-1,j}(0_{\alpha_{i-1},\beta_j})\rightarrow HS_{i,j}(0_{\beta_j})\rightarrow HS_{i,j}(1_{\beta_j})\rightarrow HS_{i,j}(2_{\beta_j})\rightarrow HS_{i,j}(s_2)\rightarrow HS_{i,j}(s_1)$

                                \end{itemize}
                    \end{itemize}

  \item \textbf{$x$ is a main terminal}: Suppose $x$ is a terminal in $MS_{i,j}$ for some $1\leq i,j\leq \ell$.
                    \begin{itemize}
                      \item We first show the existence of an $M_{i,j}(s_2)\leadsto t^*$ path which also contains the vertex $M_{i,j}(s_1)$. This path is obtained by concatenating the following paths (in order):
                                \begin{itemize}
                                  \item $M_{i,j}(s_2)\rightarrow M_{i,j}(s_1)\rightarrow M_{i,j}(1_{\alpha_i,\beta_j})\rightarrow M_{i,j}(2_{\alpha_i,\beta_j})\rightarrow M_{i,j}(3_{\alpha_i,\beta_j})$
                                  \item $M_{i,j}(3_{\alpha_i,\beta_j})\rightarrow HS_{i+1,j}(0_{\beta_j})\rightarrow HS_{i+1,j}(1_{\beta_j})\rightarrow HS_{i+1,j}(2_{\beta_j})\rightarrow HS_{i+1,j}(s_2)$
                                  \item The $HS_{i+1,j}(s_2)\leadsto t^*$ path guaranteed by Case 2
                                \end{itemize}

                      \item We now show the existence of an $s^* \leadsto M_{i,j}(s_1)$ path which also contains the vertex $M_{i,j}(s_2)$. This path is obtained by concatenating the following paths (in order):
                                \begin{itemize}
                                  \item The $s^* \leadsto HS_{i,j}(s_1)$ path guaranteed by Case 2
                                  \item $HS_{i,j}(s_1)\rightarrow HS_{i,j}(1_{\beta_j})\rightarrow HS_{i,j}(2_{\beta_j})\rightarrow HS_{i,j}(3_{\beta_j})\rightarrow M_{i,j}(0_{\alpha_i,\beta_j})$
                                  \item $M_{i,j}(0_{\alpha_i,\beta_j})\rightarrow M_{i,j}(1_{\alpha_i,\beta_j})\rightarrow M_{i,j}(2_{\alpha_i,\beta_j})\rightarrow M_{i,j}(s_2)\rightarrow M_{i,j}(s_1)$
                                \end{itemize}

                    \end{itemize}
\end{enumerate}

\subsection{\textcolor[rgb]{0.00,0.00,0.00}{Instance $(G,\mathcal{T})$ of \mbox{{\altscssP} has a solution of cost $\leq B^*$
$\Rightarrow$ {\altgt}} answers YES}}
\label{sec:hard-gt}


\textcolor[rgb]{0.00,0.00,0.00}{Suppose that the instance $(G,\mathcal{T})$ of \scssP has a solution $N$ of cost at most $B^*$. We will now show that this
implies that \gt answers YES. This implies that if \gt answers NO then the cost
of an optimal solution (and hence the cost of an optimal planar solution, if
one exists) is greater than $B^*$.}

\begin{lemma}
$N$ contains the strong edge $(t^*, s^*)$
\label{lem:strong-edge}
\end{lemma}
\begin{proof}
The vertex $t^*$ is a terminal and the only outgoing edge incident on $t^*$ is the strong edge $(t^*, s^*)$.
\end{proof}

\begin{lemma}
$N$ contains at least $4\ell$ orange edges. In fact, for each $1\leq i\leq \ell$
we have that $N$ contains at least one
\begin{itemize}
\item outgoing orange edge from $a_i$
\item incoming orange edge into $b_i$
\item outgoing orange edge from $c_j$
\item incoming orange edge into $d_j$
\end{itemize}
\label{lem:orange-planar}
\end{lemma}
\begin{proof}
Fix $i\in [\ell]$. The gadget $VS_{i,\ell+1}$  has two terminals
$VS_{i,\ell+1}(s_1)$ and $VS_{i,\ell+1}(s_2)$, and the only incoming edges into
the gadget $VS_{i,\ell+1}$ are the top orange edges outgoing from $a_i$. Hence,
for strong connectivity it follows that $N$ contains at least one top-orange
edge outgoing from $a_i$ for each $i\in [\ell]$, i.e., $N$ contains at least
$\ell$ top-orange edges.

The other three claims follow by similar arguments.
\end{proof}

\begin{lemma}
$N$ contains each of the $2\ell$ source edges and each of the $2\ell$ sink edges
\label{lem:source-sink-planar}
\end{lemma}
\begin{proof}
Fix $i\in [\ell]$. The gadget $VS_{i,\ell+1}$  has two terminals
$VS_{i,\ell+1}(s_1)$ and $VS_{i,\ell+1}(s_2)$, and the only incoming edges into
the gadget $VS_{i,\ell+1}$ are the top orange edges outgoing from $a_i$.
Moreover, the only incoming edge into $a_i$ is the source edge $(s^*, a_i)$.
Hence, for strong connectivity it follows that $N$ contains the source edge
$(s^*, a_i)$ for each $i\in [\ell]$. Similarly, one can show that $N$ contains
all the other source edges and sink edges as well.
\end{proof}

\begin{lemma}
For every $1\leq i,j\leq \ell$ the edge set $N$ restricted to the main gadget
$M_{i,j}$ satisfies the ``in-out'' property. Hence, $N$ has weight at least $6$
in $M_{i,j}$.
\label{lem:in-out-main}
\end{lemma}
\begin{proof}
The main gadget $M_{i,j}$  has two terminals $M_{i,j}(s_1)$ and $M_{i,j}(s_2)$.
The only incoming edges into $M_{i,j}$ are the top-red and left-red edges which
are incident on the $0$-vertices of $M_{i,j}$. Hence, each terminal of $M_{i,j}$
has to be reachable from some $0$-vertex of $M_{i,j}$. Similarly, the only
outgoing edges from $M_{i,j}$ are the bottom-red and right-red edges which are
incident on the $3$-vertices of $M_{i,j}$. Hence, each terminal of $M_{i,j}$ has
to be able to reach some $3$-vertex of $M_{i,j}$. Hence, $N$ restricted to the
main gadget $M_{i,j}$ satisfies the ``in-out'' property (recall
Definition~\ref{defn-in-out}). By  Lemma~\ref{lem:macro-uniqueness-gadget}, the
claim follows.
\end{proof}

Analogous lemmas hold also for the horizontal secondary gadgets and the
vertical secondary gadgets:
\begin{lemma}
For every $1\leq i\leq \ell, 1\leq j\leq \ell+1$ the edge set $N$ restricted to
the horizontal secondary gadget $HS_{i,j}$ satisfies the ``in-out'' property.
Hence, $N$ has weight at least $6$ in $HS_{i,j}$.
\label{lem:in-out-horizontal}
\end{lemma}

\begin{lemma}
For every $1\leq i\leq \ell+1, 1\leq j\leq \ell$ the edge set $N$ restricted to
the vertical secondary gadget $VS_{i,j}$ satisfies the ``in-out'' property.
Hence, $N$ has weight at least $6$ in $VS_{i,j}$.
\label{lem:in-out-vertical}
\end{lemma}

\begin{lemma}
For each $1\leq i,j\leq \ell$, the solution $N$ contains at least one
\begin{itemize}
  \item top-red edge incident on $M_{i,j}$
  \item right-red edge incident on $M_{i,j}$
  \item bottom-red edge incident on $M_{i,j}$
  \item left-red edge incident on $M_{i,j}$
\end{itemize}
\label{lem:main-at-least-4-red-planar}
\end{lemma}
\begin{proof}
Fix some $1\leq i,j\leq \ell$. We now show that $N$ contains a top-red edge
incident on $M_{i,j}$ (the other 3 claims can be shown analogously). The
vertical gadget $VS_{i,j+1}$ has two terminals $VS_{i,j+1}(s_1)$ and
$VS_{i,j+1}(s_2)$. The only outgoing edges incident on $VS_{i,j+1}$ are the
top-red edges incident on $M_{i,j}$. Hence, strong connectivity of $N$ implies
that it contains at least one top-red edge incident on $M_{i,j}$.
\end{proof}

We show now that we have no slack, i.e., the weight of $N$ must be exactly
$B^*$.

\begin{lemma}
The weight of $N$ is exactly $B^*$, and hence it is minimal (under edge
deletions) since no edges have zero weight.
\label{lem:exact-B^*-planar}
\end{lemma}
\begin{proof}
We have the following collection of pairwise disjoint sets of edges which are
guaranteed to be contained in $N$
\begin{itemize}
  \item The strong edge $(t^*, s^*)$ (from Lemma~\ref{lem:strong-edge}). This incurs a cost of $1$.
  \item $4\ell$ orange edges (from Lemma~\ref{lem:orange-planar}). This incurs a cost of $4\ell$.
  \item $2\ell$ sources edges and $2\ell$ sink edges (from Lemma~\ref{lem:source-sink-planar}). This incurs a cost of $4\ell$.
  \item A cost of at least $6$ from edges which have both endpoints in each
vertical secondary gadget (from Lemma~\ref{lem:in-out-vertical}). This incurs a
cost of $6\ell(\ell+1)$
  \item A cost of at least $6$ from edges which have both endpoints in each
horizontal secondary gadget (from Lemma~\ref{lem:in-out-horizontal}). This
incurs a cost of $6\ell(\ell+1)$
  \item A cost of at least $6$ from edges which have both endpoints in each main
gadget (from Lemma~\ref{lem:in-out-main}). This incurs a cost of $6\ell^2$
  \item A cost of at least $4$ from edges which have exactly one endpoint in
each main gadget (from Lemma~\ref{lem:main-at-least-4-red-planar}). This incurs
a cost of $4\ell^2$
\end{itemize}

Hence, the cost of $N$ is at least
$1+4\ell+4\ell+6\ell(\ell+1)+6\ell(\ell+1)+6\ell^2 + 4\ell^2 = B^*$. But, we are
given that cost of $N$ is at most $B^*$. Hence, cost of $N$ is exactly $B^*$.
\end{proof}

The following corollary follows from Lemma~\ref{lem:exact-B^*-planar} and Lemma~\ref{lem:macro-uniqueness-gadget}:
\begin{corollary}
The weight of $N$ restricted to each gadget (main, vertical secondary or
horizontal secondary) is exactly $6$. Moreover,
\begin{itemize}

\item for each $1\leq i\leq \ell+1, 1\leq j\leq \ell$, the vertical secondary
gadget $VS_{i,j}$ is represented by some $x_{i,j}\in [n]$,

\item for each $1\leq i\leq \ell, 1\leq j\leq \ell+1$, the horizontal secondary
gadget $HS_{i,j}$ is represented by some $y_{i,j}\in [n]$,

\item for each $1\leq i, j\leq \ell$, the main gadget $M_{i,j}$ is represented
by some $(\lambda_{i,j}, \delta_{i,j})\in S_{i,j}$.
\end{itemize}
\label{crl:gadgets-representation}
\end{corollary}

The following corollary follows from Lemma~\ref{lem:exact-B^*-planar} and Lemma~\ref{lem:main-at-least-4-red-planar}:
\begin{corollary}
For each $1\leq i,j\leq \ell$, the solution $N$ contains exactly one
\begin{itemize}
  \item top-red edge incident on $M_{i,j}$
  \item right-red edge incident on $M_{i,j}$
  \item bottom-red edge incident on $M_{i,j}$
  \item left-red edge incident on $M_{i,j}$
\end{itemize}
\label{crl:main-exactly-4-red-planar}
\end{corollary}

Consider a main gadget $M_{i,j}$. The main gadget has four secondary gadgets
surrounding
it: $VS_{i,j}$ below it, $VS_{i,j+1}$ above it, $HS_{i,j}$ to the left and
$HS_{i+1,j}$ to the right. By Corollary~\ref{crl:gadgets-representation}, these gadgets
are represented by $x_{i,j}, x_{i,j+1}, y_{i,j}$ and $y_{i+1,j}$ respectively.
The main gadget $M_{i,j}$ is represented by $(\lambda_{i,j}, \delta_{i,j})$.

\begin{lemma}
(\textbf{propagation}) For every main gadget $M_{i, j}$, we have
$x_{i,j}=\lambda_{i,j}=x_{i,j+1}$ and $y_{i,j}=\delta_{i,j}=y_{i+1,j}$.
\label{lem:agreement-tight-planar}
\end{lemma}
\begin{proof}
Due to symmetry, it suffices to only argue that $x_{i,j}=\lambda_{i,j}$. By
Corollary~\ref{crl:gadgets-representation}, the main gadget $M_{i,j}$ is
represented by $(\lambda_{i,j}, \delta_{i,j})$ and the vertical secondary gadget
is represented by $x_{i,j}$. Hence, it follows that (recall
Definition~\ref{defn-representation}) the only $2-3$ edge from $M_{i,j}$ in $N$
is $M_{i,j}(2_{(\lambda_{i,j}, \delta_{i,j})})\rightarrow
M_{i,j}(3_{(\lambda_{i,j}, \delta_{i,j})})$ and the only $0-1$ edge from
$VS_{i,j}$ in $N$ is $VS_{i,j}(0_{x_{i,j}})\rightarrow VS_{i,j}(1_{x_{i,j}})$.
By Corollary~\ref{crl:main-exactly-4-red-planar}, $N$ contains exactly one
bottom-red edge incident on $M_{i,j}$. Since this is the only incoming edge into
$VS_{i,j}$ it follows that $x_{i,j}=\lambda_{i,j}$.
\end{proof}

\begin{lemma}
The \gt instance $(\ell,n, \{S_{i,j}\ : i,j\in [\ell]\})$ has a solution.
\label{thm:gt-says-yes}
\end{lemma}
\begin{proof}
By Lemma~\ref{lem:agreement-tight-planar}, it follows that for each
$1\leq i,j\leq \ell$ we have $x_{i,j}=\lambda_{i,j}=x_{i,j+1}$ and
$y_{i,j}=\delta_{i,j}=y_{i+1,j}$ in addition to $(\lambda_{i,j},
\delta_{i,j})\in S_{i,j}$ (by the definition of the main gadget). This implies
that \gt has a solution.
\end{proof}

\subsection{Finishing the proof of Theorem~\ref{thm:scssp-no-epas}}

There is a simple reduction~\cite[Theorem 14.28]{fpt-book} from $\ell$-\textsc{Clique} on $n$ vertex graphs to $(\ell,n)$-\gt. Combining the two directions from Section~\ref{sec:easy-gt} and
Section~\ref{sec:hard-gt} gives a parameterized reduction from $(\ell,n)$-\gt to \scssP on graphs with $(n+\ell)^{O(1)}$vertices and $k=O(\ell^2)=p$. Composing the two reductions, we get a parameterized reduction from $\ell$-\textsc{Clique} on $n$-vertex graphs to \scssP on $(n+\ell)^{O(1)}$ vertex graphs with $k=O(\ell^2)=p$. Hence, the W[1]-hardness of \scssP parameterized by $(k+p)$ follows from the
W[1]-hardness of $\ell$-\textsc{Clique} parameterized by $\ell$. Moreover, Chen et al.~\cite{chen-hardness} showed that, for any function $f$, the existence of an $f(\ell)\cdot n^{o(\ell)}$ algorithm for \textsc{Clique} violates ETH.
Hence, we obtain that, under ETH, there is no $f(k,p)\cdot n^{o(\sqrt{k+p})}$ time algorithm for \scssP.

Suppose now that there is an algorithm $\mathbb{A}$ which runs in time
$f(k,p,\eps)\cdot n^{o(\sqrt{k+p+1/\epsilon})}$ (for some computable function $f$) and computes an $(1+\eps)$-approximate
solution for \scssP. Recall that our reduction works as follows: \gt answers
YES if and only if \scssP has a solution of cost $B^* \leq 43\ell^2 < 44\ell^2$.
Consequently, running $\mathbb{A}$ with $\eps=\frac{1}{44\ell^2}$ implies that
$(1+\eps)\cdot B^{*} <B^{*}+1$. Every edge of our constructed graph $G$
has weight at least $1$, and hence a $(1+\epsilon)$-approximation is in fact
forced to find a solution of cost at most $B^*$, i.e., $\mathbb{A}$ finds
an optimum solution. Since $k=O(\ell^2), p=O(\ell^2)$ and $1/\epsilon= O(\ell^2)$ it follows $f(k,p,\eps)\cdot n^{o(\sqrt{k+p+1/\epsilon})}=g(\ell)\cdot n^{o(\ell)}$ for some computable function $g$. By the previous paragraph, this is not possible under ETH.

~\\
\noindent \textbf{Acknowledgements:} We would like to thank Pasin Manurangsi for helpful discussions.

\bibliography{docsdb,papers}
\bibliographystyle{plainnat}

\appendix


\section{Completeness of the reduction in Lemma~\ref{lem:dirmc-4}}
\label{app:completeness}

We start with a simple observation about the structure of the graph $G'$. While
the $x$-, $y$-, and $z$-paths are bidirected, they --- together with the
$p$-grids --- are arranged in a DAG-like fashion.
That is, there are directed arcs from $x$-paths to $z$-paths, from $z$-paths to
$y$-paths, from $x$-paths to $p$-grids, and from $p$-grids to $y$-paths,
but all cycles in $G'$ are contained in one $x$-, $y$-, or $z$-path.

Consider first the terminal pair $(\sincx, \tincx)$. The out-neighbors of
$\sincx$ are the endpoints $x^{i,j}_0$ for every pair $(i,j)$ with $1 \leq i,j
\leq n$, $i \neq j$;
the only in-neighbors of $\tincx$ are the endpoints $z^i_n$ for every $1 \leq i
\leq n$. Thus, by the previous observation, the only paths from $\sincx$ to
$\tincx$ in the graph $G'$ start by going to some vertex $x^{i,j}_0$, traverse
the $x$-path for the pair $(i,j)$ up to some vertex $x^{i,j}_a$, use the arc
$(x^{i,j}_a,z^i_a)$
to fall to the $z$-path for the color class $i$, and then traverse this $z$-path
to the endpoint $z^i_n$. However, all such paths for $a \geq \alpha(i)$ are cut
by the vertex $\hat{x}^{i,j}_{\alpha(i)} \in X$, while all such paths for $a <
\alpha(i)$ are cut by the vertex $\hat{z}^i_{\alpha(i)} \in X$. Consequently,
the terminal pair $(\sincx, \tincx)$ is separated in $G' \setminus X$.

A similar argument holds for the pair $(\sincy,\tincy)$. By the same reasoning,
the only paths between $\sincy$ and $\tincy$ in the graph $G'$
are paths that start by going to some vertex $z^i_0$, traverse the $z$-path for
the color class $i$ up to some vertex $z^i_a$, use the arc
$(z^i_a,y^{i,j}_a)$ for some $j \neq i$ to fall to the $y$-path for the pair
$(i,j)$, and then continue along this $y$-path to the vertex $y^{i,j}_n$.
However, all such paths for $a \geq \alpha(i)$ are cut by the vertex
$\hat{z}^{i}_{\alpha(i)} \in X$, while all such paths for $a < \alpha(i)$ are
cut by the vertex $\hat{y}^{i,j}_{\alpha(i)} \in X$.

Let us now focus on the terminal pair $(\sdeclt,\tdeclt)$. Observe that there
are two types of paths from $\sdeclt$ to $\tdeclt$ in the graph $G'$.
The first type consists of paths that starts by going to some vertex $x^{i,j}_n$
where $i < j$, traverse the $x$ path for the pair $(i,j)$ up to some vertex
$x^{i,j}_a$,
use the arc $(x^{i,j}_a,z^i_a)$ to fall to the $z$-path for the color class $i$,
then traverse this $z$-path up to some vertex $z^i_b$, use
the arc $(z^i_b,y^{i,j'}_b)$ for some $j' > i$ to fall to the $y$-path for the
pair $(i,j')$, and finally traverse this $y$-path to the endpoint $y^{i,j'}_0$.
However, similarly as in the previous cases, the vertices
$\hat{x}^{i,j}_{\alpha(i)}, \hat{z}^i_{\alpha(i)}, \hat{y}^{i,j'}_{\alpha(i)}
\in X$ cut all such paths.

The second type of paths for $(\sdeclt,\tdeclt)$ use the $p$-grids in the
following manner: the path starts by going to some vertex $x^{i,j}_n$ where $i
<
j$,
traverse the $x$-path for the pair $(i,j)$ up to some vertex $x^{i,j}_a$, use
the arc $(x^{i,j}_a,p^{i,j}_{a,1})$ to fall to the $p$-grid for the pair
$(i,j)$,
traverse this $p$-grid up to a vertex $p^{i,j}_{b,n}$ where $b \geq a$, use the
arc $(p^{i,j}_{b,n}, y^{i,j}_{b-1})$ to fall to the $y$-path for the pair
$(i,j)$,
and traverse this path to the endpoint $y^{i,j}_0$. These paths are cut by $X$
as follows: the paths where $a < \alpha(i)$ are cut by
$\hat{x}^{i,j}_{\alpha(i)} \in X$,
the paths where $b > \alpha(i)$ are cut by $\hat{y}^{i,j}_{\alpha(i)}$, while
the paths where $a=b=\alpha(i)$ are cut by the vertex
$p^{i,j}_{\alpha(i),\alpha(j)} \in X$;
note that the $\alpha(i)$-th row of the grid is the only path from
$p^{i,j}_{a(i),1}$ to $p^{i,j}_{a(i),n}$.
Please observe that the terminal $\tdeclt$ cannot be reached from the $p$-grid
for the pair $(i,j)$ by going to the other $y$-path reachable from this
$p$-grid,
namely the $y$-path for the pair $(j,i)$; the $y$-path for the pair $(j,i)$ has
only outgoing arcs to the terminal $\tdecgt$ since $j > i$.

A similar argument holds for the pair $(\sdecgt,\tdecgt)$. The paths going
through an $x$-path for a pair $(i,j)$, $i > j$,
the $z$-path for the color class $i$, and the $y$-path for a pair $(i,j')$, $i >
j'$, are cut by vertices
$\hat{x}^{i,j}_{\alpha(i)}, \hat{z}^i_{\alpha(i)}, \hat{y}^{i,j'}_{\alpha(i)}
\in X$.
The paths going through an $x$-path for a pair $(i,j)$, $i > j$, the $p$-grid
for the pair $(j,i)$, and the $y$-path for a pair $(i,j)$,
are cut by the vertices $\hat{x}^{i,j}_{\alpha(i)}, \hat{y}^{i,j}_{\alpha(i)},
\hat{p}^{j,i}_{\alpha(j),\alpha(i)}$. Again, it is essential that the other
$y$-path
reachable from the $p$-grid for the pair $(j,i)$, namely the $y$-path for the
pair $(j,i)$, does not have outgoing arcs to the terminal $\tdecgt$, but only
to the terminal $\tdeclt$.

We infer that $X$ is a solution for the instance $(G',\mathcal{T}')$ of
\dirmcfour.

\end{document}